\documentclass[acmsmall,screen,nonacm]{acmart}

\usepackage[english]{babel}
\usepackage{graphicx}
\usepackage{textcomp}
\usepackage{amsthm}
\usepackage{amsfonts}
\usepackage{tikz}
\usetikzlibrary{calc,shapes.multipart,chains,arrows,positioning,shapes,fit}
\usepackage{mathtools}
\usepackage{galois}
\usepackage{stackengine}
\usepackage{stmaryrd}
\usepackage{proof}
\usepackage{scalerel}
\usepackage{xcolor}
\usepackage{rotating}
\usepackage{multirow}
\usepackage{array}
\usepackage{listings}
\usepackage{accsupp}
\usepackage[caption=false]{subfig}
\usepackage{url}
\usepackage{hyperref}
\usepackage{hyphenat}
\usepackage{float}
\usepackage{comment}
\usepackage{caption}  
\usepackage{subcaption}  
\usepackage{algorithm}
\usepackage{algpseudocode}
\usepackage{multicol}
\usepackage{cancel}
\usepackage{enumitem}
\usepackage{float}
\usepackage{multirow}
\usepackage{tabularx}
\usepackage{makecell}
\usepackage{hhline}
\DeclareMathOperator*{\sepconj}{\scalerel*{\ast}{\textstyle\sum}}
\def\eqbydef{\mathrel{\ensurestackMath{\stackon[1pt]{=}{\scriptscriptstyle\triangle}}}}

\newcommand*{\llbrace}{%
  \BeginAccSupp{method=hex,unicode,ActualText=2983}%
    \textnormal{\usefont{OMS}{lmr}{m}{n}\char102}%
    \mathchoice{\mkern-4.05mu}{\mkern-4.05mu}{\mkern-4.3mu}{\mkern-4.8mu}%
    \textnormal{\usefont{OMS}{lmr}{m}{n}\char106}%
  \EndAccSupp{}%
}
\newcommand*{\rrbrace}{%
  \BeginAccSupp{method=hex,unicode,ActualText=2984}%
    \textnormal{\usefont{OMS}{lmr}{m}{n}\char106}%
    \mathchoice{\mkern-4.05mu}{\mkern-4.05mu}{\mkern-4.3mu}{\mkern-4.8mu}%
    \textnormal{\usefont{OMS}{lmr}{m}{n}\char103}%
  \EndAccSupp{}%
}

\newcommand{\unionast}{\mathbin{\ooalign{$\cup$\cr\hfil\raise0.42ex\hbox{$\scriptscriptstyle\ast$}\hfil\cr}}}
\newcommand{\unioncircledast}{\mathbin{\ooalign{$\cup$\cr\hfil\raise0.42ex\hbox{$\scriptscriptstyle\circledast$}\hfil\cr}}}
\newcommand{\emp}{\mathsf{emp}}
\newcommand{\true}{\mathsf{true}}
\newcommand{\false}{\mathsf{false}}
\newcommand{\skipexp}{\mathsf{skip}}
\newcommand{\heapcore}{\hat{M}}
\newcommand{\abort}{\textbf{abort}}
\newcommand{\sepimp}{\mathrel{-\mkern-6mu*}}
\newcommand{\ok}{\textcolor{mygreen}{ok }}
\newcommand{\er}{\textcolor{red}{er }}
\newcommand{\oker}{\textcolor{blue}{\epsilon }}
\newcommand{\valdomain}{\mathbb{V}}
\newcommand{\locdomain}{\mathbb{L}}
\newcommand{\progvar}{\mathbb{X}_\mathsf{p}}
\newcommand{\logicvar}{\mathbb{X}_\mathsf{l}}
\NewDocumentCommand{\slstore}{O{}}{\mathbb{S}_\mathsf{#1}}
\NewDocumentCommand{\islstore}{O{}}{\mathbb{S}_\mathsf{#1}}
\newcommand{\slheap}{\mathbb{H}^{\circled{1}}}
\newcommand{\islheap}{\mathbb{H}^{\circled{2}}}
\newcommand{\heap}{\mathbb{H}^{\circled{i}}}
\newcommand{\slmem}{\mathbb{M}^{\circled{1}}}
\newcommand{\islmem}{\mathbb{M}^{\circled{2}}}
\newcommand{\mem}{\mathbb{M}^{\circled{i}}}
\newcommand{\sast}{\bullet}

\newcommand{\sastconcur}{%
  \mathbin{%
    \ooalign{$\odot$\cr
      \hfil\raise0.15ex\hbox{\scalebox{1.35}{$\scriptscriptstyle\bullet$}}\hfil\cr
    }%
  }%
}

\newcommand{\nnmapsto}{\mathrel{\mapsto\!\!\!\!\!\not\!\not\;\;\;\;}}

\newcommand*\circled[1]{\ensuremath{\underline{\mathtt{#1}}}}

\makeatletter
\newsavebox{\@brx}
\newcommand{\llangle}[1][]{\savebox{\@brx}{\(\m@th{#1\langle}\)}
	\mathopen{\copy\@brx\mkern2mu\kern-0.9\wd\@brx\usebox{\@brx}}}
\newcommand{\rrangle}[1][]{\savebox{\@brx}{\(\m@th{#1\rangle}\)}%
	\mathclose{\copy\@brx\mkern2mu\kern-0.9\wd\@brx\usebox{\@brx}}}
\makeatother

\NewDocumentCommand{\il}{O{black} O{black} m m m}{
    \textcolor{#1}{[#3]} \;#4\; \textcolor{#2}{[#5]}
}
\NewDocumentCommand{\hl}{O{black} O{black} m m m}{
    \textcolor{#1}{\{#3\}} \;#4\; \textcolor{#2}{\{#5\}}
}
\NewDocumentCommand{\sil}{O{black} O{black} m m m}{
    \textcolor{#1}{\llangle#3\rrangle} \;#4\; \textcolor{#2}{\llangle#5\rrangle}
}
\NewDocumentCommand{\nc}{O{black} O{black} m m m}{
    \textcolor{#1}{(#3)} \;#4\; \textcolor{#2}{(#5)}
}

\makeatletter
\newcommand{\stacktwo}[3][\displaystyle]{%
  \mathrel{\vcenter{\offinterlineskip
    \halign{\hfil$#1\m@th##$\hfil\cr
      #2\cr
      #3\cr
    }%
  }}%
}
\makeatother

\newcommand{\rightarrowoverdot}{\stacktwo[\scriptstyle]{\rightarrow}{\cdot}}
\newcommand{\rightarrowbelowdot}{\stacktwo[\scriptstyle]{\cdot}{\rightarrow}}
\newcommand{\leftarrowoverdot}{\stacktwo[\scriptstyle]{\leftarrow}{\cdot}}
\newcommand{\leftarrowbelowdot}{\stacktwo[\scriptstyle]{\cdot}{\leftarrow}}

\newcommand{\forwardarrows}{\stacktwo[\scriptstyle]{\rightarrow}{\rightarrow}}
\newcommand{\backwardarrows}{\stacktwo[\scriptstyle]{\leftarrow}{\leftarrow}}
\newcommand{\leftrightleftrightarrows}{\stacktwo[\scriptstyle]{\leftrightarrow}{\leftrightarrow}}

\NewDocumentCommand{\fwdovertriple}{m m m}{(#1, #2, #3)^{\rightarrowoverdot}}
\NewDocumentCommand{\fwdundertriple}{m m m}{(#1, #2, #3)^{\rightarrowbelowdot}}
\NewDocumentCommand{\bwdovertriple}{m m m}{(#1, #2, #3)^{\leftarrowoverdot}}
\NewDocumentCommand{\bwdundertriple}{m m m}{(#1, #2, #3)^{\leftarrowbelowdot}}

\NewDocumentCommand{\triple}{m m m}{(#1, #2, #3)^{\leftrightleftrightarrows}}
\NewDocumentCommand{\fwdtriple}{m m m}{(#1, #2, #3)^{\forwardarrows}}
\NewDocumentCommand{\bwdtriple}{m m m}{(#1, #2, #3)^{\backwardarrows}}

\NewDocumentCommand{\closure}{m m}{\ensuremath{\rho_{\mathsf{#1}}^{#2}}}

\newcommand{\heapcomp}{\mathrel{\scalebox{0.9}{$\blacktriangleright\!\!\blacktriangleleft$}}}
\newcommand{\logicheapcomp}{\mathrel{\bowtie}}

\newcommand{\reserved}{\boxtimes}

\NewDocumentCommand{\semantics}{m m m}{\llbracket \mathsf{#1} \rrbracket^\mathsf{#2}_{#3}}
\NewDocumentCommand{\proofsystem}{m m}{\vdash_{\mathsf{#1}}^{#2}}

\NewDocumentCommand{\hllabel}{O{\relax} m m}{%
  \ensuremath{\{#1{\mathsf{#2}}\}^{\mathsf{#3}}_{+}}%
}

\NewDocumentCommand{\illabel}{O{\relax} m m}{%
  \ensuremath{[#1{\mathsf{#2}}]^{\mathsf{#3}}_{+}}%
}

\NewDocumentCommand{\sillabel}{O{\relax} m m}{%
  \ensuremath{\llangle#1{\mathsf{#2}}\rrangle^{\mathsf{#3}}_{+}}%
}

\NewDocumentCommand{\nclabel}{O{\relax} m m}{%
  \ensuremath{(#1{\mathsf{#2}})^{\mathsf{#3}}_{+}}%
}

\NewDocumentCommand{\fwdovertripleset}{O{\mathsf{c}}}{
    \mathsf{Triples}_{\llbracket #1 \rrbracket}^{\rightarrowoverdot}
}
\NewDocumentCommand{\fwdundertripleset}{O{\mathsf{c}}}{
    \mathsf{Triples}_{\llbracket #1 \rrbracket}^{\rightarrowbelowdot}
}
\NewDocumentCommand{\bwdovertripleset}{O{\mathsf{c}}}{
    \mathsf{Triples}_{\llbracket #1 \rrbracket}^{\leftarrowoverdot}
}
\NewDocumentCommand{\bwdundertripleset}{O{\mathsf{c}}}{
    \mathsf{Triples}_{\llbracket #1 \rrbracket}^{\leftarrowbelowdot}
}

\NewDocumentCommand{\fwdoveraxiomsset}{O{\mathsf{c}}}{
    \mathsf{Axioms}_{\llbracket #1 \rrbracket}^{\rightarrowoverdot}
}
\NewDocumentCommand{\fwdunderaxiomsset}{O{\mathsf{c}}}{
    \mathsf{Axioms}_{\llbracket #1 \rrbracket}^{\rightarrowbelowdot}
}
\NewDocumentCommand{\bwdoveraxiomsset}{O{\mathsf{c}}}{
    \mathsf{Axioms}_{\llbracket #1 \rrbracket}^{\leftarrowoverdot}
}
\NewDocumentCommand{\bwdunderaxiomsset}{O{\mathsf{c}}}{
    \mathsf{Axioms}_{\llbracket #1 \rrbracket}^{\leftarrowbelowdot}
}

\definecolor{mygreen}{rgb}{0,0.6,0}
\definecolor{myorange}{rgb}{1, 0.392, 0.125}
\definecolor{noveltycolor}{rgb}{0.365, 0.62, 0.863}

\NewDocumentCommand{\revision}{m}{\textcolor{black}{#1}}
\NewDocumentCommand{\revisiontmp}{m}{\textcolor{black}{#1}}

\AtBeginDocument{}
\setcopyright{none}

\newtheorem{theorem}{Theorem}[section]
\newtheorem{property}{Property}[section]
\newtheorem{lemma}{Lemma}[section]
\newtheorem{corollary}[theorem]{Corollary}
\newtheorem{example}{Example}[section]
\newtheorem{definition}{Definition}[section]

\newtheorem{remark}{Remark}

\belowcaptionskip\abovecaptionskip

\begin{document}

\title{Systematic Design of Separation Logics}

\author{Roberto Bruni}
\orcid{0000-0002-7771-4154}
\affiliation{%
    \department{Dipartimento di Informatica}
    \institution{Università di Pisa}
    \city{Pisa}
    \country{Italy}
}
\email{roberto.bruni@unipi.it}

\author{Lorenzo Gazzella}
\orcid{0009-0004-2778-2847}
\affiliation{%
    \department{Dipartimento di Informatica}
    \institution{Università di Pisa}
    \city{Pisa}
    \country{Italy}
}
\email{lorenzo.gazzella@phd.unipi.it}

\author{Roberta Gori}
\orcid{0000-0002-7424-9576}
\affiliation{%
    \department{Dipartimento di Informatica}
    \institution{Università di Pisa}
    \city{Pisa}
    \country{Italy}
}
\email{roberta.gori@unipi.it}


\begin{abstract}
\revisiontmp{Thanks to the locality principle, separation logics support modular, scalable analysis of large codebases by relying on local axioms and frame rules to focus only on the heap fragments required for verification.
However, depending on the direction---forward vs. backward---and sense of approximation---over vs. under---of the analysis, designing the corresponding proof systems can require some ingenuity.
In his work on the \emph{calculational design} of program logics, Patrick Cousot outlines a methodology for deriving proof systems directly from program semantics using abstract interpretation, covering both correctness and incorrectness analyses.
Unfortunately, when applied to heap-manipulating programs, Cousot's calculational approach cannot handle the locality principle, because it does not provide a calculational way to derive frame rules and produces axioms that refer to the global heap.
In this paper, we propose a general methodology for systematically deriving local axioms in which the locality principle is embedded by construction.}
For heap-manipulating primitives, we can derive the minimal required heap and the corresponding pre- and postconditions, complemented by universal frame rules without additional syntactic side conditions.
Our method is parametric w.r.t. a set of semantic closure properties that are exploited to design local axioms; it can deal with different memory models; it favors the reuse of many inference rules across over- and under-approximation; and it produces logical systems capable of deriving a broader range of triples w.r.t. existing, cleverly designed, program logics for (in)correctness, ranging from Separation Logic (SL) and Incorrectness Separation Logic (ISL)
to Separation Sufficient Incorrectness Logic (SepSIL).
Furthermore, we demonstrate the flexibility of our methodology by applying it to design a novel proof system for \revisiontmp{inferring} necessary preconditions with separation logic. 
\end{abstract}


\maketitle

\section{Introduction} \label{sec:introduction}

\paragraph*{Program analysis}
In the development of increasingly complex and critical software systems, software verification plays a fundamental role in ensuring reliability, security, and correctness of program behavior.
Since their conception, frameworks for program analysis, as envisaged by \citet{Turing49-program-proof} and first formalized by \citet{Floyd1967Flowcharts} and \citet{DBLP:journals/cacm/Hoare69} with Hoare Logic (HL), have focused on program correctness, that is,  proving the absence of bugs. 
More recently, motivated by the goal of making formal methods effective in industrial settings, \citet{DBLP:journals/pacmpl/OHearn20} has introduced Incorrectness Logic (IL), which shifts the perspective to identifying the presence of bugs.
Unlike correctness analyses, IL cannot raise false alarms, which can be an advantageous support for programmers.
Whether the analysis is aimed at proving correctness or exposing incorrectness, some form of approximation of the exact program semantics is necessary, because Rice's theorem \cite{rice1953classes} establishes that any non-trivial program property is undecidable.  
Sound methods for proving correctness are based on over-approximation, which reports all true alarms, possibly mixed with false alarms. 
Vice versa, under-approximations are best suited for incorrectness techniques, which only raise true alarms, even if not necessarily all of them. 

\begin{table}[t]
    \centering
    \renewcommand{\arraystretch}{1.3}
	\begin{tabular}{p{0.08\textwidth}|p{0.42\textwidth}|p{0.42\textwidth}|}
    & \textit{forward} & \textit{backward} \\
    \hline
	\textit{over}
    & Partial correctness analysis: it shows absence of errors, but can report false alarms \cite{DBLP:journals/cacm/Hoare69,Floyd1967Flowcharts}. 
    & Necessary condition (NC) analysis: 
    points out some entry states that inevitably lead to errors \cite{DBLP:conf/vmcai/CousotCFL13}.  \\
    baseline
    & Separation logic (SL) \cite{DBLP:conf/csl/OHearnRY01} 
    & Currently none available. \\
    \hline
	\textit{under} 
    & Incorrectness analysis: shows some reachable bugs, without false alarms \cite{DBLP:conf/sefm/VriesK11,DBLP:journals/pacmpl/OHearn20}.
    & Lisbon logic (a.k.a. SIL): exposes some entry states with at least an erroneous execution \cite{DBLP:journals/pacmpl/ZhangK22,DBLP:journals/pacmpl/ZilbersteinDS23,DBLP:journals/pacmpl/ZilbersteinSS24,DBLP:journals/pacmpl/RaadVO24}. \\
    baseline 
    & Incorrectness separation logics (ISL and CISL) \cite{DBLP:conf/cav/RaadBDDOV20,DBLP:journals/pacmpl/RaadBDO22}.
    & Separation SIL \cite{DBLP:journals/pacmpl/AscariBGL25}.  \\
    \hline
	\end{tabular}
	\caption{\revision{Taxonomy of logics with reference frameworks for heap manipulating programs. Columns indicate  the direction, i.e., whether the validity conditions of triples is based on forward or backward semantics, while rows indicate the sense of approximation, i.e., over- or under-approximation.}}
	\Description{Fully described in the text}
	\label{tab:taxonomy}
\end{table}

\revisiontmp{Here, the notion of approximation refers to \emph{forward} program analyses. Indeed, in Hoare triples the postcondition over-approximates all final states reachable from states satisfying the precondition.}
The same idea can also be defined over initial states, as in \citet{DBLP:conf/vmcai/CousotCFL13}, yielding \emph{backward} program analyses.
We will use the term \emph{direction} to refer to the type of analysis (forward or backward), and the term \emph{sense}  to indicate the kind of approximation (over or under).
\revision{Each combination of sense and direction is best suited for achieving specific analysis goals, as summarized in Table~\ref{tab:taxonomy}, triggering a new thread of research in program logics, with efforts to integrate them with separation logic, concurrency, and abstract interpretation, as well as to combine correctness and incorrectness reasoning within unifying frameworks \cite{DBLP:conf/cav/RaadBDDOV20,DBLP:conf/RelMiCS/MollerOH21,DBLP:journals/pacmpl/LeRVBDO22,DBLP:journals/jacm/BruniGGR23,DBLP:conf/ecoop/MaksimovicCLSG23,DBLP:journals/pacmpl/RaadVO24,DBLP:journals/pacmpl/ZilbersteinSS24,DBLP:journals/pacmpl/AscariBGL25,DBLP:journals/pacmpl/VerbeekSFR25,DBLP:journals/pacmpl/VerschtK25,DBLP:conf/ecoop/CarrottAR25}.
Table~\ref{tab:taxonomy} also shows the separation logic instances that serve as reference models
for comparing our proposals. 
Highly sophisticated toolsets grounded in these theories are being developed in both academic and industrial settings, spanning over a wide range of applications \cite{DBLP:journals/cacm/DistefanoFLO19, DBLP:journals/pacmpl/LeRVBDO22, DBLP:journals/jfp/JungKJBBD18,DBLP:journals/pacmpl/VistrupSJ25}.
}


\paragraph*{Cousot's calculational design}
As different logics address different goals, their proof systems are often designed ad hoc, possibly re-using, extending or simply adapting inference rules already available from the literature.
To bring clarity in the program logic zoo and facilitate mutual comparisons for expressiveness and applicability, \citet{DBLP:journals/pacmpl/Cousot24} promoted a semantic-driven methodology, called \emph{calculational design}.
The methodology for deriving a proof system directly from program semantics builds on well-established Abstract Interpretation techniques~\cite{DBLP:conf/popl/CousotC77,cousot2021principles}, and enables the automatic design of sound and complete proof systems through the abstraction associated with the analysis.
Cousot's method begins by specifying the formal semantics of the language, followed by establishing the theory of the logic through Galois connections that abstract suitable semantic properties. The proof system's rules are then automatically derived using fixpoint abstractions, principles of fixpoint induction, and Aczel's approach to constructing deductive rule-based systems from fixpoints \cite{aczel1977introduction}.
Thanks to calculational design, Cousot is able to position all assertional transformational logics within a comprehensive taxonomy \revision{of abstractions~\cite[Figure 3]{DBLP:journals/pacmpl/Cousot24}, which is highly valuable for comparing analyses by varying their underlying Galois connections and for constructing sound and complete proof systems through semantic calculation.}

\paragraph*{Problem statement: local axiom design}
\revision{To make program analysis scalable and effective for large, frequently updated codebases, it is necessary to reuse previous analyses of code fragments, rather than recomputing them each time. In the case of heap-manipulating programs, a successful approach in this respect is provided by logics grounded in the \emph{locality principle}~\cite{DBLP:journals/cacm/DistefanoFLO19}: so-called \emph{`local' axioms} allow the computation of local summaries for individual code fragments, which can then be composed using suitable \emph{frame rules} and bi-abductive techniques \cite{DBLP:conf/popl/CalcagnoDOY09, DBLP:journals/jacm/CalcagnoDOY11}. 
As advised by \citet{DBLP:conf/csl/OHearnRY01}, \emph{to understand how a program works, it should be possible for reasoning and specification to be confined to the cells that the program actually accesses. The value of any other cell will automatically remain unchanged}.
The locality principle enables incremental analysis to reason on larger portions of programs by focusing only on the parts of the code that have changed, rather than requiring a full inspection of the entire program at every update.
Local axioms and frame rules are in fact the key ingredients of separation logics~\cite{DBLP:conf/csl/OHearnRY01, DBLP:conf/lics/Reynolds02, DBLP:conf/cav/RaadBDDOV20, DBLP:journals/entcs/BornatCY06, DBLP:conf/ecoop/MaksimovicCLSG23, DBLP:journals/pacmpl/AscariBGL25, DBLP:conf/lics/CalcagnoOY07}. 
\revisiontmp{Separation logics exploit basic ownership assertions such as $\mathsf{x \mapsto v}$, stating  that the cell at address $\mathsf{x}$ stores the value $\mathsf{v}$, and the separating conjunction $\ast$, which composes assertions over disjoint portions of memory, e.g., the assertion $\mathsf{x \mapsto v} \ast \mathsf{y \mapsto v}$ implies that $\mathsf{x}$ and $\mathsf{y}$ refer to \emph{different} cells of the heap.}
While the program logics taxonomies proposed in~\cite{DBLP:journals/pacmpl/AscariBGL25,DBLP:journals/pacmpl/VerschtK25,DBLP:journals/pacmpl/Cousot24} can serve as a compass to determine the right analysis framework to use for a given goal, the design of their corresponding local axioms and frame rules are not yet supported by Cousot's calculational design. 
This limitation stems from the fact that program semantics is defined to handle arbitrary inputs, and therefore must account for the entire memory. Consequently, Cousot's calculational method produces inference rules that inherently refer to the entire memory state.}

\paragraph*{Main contribution: Systematic design}
\revision{The design, extension, and automation of  separation logics require a certain degree of ingenuity. In this paper, we propose a novel, well-disciplined method to derive local axioms and frame rules directly from program semantics. We exploit semantic closure properties to derive the minimal set of axioms from which all valid triples can be obtained.
Unlike Cousot's calculational design, we do not aim to derive the entire proof system; instead, we focus on the axioms for atomic commands. The inference rules for the remaining language constructs can be borrowed from existing proposals, including Cousot's own.\footnote{More explicitly, we do not address directly the issue of designing rules for sequential composition, conditionals, nondeterministic choice, Kleene star loops, while loops, etc., for which existing solutions are fully reusable.} 
To avoid confusion, we call our method \emph{systematic design}.}

To witness the full generality of our solution, we apply our methodology to a variety of program analyses, addressing different memory models and comparing the \revision{new resulting proof systems} with the ones previously proposed in the literature. 
\revision{Simple examples show that our systematic design improves the expressiveness of existing frameworks in the vast majority of cases.}

 Having fixed the collecting semantics $\llbracket \cdot\rrbracket$ and a set of analysis-preserving closure operators, our approach is first formulated at the purely semantic level and then instantiated to reconstruct assertional program logics based on (in)correctness separation logics formulas~\cite{DBLP:conf/csl/OHearnRY01,DBLP:conf/cav/RaadBDDOV20}, so to exploit the usual separating conjunction operator $\ast$ for local reasoning. 
The distinguishing features of our methodology are:
\emph{\textbf{1)~Principled locality:}} It guarantees the derivation of sound and complete local axioms by construction. Matched with universal frame rules, it enables modular reasoning about program behaviour.
\emph{\textbf{2)~Broad generality:}} It seamlessly supports any combination of over- and under-approximation, as well as forward and backward analyses. Moreover, it accommodates diverse memory management policies, including both reuse and non-reuse of deallocated memory cells.
\emph{\textbf{3)~Reusable core rules:}} It identifies a common core of inference rules that are agnostic to the sense of approximation, supporting their reuse across multiple analyses.
\emph{\textbf{4)~Expressiveness:}} It yields proof systems that derive strictly more valid triples w.r.t. existing analyses, as soundness and completeness are ensured by construction, \revision{rather than by ad hoc and/or ingenious arguments.}
\revision{\emph{\textbf{5)~Mechanization:}} To gain confidence in the proposed methodology, we formalized a proof of concept of the procedure in Isabelle/HOL.}

\paragraph*{\revision{Relevance of results}}
\revisiontmp{To clarify how our method advances the state of the art and to highlight the advantages of a systematic design, we revisit several key issues that our approach addresses in a principled way, largely independent of the sense and direction of analysis. }
\revisiontmp{In contrast, previous ad hoc solutions left room for improvement. Although the literature contains different interpretations of separating assertions (see, e.g., \cite{DBLP:conf/vstte/Parkinson10,DBLP:conf/aplas/CaoCA17}), we follow the line of research in which they impose exact (rather than minimal) requirements on the heap. For example, $\emp$ denotes the empty heap, and $\mathsf{x\mapsto 1}\ast \mathsf{y\mapsto 2}$ denotes a heap containing exactly two cells.}

\paragraph*{Issue \#1: Frame rule soundness across analyses}
\revisiontmp{A key feature of separation logics is the reuse of local code summaries, allowing the analysis to focus on individual program fragments independently from the entire global state. 
To this end, separation logics deal with triples whose pre- and postconditions are confined to the part of the memory accessed by the code. Idle memory portions can be added on demand to scale the reasoning thanks to the frame rule. Frames are possibly computed by the bi-abduction method \cite{DBLP:conf/popl/CalcagnoDOY09, DBLP:journals/jacm/CalcagnoDOY11} to extend and compose local analyses of `adjacent' code fragments that were carried out in isolation. 
As a result, separation logics can scale to large codebases, even when the complete program is unavailable \cite{DBLP:journals/cacm/DistefanoFLO19, DBLP:conf/nfm/CalcagnoDDGHLOP15}. 
Since the notion of soundness is dependent on the sense and direction of the analysis, the same frame rule is not necessarily reusable across different analyses or memory models, unless certain semantic properties are known to hold \cite{DBLP:conf/aplas/CaoCA17} or suitable side conditions are added.}
For example, in Separation Logic (SL) \cite{DBLP:conf/csl/OHearnRY01, DBLP:conf/lics/Reynolds02}---a forward over-approximation logic---the soundness requirement of a triple $\hl{\mathsf{P}}{\mathsf{r}}{\mathsf{Q}}$ is $\llbracket \mathsf{r}\rrbracket \mathsf{P}\subseteq \mathsf{Q}$, namely that the postcondition $\mathsf{Q}$ is satisfied by any final state reachable by executing the program $\mathsf{r}$ on any initial state that satisfies the precondition $\mathsf{P}$.
The classical frame rule is written as follows: 
\begin{equation}
    \vcenter{\hbox{
    \infer{\hl{\mathsf{R\ast P}}{\mathsf{r}}{\mathsf{R\ast Q}}}{\hl{\mathsf{P}}{\mathsf{r}}{\mathsf{Q}} & \mathsf{mod(r)\cap fv(R) = \emptyset}}
    }}
    \label{eq:classicalframe}
\end{equation}
where $\mathsf{mod(r)}$ is the set of variables that appear in the left-hand side of some assignment in $\mathsf{r}$ and $\mathsf{fv(R)}$ the set of free variables in $\mathsf{R}$.
\revisiontmp{The disjointness condition between $\mathsf{fv(R)}$ and $\mathsf{mod(r)}$ guarantees that the memory described by $\mathsf{R}$ is untouched by $\mathsf{r}$ and thus the soundness condition $\llbracket \mathsf{r}\rrbracket (\mathsf{R}\ast \mathsf{P})\subseteq \mathsf{R}\ast \mathsf{Q}$ is trivially satisfied whenever $\llbracket \mathsf{r}\rrbracket \mathsf{P}\subseteq \mathsf{Q}$, because  $\llbracket \mathsf{r}\rrbracket (\mathsf{R}\ast \mathsf{P})\subseteq\mathsf{R}\ast \llbracket \mathsf{r}\rrbracket \mathsf{P}$ and $\mathsf{R}\ast \llbracket \mathsf{r}\rrbracket \mathsf{P}\subseteq \mathsf{R}\ast \mathsf{Q}$ follow from the properties of the collecting semantics $\llbracket \cdot\rrbracket$.}

For a triple $\il{\mathsf{P}}{\mathsf{r}}{\mathsf{Q}}$ in Incorrectness Separation Logic (ISL)~\cite{DBLP:conf/cav/RaadBDDOV20}---a forward under-approximation logic---the soundness requirement is $\llbracket \mathsf{r}\rrbracket \mathsf{P}\supseteq \mathsf{Q}$ and the above reasoning cannot be replicated, unless $\llbracket \mathsf{r}\rrbracket (\mathsf{R}\ast \mathsf{P})=\mathsf{R}\ast \llbracket \mathsf{r}\rrbracket \mathsf{P}$. \revision{Unfortunately this fails in general:} if frame $\mathsf{R}$ shares some memory cell with $\mathsf{P}$ that is deallocated by $\mathsf{r}$, it can happen that $\llbracket \mathsf{r}\rrbracket (\mathsf{R}\ast \mathsf{P}) = \llbracket \mathsf{r}\rrbracket \false = \false$ but $\mathsf{R}\ast \llbracket \mathsf{r}\rrbracket \mathsf{P}\neq \false$.
\revisiontmp{Consequently, \citet{DBLP:conf/cav/RaadBDDOV20} introduced a different memory model in which deallocated cells are explicitly recorded along the computation, denoted with the assertion $\mathsf{x\not\mapsto}$. This makes the heap grow monotonically. Under this model, $\mathsf{R}\ast \mathsf{P}\neq \false$ whenever $\mathsf{R}\ast \llbracket \mathsf{r}\rrbracket \mathsf{P}\neq \false$, and the classical frame rule \eqref{eq:classicalframe} is sound, but their solution requires a different memory model.}

\revisiontmp{
We note that several approaches for avoiding syntactic side conditions have already appeared in the literature. All of them can be traced back to the core insight of \cite{DBLP:journals/entcs/BornatCY06}, although they differ in the semantic assumptions they adopt and in the programming paradigms they target.
\citet{DBLP:journals/entcs/BornatCY06} treat program variables as resources: they introduce the predicate $\mathsf{Own(x)}$ to denote stores that contain exactly the single variable $\mathsf{x}$, and they use the separating conjunction to compose such stores. In this setting, the frame rule can be detached from  syntactic side conditions by adding a semantic requirement that relates the pre- and postcondition; their development, however, is restricted to forward correctness analyses.
A closely related line is pursued in \cite{DBLP:journals/jfp/JungKJBBD18, DBLP:journals/pacmpl/VistrupSJ25}, where the side condition is eliminated entirely,
}
\revisiontmp{
albeit 1) under an affine logic ---where ownership means \emph{"at least this resource"}, rather than \emph{"exactly this resource"}--- 2) without an explicit account of mutable stores and 3) without addressing backward directed rules.
Finally, \citet{DBLP:journals/pacmpl/RaadBDO22} follow a similar principle for incorrectness analysis by modeling both store and heap as a separation algebra, and by adopting explicit deallocation predicates as in  \cite{DBLP:conf/cav/RaadBDDOV20}.
Together, these semantic assumptions suffice to remove side conditions altogether.}

\textit{Advancements on issue \#1.}
We define a `universal' characterization of frames based on the concept of \emph{logical} variables,  a set of inherently immutable variables disjoint from  the program ones.
For each direction of the analysis, we define a \emph{unique frame rule that is independent from the sense of approximation and from the memory model}.
\revisiontmp{In contrast to other approaches, where side-condition-free framing emerges from specific semantic properties,} our method avoids syntactic side conditions on the program without treating variables as resources, 
\revisiontmp{applies directly to settings with mutable stores, and does not require heap monotonicity in order to support incorrectness analysis.}

\paragraph*{Issue \#2: Design of local axioms}

Local axioms are concerned with the minimal amount of information that is necessary to derive the behaviour of each atomic command.
Defining local axioms that are expressive enough to recover all possible valid triples can require some ingenuity.
Moreover, one has to keep in mind that they can be applied jointly with other rules, like the frame rule or the consequence rules, that may depend themselves on the sense and direction of the analysis. In fact, different ad hoc solutions are often chosen in the literature for different logics.

For example, let us take the deallocation command $\mathsf{free(x)}$ that disposes the location pointed to by the variable $\mathsf{x}$ in three different separation logics: the already mentioned SL and ISL, and Separation Sufficient Incorrectness Logic (SepSIL)~\cite{DBLP:journals/pacmpl/AscariBGL25}---a more recent backward under-approximation logic for disclosing the sources of incorrectness, in the spirit of so-called \emph{Lisbon triples} and \emph{manifest error detection}, also studied in \cite{DBLP:journals/pacmpl/ZhangK22,DBLP:journals/pacmpl/ZilbersteinSS24,DBLP:journals/pacmpl/RaadVO24}. SepSIL triples are written $\sil{\mathsf{P}}{\mathsf{r}}{\mathsf{Q}}$ and their soundness requirement guarantees that any initial state that satisfies $\mathsf{P}$ has at least one execution leading to some state in $\mathsf{Q}$. When $\mathsf{Q}$ describes error states, it means that some error is reachable starting from any initial state.
The axioms proposed in the literature for the three different logics are the following.
\[\resizebox{.99\linewidth}{!}{
\infer[\mathsf{(SL)}]{\hl{\mathsf{x\mapsto \_}}{\mathsf{free(x)}}{\mathsf{\emp}}}{}\;
\infer[\mathsf{(ISL)}]{\il{\mathsf{x\mapsto v}}{\mathsf{free(x)}}{\mathsf{x\not\mapsto}}}{}\;
\infer[\mathsf{(SepSIL)}]{\sil{\mathsf{x\mapsto \_}}{\mathsf{free(x)}}{\mathsf{x\not\mapsto}}}{}
}\]

The axiom of SL asserts that the deallocation is possible whenever $\mathsf{x}$ was allocated, no matter the stored value, and then the empty heap is returned because it applies to the memory model without explicit reallocation. Thanks to the consequence rule of SL, it is possible to strengthen the precondition and derive, e.g., $\hl{\mathsf{x\mapsto v}}{\mathsf{free(x)}}{\mathsf{\emp}}$ for any value $\mathsf{v}$.

The axiom of ISL requires a different memory model to guarantee the soundness of the frame rule (as explained above), where the postcondition $\mathsf{x\not\mapsto}$ records the deallocation of $\mathsf{x}$ for future reuse. Moreover, the consequence rule of ISL allows to weaken the precondition, not to strengthen it, because the sense of approximation is reversed w.r.t. SL. Therefore the axiom must be given for any value $\mathsf{v}$, which can then be generalized to derive $\il{\mathsf{x\mapsto \_}}{\mathsf{free(x)}}{\mathsf{x\not\mapsto}}$.

Finally, the axiom of SepSIL has been given for the same memory model as ISL and designed by taking into account the same consequence rule as SL.

Given such ad hoc solutions, one may wonder whether certain analysis are possible only for certain memory models.
Furthermore, we miss a disciplined approach to select the right axioms each time and this makes it complicated to deal with new primitives as well as new logics.

\textit{Advancements on issue \#2.}
Our methodology allows to systematically derive local axioms from the semantics, independently from the sense and direction of the analysis and from the memory model. This brings clarity and automation in the design of novel program logics, which we exemplify in Section~\ref{sec:systematic_backward} by the systematic design of a proof system for backward over-approximation.

\paragraph*{Issue \#3: Derivability and completeness}
\revisiontmp{
While the soundness of a proof system follows directly from the soundness of each of its rules, completeness emerges only from the combined effect of all the rules. }
\revisiontmp{Below, we show that even well-established program logics fail to derive triples whose validity is straightforward. This is where a methodological approach based on closure operators becomes valuable, as it supports reasoning about the proof system globally.
}
\revisiontmp{The presence of a syntactic side condition in the frame rule of SL  may hinder its own application even in cases where the command under analysis does not modify the heap.  }
For example, the following triple can be derived in SL, where the assertion $\mathsf{\emp}$ means the heap is empty:
\[  \hl{\mathsf{\emp}}{\mathsf{if}\; \mathsf{ (\true)} \; \mathsf{then}\;\skipexp\; \mathsf{else}\; \mathsf{x}:=[\mathsf{y}]}{ \mathsf{ \emp}} .\]

\noindent
However, despite the assignment $\mathsf{x}:=[\mathsf{y}]$ being dead code, due to the frame rule's side condition $\mathsf{mod(r)\cap fv(R) = \emptyset}$, we cannot extend this reasoning to fit the frame $\mathsf{R}=\mathsf{x\mapsto \_}$.
\revisiontmp{Hence, the following valid SL triple cannot be derived within this proof system:}
\[  \hl{\mathsf{x}\mapsto \_}{\mathsf{if}\; \mathsf{ (\true)} \; \mathsf{then}\;\skipexp\; \mathsf{else}\; \mathsf{x}:=[\mathsf{y}]}{ \mathsf{x}\mapsto \_} .\]

Notably, the same impossibility argument can be used as a counterexample to the completeness of ISL. There, the problem is even more evident: under-approximation triples $\il{\mathsf{P}}{\mathsf{r}}{\mathsf{Q}}$ are allowed to drop arbitrary path inspections from the analysis, meaning that variables modified in discarded branches of $\mathsf{r}$ should not be taken into account.
Roughly, $\mathsf{mod(r)}$ over-approximates the set of modified variables. As a consequence, 
the side condition may become a source of incompleteness, inhibiting the derivation of valid ISL triples.
This can be substantiated by the valid ISL triple\footnote{\revision{Adhering to~\cite{DBLP:conf/cav/RaadBDDOV20} green assertions are for normal termination, red for erroneous one, and blue for both.}} 
\begin{equation}\label{eq:triple_ex}
    \il[blue][red]{\mathsf{y\mapsto v}}{\mathsf{x:=alloc();\;free(y);\;x:=[y]}}{{\mathsf{y}\not\mapsto} \ast \mathsf{x}\mapsto \_}
\end{equation}
\noindent
whose precondition requires that $\mathsf{y}$ is pointing to the value $\mathsf{v}$, while the postcondition states that $\mathsf{y}$ has been deallocated and that $\mathsf{x}$ is pointing to some value.
When handling the command $\mathsf{x}:=[\mathsf{y}]$ after $\mathsf{y}$ has been deallocated, we end up with an error without any actual modification to the variable $\mathsf{x}$. Despite the fact that $\mathsf{x}$ is not actually modified, we cannot use the frame rule to frame out the assertion on $\mathsf{x}$ due to the side condition $\mathsf{mod(r)\cap fv(R) = \emptyset}$, which prevents the frame from mentioning the variable $\mathsf{x}$, leading to the impossibility of applying the local axiom for $\mathsf{x}:=[\mathsf{y}]$.

\revisiontmp{The next example shows that incompleteness does not stem solely from the frame rule.} Concurrent Incorrectness Separation Logic ($\text{CISL}_{\text{DC}}$) \cite{DBLP:journals/pacmpl/RaadBDO22}, the concurrent version of ISL for detecting memory safety bugs, uses a different underlying memory model that enables a frame rule in the style of \eqref{eq:classicalframe} but without any side condition. 
\revisiontmp{Despite this, $\text{CISL}_{\text{DC}}$ local axioms may still give rise to sources of incompleteness.}
To see this, let us focus on the axiom for allocation:
\[
    \infer{\il[blue][mygreen]{\mathsf{x\mapsto v_1}}{\mathsf{x:=alloc()}}{\mathsf{\exists l. x\mapsto l \ast l\mapsto v}}}{}
\]
whose precondition assumes the variable $\mathsf{x}$ points to the value $\mathsf{v_1}$, while the postcondition states that $\mathsf{x}$ points to some location $\mathsf{l}$, which in turn points to $\mathsf{v}$. We note that since $\text{CISL}_{\text{DC}}$ is an under-approximation logic, we would like to have the possibility of expressing the strongest possible postcondition, which would carry a generic choice of value $\mathsf{v}$. Unfortunately, the valid triple
\[
\il[blue][mygreen]{\mathsf{x\mapsto v_1}}{\mathsf{x:=alloc()}}{\mathsf{\exists l. \exists v. x\mapsto l \ast l\mapsto v}}
\]
cannot be derived in $\text{CISL}_{\text{DC}}$ since we would need an infinite disjunction rule to collect every possible value $\mathsf{v}$ or the `exists' rule, but neither is present in the logic system. 

\textit{Advancements on issue \#3.}
\revisiontmp{In contrast to approaches where these interactions are handled implicitly at the semantic level (e.g., Iris-style or CISL-style systems),}
our approach systematically identifies the necessary side conditions for applying frame rules in both forward and backward analyses, thereby extending the range of derivable triples well beyond existing proposals.
Moreover, the derivability of all valid triples for any atomic command $\mathsf{c}$ is ensured by construction.

\paragraph*{Technical contribution}
Our approach refines and extends the semantic methodology outlined in the introduction by systematically incorporating three key parameters of program analysis: 
\emph{direction} (forward/backward analysis), \emph{sense} (over-/under-approximation), and the enforcement of the \emph{locality principle} through \emph{semantic}, rather than syntactic, conditions.

From a methodological standpoint, we characterize the set of valid triples for a given atomic command $\mathsf{c}$ by identifying semantic \emph{closure properties} that are preserved under the intended analysis. 
These include:
(i)~consequence closure (e.g., weakening postconditions in over-approximation),
(ii)~frame closure (capturing heap locality), and
(iii)~existential closure (generalization over variables).
These closure properties play a crucial role in determining the simplest triples from which all others can be derived using closure operators. Indeed, if a triple can be obtained through a closure operator, including it explicitly in the set of axioms would be redundant. 
We show that some of these closures commute and that each valid triple can, therefore, be derived via a short canonical sequence of such closures in a given order, allowing axioms to be selected by minimality considerations. 
This yields a systematic and unified proof system design that operates on pre- and postconditions that are concrete sets of states. 
\revision{The translation from the semantic proof system to program logics is then achieved by selecting appropriate assertion languages, building on separation logic constructs.}

\revision{Technically speaking, the ingredients that make systematically designed program logics more expressive, modular, and effective than existing ones are immutable variables and semantic closures.}

\revision{Immutable variables are not a new concept; on the contrary, they have been widely employed in many proposals before \cite{DBLP:journals/pacmpl/DardinierM24,Floyd1967Flowcharts, DBLP:conf/cav/RaadBDDOV20, DBLP:conf/csl/OHearnRY01, DBLP:conf/lics/Reynolds02, DBLP:journals/pacmpl/RaadVO24}, even if with different names (e.g., auxiliary, logical or shadow variables) and not necessarily injected into the memory model at the semantic level, as we do here. 
They can be used to freeze and retain the values of program variables across transitions, thereby characterizing the class of admissible frames. 
Thanks to logical variables, we define frame compatibility directly at the semantic level, thus avoiding syntactic side conditions that are a well-known source of incompleteness. 
Their principled treatment allows us to distill universal frame rules.}

\revision{Closure operators are instrumental to obtain local axioms that are independent of the memory model and the sense of approximation. This supports a uniform treatment across multiple program logics, including both correctness and incorrectness oriented analyses.}

\begin{table}[t]
    \centering
    \renewcommand{\arraystretch}{1.5}
    \begin{tabular}{c|c||c|p{0.35\textwidth}|p{0.35\textwidth}|}
    \multicolumn{2}{c||}{} & \textbf{sense} 
    & \textbf{no memory reallocation}  ($\slmem$)
    & \textbf{memory reallocation} ($\islmem$) \\
    \hline
    \hline
    \multirow{2}{*}{%
    \raisebox{0pt}[0pt][0pt]{%
    \rotatebox[origin=c]{90}{\textit{forward}\hspace{25pt}}%
    }} &
    SL+ & \textit{over} 
    & $\vdash_{\mathsf{SL+}}^{\circled{1}}$ is \textcolor{noveltycolor}{more expressive} than SL (see Th.~\ref{th:sl_completeness} and Ex.~\ref{ex:sl+vssl})
    & \textcolor{noveltycolor}{New logic:} 
      $\vdash_{\mathsf{SL+}}^{\circled{2}}$ \\
    \cline{2-5} &
    ISL+ & \textit{under}
    & \textcolor{noveltycolor}{New logic:} 
      $\vdash_{\mathsf{ISL+}}^{\circled{1}}$
    & $\vdash_{\mathsf{ISL+}}^{\circled{2}}$ is \textcolor{noveltycolor}{more expressive} than ISL (see Th.~\ref{th:relative_completeness_loopfree_isl}--\ref{th:statewise_completeness_isl} and Ex.~\ref{ex:isl+vsisl}) \\
    \hline
    \hline
    \multirow{2}{*}{%
    \raisebox{0pt}[0pt][0pt]{%
    \rotatebox[origin=c]{90}{\textit{backward}\hspace{25pt}}%
    }} &
    SIL+ & \textit{under} 
    & \textcolor{noveltycolor}{New logic:} 
      $\vdash_{\mathsf{SIL+}}^{\circled{1}}$
    & Unlike SepSIL, \textcolor{noveltycolor}{all axioms of $\vdash_{\mathsf{SIL+}}^{\circled{2}}$ are local} \\
    \cline{2-5} &
    NC+  & \textit{over}
    & \textcolor{noveltycolor}{New logic:} 
      $\vdash_{\mathsf{NC+}}^{\circled{1}}$
    &  \textcolor{noveltycolor}{New logic:} 
      $\vdash_{\mathsf{NC+}}^{\circled{2}}$ \\[8pt]
    \hline
    \end{tabular}
    \caption{\revision{Systematically designed logics and their improvements w.r.t. existing ones.
    Logics are organized according to the direction of the analysis, sense of approximation and underlying memory model.
    We write $\vdash_\mathsf{L}^\text{\circled{i}}$ to denote the logic $\mathsf{L}$ under the memory model $\mem$.
    All original contributions are highlighted in \textcolor{noveltycolor}{light blue}.}}
    \label{tab:summary+}
\end{table}

\revision{
Overall, our methodology produces proof systems that are sound-by-construction.
All synthesized proof systems and their novelties are outlined in Table~\ref{tab:summary+}. 
Besides increasing the expressiveness of existing logics,
new logics can be designed with minimal effort.
This has been done to address different memory models (see entries 
$\vdash_{\mathsf{SL+}}^{\circled{2}}$,
$\vdash_{\mathsf{ISL+}}^{\circled{1}}$, and
$\vdash_{\mathsf{SIL+}}^{\circled{1}}$), but also
for distilling the proof system for Necessary Conditions (NC) ~\cite{DBLP:conf/vmcai/CousotCFL13} (see entries $\vdash_{\mathsf{NC+}}^{\circled{1}}$ and $\vdash_{\mathsf{NC+}}^{\circled{2}}$), to be used in automatic precondition inference.
As a further instance, not reported in Table~\ref{tab:summary+}, we apply our methodology to the setting of $\text{CISL}_{\text{DC}}$ \cite{DBLP:journals/pacmpl/RaadBDO22}, where a different store model is required.}

\paragraph*{Structure of the paper}

After recalling some necessary background in Section~\ref{sec:background}, we introduce our semantic framework and the methodology for systematic design in Sections~\ref{sec:new_semantic_model}--\ref{sec:systematic_program_logics}. 
In Sections \ref{sec:systematic_forward}–\ref{sec:systematic_backward}, we apply our methodology to both forward and backward analyses, addressing under- and over-approximations. 
While related work is thoroughly discussed throughout the paper, we conclude in Section~\ref{sec:conclusion} with future work. 
Auxiliary material is in the Appendix for reviewers' convenience.

\section{Background} \label{sec:background}
\revision{Following \citet{DBLP:journals/pacmpl/OHearn20}, we consider a language of regular commands, parametric to suitable atomic commands $\mathsf{c}$, including heap-manipulating primitives, according to the grammar below:}
\[
    \mathsf{Reg} \ni \mathsf{r} ::= \; \mathsf{c} \;|\; \mathsf{r;r} \;|\; \mathsf{r + r} \;|\; \mathsf{r^*} 
\]

We write $\mathsf{r;r}$ for sequential composition, $\mathsf{r+r}$ for nondeterministic choice, and $\mathsf{r^\ast}$ for the Kleene iteration of $\mathsf{r}$, where the command $\mathsf{r}$ can be executed any finite number of times, possibly zero. 
\revision{
Regular commands are a well-established framework \cite{DBLP:conf/focs/Pratt76, DBLP:journals/jcss/FischerL79, dynamic_logic/10.5555/557365} that serves as a general template, allowing for different instantiations depending on the chosen set of atomic commands, to be detailed later (see~\eqref{eq:atom}).
Standard imperative constructs, like conditionals and loops, can be coded following, e.g., \cite{DBLP:books/daglib/0070910}.}


\subsection{Memory models} \label{sec:semantic_model}
As briefly discussed in Section \ref{sec:introduction}, SL \cite{DBLP:conf/csl/OHearnRY01, DBLP:conf/lics/Reynolds02} and ISL \cite{DBLP:conf/cav/RaadBDDOV20} consider two different memory models: the ISL memory model extends SL by allowing for tracking of deallocated locations and their subsequent reallocation. To distinguish the two models, throughout the paper we will use $\circled{1}$ to refer to the SL model and $\circled{2}$ to refer to the ISL model, using $\circled{i}$ whenever we do not care whether it is the model $\circled{1}$ or $\circled{2}$. In particular, we define the set of stores $\slstore[p]$ as the set of total functions from program variables $\progvar$ to values $\valdomain$. The set of heaps $\slheap$ in the memory model $\circled{1}$ is the set of partial functions from locations $\locdomain\subseteq \valdomain$ to values $\valdomain$. Finally, the set of memories $\slmem$ is the product domain between $\slstore[p]$ and $\slheap$. The model $\circled{2}$ extends the heap domain $\slheap$ with a new special value $\bot$, which tracks for deallocated locations, i.e., given an heap $h$ in $\circled{2}$ if $h(l)=\bot$, then the location $l$ was previously allocated and then deallocated. Formally,
\begin{align*}
    \circled{1}: \qquad &\slstore[p] \triangleq \progvar \rightarrow \valdomain &&
    \slheap\triangleq \locdomain\rightharpoonup \valdomain &
    \slmem\triangleq \slstore[p] \times \slheap \ ;\\
    \circled{2}: \qquad & \islstore[p] \triangleq \progvar \rightarrow \valdomain &&
    \islheap\triangleq \locdomain\rightharpoonup (\valdomain\cup \{\bot\}) &
    \islmem\triangleq \islstore[p]\times \islheap \ .
\end{align*}

We write $h_p \# h_q$ to indicate that $h_p$ and $h_q$ are heaps with disjoint domains (i.e. $dom(h_p) \cap dom(h_q) = \emptyset$), while $h_p \bullet h_q$ indicates the join of $h_p$ and $h_q$ and it is defined if and only if $h_p \# h_q$. 
\revision{With a slight abuse of notation, we lift the operator $\bullet$ to compose sets of memories, defined as:
\[
P\sast Q \triangleq \{(s,h_p\bullet h_q)\;|\;(s,h_p)\in P \wedge (s,h_q)\in Q \wedge h_p\#h_q\}
\]
Moreover, we find it convenient to define existential quantification over sets as follows:
\[
    \exists \mathsf{X}.P \triangleq \{(s',h)\;|\;(s,h)\in P \wedge s'_{|_{\mathbb{X}\setminus \mathsf{X}}}=s_{|_{\mathbb{X}\setminus \mathsf{X}}}\}.
\]
where the function $f_{|_X}:X\rightarrow B$ is the restriction of $f:A\rightarrow B$ to the domain $X\subseteq A$.
}

\subsection{Assertion language}\label{sec:al}
We rely on the assertion language of \citet{DBLP:conf/cav/RaadBDDOV20}, which extends the classical definition \cite{DBLP:conf/csl/OHearnRY01, DBLP:conf/lics/Reynolds02} by introducing the negative heap assertion $\mathsf{x\not\mapsto}$. Since the assertion $\mathsf{x\not\mapsto}$ is used to track deallocated locations, it will only be used in the memory model \circled{2}. We let $\mathsf{Ast}$ be defined as follows, where $\asymp \in \{=, \leq, ...\}$ represents standard comparison operators, $\mathsf{X}$ is a set of variables and $\mathsf{e,e_1,e_2}$ are arithmetic expressions:
\begin{equation} \label{eq:ast}
    \mathsf{Ast \ni P,Q ::=\; \textcolor{gray}{\false \;|\; \true \;|\; e_1 \asymp e_2 \;|\; P\wedge Q \;|\; P\vee Q \;|\; \exists X. P \;|} \;\emp \;|\; x \mapsto e \;|\; x\not\mapsto \;|\; P \ast Q}
\end{equation}

The standard operators from first-order logic, here in light gray, define the so-called \emph{pure} assertions.
The remaining operators define the \emph{structural} assertions. The assertion $\mathsf{\emp}$ denotes the empty heap, the ownership predicate $\mathsf{x\mapsto e}$ denotes memories in which the heap contains a single location pointed by $\mathsf{x}$ and whose content is $\mathsf{e}$, while $\mathsf{x\not\mapsto}$ stands for memories in which the heap has a single location pointed by $\mathsf{x}$, but that has been deallocated. The separating conjunction $\mathsf{P \ast Q}$ represents memories where the heap (not the store) can be split into two disjoint sub-heaps, one satisfying $\mathsf{P}$ and the other satisfying $\mathsf{Q}$. 
We let $\mathsf{fv(P)}$ be the set of free variables in $\mathsf{P}$.
The intuitive meaning of structural assertions is summarized in Table \ref{tab:ast}. 
To simplify the exposition, and with a slight abuse of notation, we introduce a single semantics for the assertion language. The reference model will be evident from the context; for example, the semantics of $\true$ is either  $\slmem$ or $\islmem$ depending on the memory model we are dealing with, and $\mathsf{x\not\mapsto}$ will be used only when the memory model $\circled{2}$ is considered. Hence, the set of all and only memories satisfying the assertion $\mathsf{P\in Ast}$ is denoted by $\llbrace \mathsf{P}\rrbrace$ and defined as follows:
\[
\begin{gathered}
\begin{aligned}
\llbrace \mathsf{\false}\rrbrace &\triangleq \emptyset \qquad
\llbrace \mathsf{\true}\rrbrace \triangleq \mem \qquad
\llbrace \mathsf{e_1 \asymp e_2}\rrbrace \triangleq \{(s,h)\mid \llparenthesis \mathsf{e_1}\rrparenthesis s \asymp \llparenthesis \mathsf{e_2}\rrparenthesis s\} \qquad
\llbrace \mathsf{\exists X.P}\rrbrace \triangleq \revision{\exists X.\,\llbrace \mathsf{P}\rrbrace}
\end{aligned}
\\
\begin{aligned}
\llbrace \mathsf{x \mapsto e}\rrbrace &\triangleq \{(s,[s(\mathsf{x}) \mapsto \llparenthesis \mathsf{e}\rrparenthesis s])\} &
\llbrace \mathsf{x \not\mapsto}\rrbrace &\triangleq \{(s,[s(\mathsf{x})\mapsto \bot])\} &
\llbrace \mathsf{\emp}\rrbrace &\triangleq \{(s,[])\} \\
\llbrace \mathsf{P \wedge Q}\rrbrace &\triangleq \llbrace \mathsf{P}\rrbrace \cap \llbrace \mathsf{Q}\rrbrace &
\llbrace \mathsf{P \vee Q}\rrbrace &\triangleq \llbrace \mathsf{P}\rrbrace \cup \llbrace \mathsf{Q}\rrbrace &
\llbrace \mathsf{P \ast Q}\rrbrace &\triangleq \revision{\llbrace \mathsf{P}\rrbrace \sast \llbrace \mathsf{Q}\rrbrace}
\end{aligned}
\end{gathered}
\]

\noindent 
where the semantics for arithmetic and Boolean expressions $\llparenthesis \cdot \rrparenthesis$ are defined as usual. 
We note that a generic assertion in $\mathsf{Ast}$ can be rewritten as a (possibly infinite) disjunction of pure assertions and structural assertions.\footnote{This fact will be exploited in completeness proofs and in Definition \ref{def:logic_heap_comp}.} 
We use the usual shortcuts $\mathsf{x\mapsto \_}$ for the ownership predicate when we disregard the content of the location and $\mathsf{x\doteq y}$ for expressing equality constraints over a memory with empty heap:
$\mathsf{x\mapsto \_} \triangleq \mathsf{\exists y. x\mapsto y}$, and $\mathsf{x\doteq y} \triangleq \mathsf{x=y\wedge \emp}$.

As a matter of notation, we use the font $\mathsf{P}, \mathsf{Q},\mathsf{R}, ... \in \mathsf{Ast}$ when dealing with logical assertions, and $P,Q,R, ...\in \wp(\mem)$ when dealing with concrete sets of memories. 

\section{Logical Variables and Universal Frames} \label{sec:new_semantic_model}
The scalability of methods based on separation logics builds on the notion of local axioms and framing. 
Local axioms should account for the least amount of information necessary to execute each atomic command.
Then, a suitable frame rule allows for extending both the pre- and postconditions by adding idle portions of the heap, whenever needed. Here, the frame rule must be designed with particular care in order to preserve the validity of the triples, according to the sense and direction of the analysis.
The exact characterization of the valid frames applicable to a given triple is a long-standing problem \cite{DBLP:conf/fossacs/YangO02, DBLP:conf/fossacs/RazaG08, DBLP:conf/lics/CalcagnoOY07}. As discussed in Section \ref{sec:introduction}, the most common side condition for the frame rule is $\mathsf{mod(r)\cap fv(R) = \emptyset}$, where $\mathsf{mod(r)}$ is the set of all variables that appear on the left hand side of some assignment in $\mathsf{r}$. 
While it prevents the frame from constraining any potentially modified program variable, this can be unnecessarily restrictive in some cases (e.g., assignments in dead code or in dropped disjuncts).

\paragraph{Logical variables}
We address and solve the above problems by introducing a `universal' characterization of frames based on the concept of \emph{logical} variables. To distinguish them from program variables, they are annotated with a prime symbol ($\mathsf{'}$).
An early example about the use of logical variables is the well-known Floyd's axiom for assignment \cite{Floyd1967Flowcharts}:
\begin{equation}
    \infer{\hl{\mathsf{P}}{\mathsf{x:=e}}{\mathsf{\exists x'. P[x'/x] \wedge x=e[x'/x]}}}{}
    \label{eq:floyd}
\end{equation}
where $\mathsf{x}'$ represents the value of the program variable $\mathsf{x}$ before the assignment. 
Over the years, various forms of auxiliary variables have been employed for similar purposes, albeit under different names, such as in \cite{Floyd1967Flowcharts,DBLP:journals/fac/AptO19, DBLP:conf/fm/OheimbN02, DBLP:books/daglib/0067018, DBLP:journals/fac/Kleymann99, DBLP:journals/pacmpl/DardinierM24}. 
In particular, they have been often used to freeze the initial values of certain variables in the precondition and link them to their updated values in the postcondition, upon termination. 
Since logical variables cannot be altered, they are inherently immutable. 
As an example, the meaning of the triple
\[\hl{\mathsf{x=x'\wedge y=y'}}{\mathsf{r}}{\mathsf{x=x' \wedge y=y' \wedge z=x'+y'}}\]
\noindent
is twofold: it says not only that whenever the command $\mathsf{r}$ terminates, the variable $\mathsf{z}$ stores the sum of $\mathsf{x}$ and $\mathsf{y}$, but also that the variables $\mathsf{x}$ and $\mathsf{y}$ remain unchanged \cite[\S10.2]{DBLP:journals/fac/AptO19}.
Thereby, \emph{if a program variable is associated with the same logical variable in the pre- and postcondition, then the program has not changed the value of that variable}.

Formally, we extend the memory models introduced in Section \ref{sec:background} by introducing a separate set $\logicvar$ of logical variables that cannot appear in any command and require the free variables of the frame to be logical ones. This way the side condition $\mathsf{mod(r)\cap fv(R) = \emptyset}$ is satisfied by definition. 

\paragraph{Extended memory models}
In order to guarantee the availability of a fresh symbol to denote a certain value, in the logical store, we require for each value, the existence of a denumerable family of logical variables with that value. This can be formalized by defining the logical store as follows:
\[
\slstore[l] \triangleq \logicvar \rightarrow \valdomain 
\qquad\text{where} \qquad
\forall s\in \slstore[l].\forall v\in \valdomain. \exists \{\mathsf{x_i'}\}_{\mathsf{i}\in\mathbb{N}}\subseteq \logicvar. s(\mathsf{x_i'})=v 
\]

The set of stores will be $\mathbb{S} \triangleq \slstore[p] \times \slstore[l]$, both for memory models \circled{1} and \circled{2}. As a matter of notation, we denote a generic pair $\langle s_\mathsf{p}, s_\mathsf{l}\rangle$ with $s$ and write just $s(x)$ and $s(x')$ in place of $s_\mathsf{p}(x)$ and $s_\mathsf{l}(x')$, respectively. 
The introduction of logical variables is just a convenient mathematical construction to support program analysis, as they should not be part of the machine model for executing the code. 
Of course, any store $s_\mathsf{p}:\mathbb{X}_\mathsf{p}\rightarrow \mathbb{V}$ can be seen as the representative of the set of stores $\{s_\mathsf{p}\} \times \slstore[l]$.

\paragraph{Universal frames}
In the following, we define a universal frame as any set of memories $\slstore[p] \times R \times \heap$, with $R\subseteq \slstore[l]$, i.e., it can only constrain logical variables and not program ones. 

\begin{definition}[Universal frame]\label{def:universalframe}
    A set of memories $R\subseteq \mem$ (with $\circled{i}\in \{\circled{1}, \circled{2}\}$) such that:
    \[(s,h)\in R \Rightarrow \forall \mathsf{x}\in \progvar. \forall v\in \mathbb{V}. (s[\mathsf{x}\mapsto v], h)\in R\]
    is called a \emph{universal frame}.
    We name $\mathbb{F}$ the set of all universal frames.
\end{definition}

In the assertion language, a universal frame is any assertion $\mathsf{R}$ such that $\mathsf{fv(R)}\subseteq \mathbb{X}_\mathsf{l}$. Since $\mathsf{mod(r)}\subseteq \mathbb{X}_\mathsf{p}$, the condition $\mathsf{mod(r)\cap fv(R) = \emptyset}$ is trivially satisfied by universal frames.

\section{\revision{Systematic Design Methodology}} \label{sec:systematic_program_logics}

In this section, we show how to pair the locality principle with the seminal work of \citet{DBLP:journals/pacmpl/Cousot24} on the calculational design of program logics through the means of abstract interpretation \cite{DBLP:conf/popl/CousotC77, cousot2021principles}. While Cousot's framework elegantly recovers a broad spectrum of program logics \cite[Figure 3]{DBLP:journals/pacmpl/Cousot24}, it does not address the direct calculation of local axioms or the frame rule from the semantics. 
Here we exploit semantic closure properties to define a preorder on valid triples and a notion of minimality that can be exploited to derive local axioms that, through finite applications of closures, can generate all valid triples.
Notably, by appropriately choosing closure operators, the same approach can be applied to a variety of program logics. 
Axioms devised with our methodology can be paired with other approaches that account for the inference rules required to handle all the other program composition operators, e.g, \citet{DBLP:journals/pacmpl/Cousot24} framework.

In the first part of this section, we present the general methodology that entirely depends on the semantic function and the identification of available closures. 
In the second part, we showcase the systematic derivation of semantic proof systems for separation logics. \revision{Both developments have been mechanized in Isabelle/HOL, providing a proof of concept for forward analysis.}

\subsection{General methodology}

A closure is a monotone function that is extensive and idempotent. 
We focus on closure properties of program logics, such that, given a valid predicate $p$, any predicate obtained by applying the closure operators to $p$ is also valid. 
Typically, these closures correspond to inference rules of the logic. 
As an example, for the theory of HL triples \cite{DBLP:journals/cacm/Hoare69}, one possible closure operator is:
\[\closure{cons}{}\;T \eqbydef \{(P,\mathsf{c},Q')\;|\;(P,\mathsf{c},Q)\in T \wedge Q\subseteq Q'\}\]

This means that if $Q$ over-approximates the semantics of the command $\mathsf{c}$ on input $P$, we can replace $Q$ with any larger $Q'$ while preserving the validity of the triple. 
This closure operator is effectively encapsulated within the proof system of HL through the consequence rule:
\begin{equation}
    \vcenter{\hbox{
    \infer
        {\hl{\mathsf{P}}{\mathsf{c}}{\mathsf{Q}'}} 
        {\hl{\mathsf{P}}{\mathsf{c}}{\mathsf{Q}} & \mathsf{Q} \Rightarrow \mathsf{Q}'}
        \label{eq:samplecons}
    }}
\end{equation}

Each closure $\rho$ induces a preorder, in the sense that $p$ precedes or is equal to $p'$ if $p'\in \rho(\{p\})$. Therefore, a minimal set of (valid) predicates $T$ such that $\rho(T)$ contains all and only valid predicates is a suitable candidate to form the axioms of the logic.
Although our discussion centers on closure operators for program triples related to atomic commands, the same reasoning could naturally be extended to the entire language $\mathsf{Reg}$. 
However, exploring this broader extension lies beyond the scope of our current focus.
The general methodology consists of the following steps:
\begin{enumerate}
    \item 
    Fix the semantic property of interest (e.g., that the postcondition is an over-approximation of the exact semantics) and let $\mathsf{Triples}_{\llbracket \mathsf{c} \rrbracket}$ denote the set of valid semantic triples for $\mathsf{c}$.
    \item 
    Select a finite set of (property preserving) closure operators $\rho_1,...,\rho_n$ (e.g., the consequence rule~\eqref{eq:samplecons} or a frame rule), such that $\rho_i(\mathsf{Triples}_{\llbracket \mathsf{c} \rrbracket}) = \mathsf{Triples}_{\llbracket \mathsf{c} \rrbracket}$ for any closure $\rho_i$.
    \item 
    Let $\rho \triangleq (\bigcup_i \rho_i)^\ast$ represent any finite sequence of applications of $\rho_1,...,\rho_n$. 
    It can be shown that $\rho$ is still a closure operator and that $\rho(\mathsf{Triples}_{\llbracket \mathsf{c} \rrbracket}) = \mathsf{Triples}_{\llbracket \mathsf{c} \rrbracket}$ holds. Moreover, under suitable circumstances, extensively  discussed in \cite{DBLP:books/hal/Cousot78, Ore1943}, the closure $\rho$ can be further simplified to a single sequence of applications $\rho_{\pi(1)}\circ\cdots \circ\rho_{\pi(n)}$ of the original closures, for some permutation $\pi$ of the indexes $1,...,n$. We achieve this result by proving that certain pairwise swaps $\rho_{\pi(j)}\circ \rho_{\pi(i)} \subseteq \rho_{\pi(i)}\circ \rho_{\pi(j)}$ of  closures are possible when $i\leq j$.
    \item Select a set
    $\mathsf{Axioms}_{\llbracket \mathsf{c} \rrbracket} \subseteq \mathsf{Triples}_{\llbracket \mathsf{c} \rrbracket}$
such that
$\rho(\mathsf{Axioms}_{\llbracket \mathsf{c} \rrbracket}) = \mathsf{Triples}_{\llbracket \mathsf{c} \rrbracket}$,
\revision{and which is minimal in the sense that no triple in
$\mathsf{Axioms}_{\llbracket \mathsf{c} \rrbracket}$ is derivable from the remaining ones via the closure operators}. We exploit the characterization $\rho = \rho_{\pi(1)} \circ \cdots \circ \rho_{\pi(n)}$ to prove that any valid semantic triple can be obtained from the axioms by applying the closures in a fixed order.
    \revision{The axioms in the set $\mathsf{Axioms}_{\llbracket \mathsf{c} \rrbracket} \subseteq \mathsf{Triples}_{\llbracket \mathsf{c} \rrbracket}$ are exactly the axioms of our semantic proof system, and each closure $\rho_1,...,\rho_n$ yields a corresponding inference rule.}
    \item \revision{The semantic rules are then transformed to a logical proof system by selecting a suitable assertion language for expressing pre- and postconditions. In particular, for each command $\mathsf{c}$, a minimal set of logical axioms is chosen that is able to cover the triples in $\mathsf{Axioms}_{\llbracket \mathsf{c} \rrbracket}$.}
\end{enumerate}

\revision{It is immediate to verify that the above methodology leads us to derive sound and (relative) complete proof systems for the analysis of interest.}

\subsection{\revision{Systematic design of separation logics}}\label{sec:systematic_separation}
\revision{In this section, we apply our methodology to the synthesis of separation logics. We let $\mathsf{c}$ be an atomic command and denote by $\llbracket \mathsf{c} \rrbracket$ its forward semantics, which may depend, e.g., on the considered memory model.
For either directions of the analysis and senses of approximation, the set of semantics triples is defined in Table~\ref{tab:triples}, where we write $\llbracket \overleftarrow{\mathsf{c}}\rrbracket$ to refer to the backward semantics of the command $\mathsf{c}$, defined as the opposite relation of the forward semantics, i.e., $\llbracket \overleftarrow{\mathsf{r}}\rrbracket \sigma'\triangleq \{\sigma\;|\;\sigma'\in \llbracket \mathsf{r}\rrbracket\}$.}
Since multiple notions of validity are involved, superscripts indicate both the analysis direction and the sense of approximation. Arrows specify direction—right for forward, left for backward semantics—while their position relative to the dot denotes approximation: above for over-, below for under-approximation.
For example, $\rightarrowoverdot$ represents an over-approximation in forward semantics.
\revision{Step (1) consists of fixing the set of triples of interest among those in Table~\ref{tab:triples}.}

\begin{table}[t]
    \centering
    \renewcommand{\arraystretch}{1.3}
    \begin{tabular}{|c|c|c|}
         \hline
         & Forward & Backward \\
         \hline  \rule{0pt}{1em}
         Over & $\fwdovertripleset\triangleq\{\fwdovertriple{P}{\mathsf{c}}{Q} \;|\; \llbracket \mathsf{c}\rrbracket P\subseteq Q\}$ & $\bwdovertripleset\triangleq\{\bwdovertriple{P}{\mathsf{c}}{Q} \;|\; \llbracket \overleftarrow{\mathsf{c}}\rrbracket Q\subseteq P\}$\\[10pt]
         Under & $\fwdundertripleset\triangleq\{\fwdundertriple{P}{\mathsf{c}}{Q} \;|\; Q \subseteq \llbracket \mathsf{c}\rrbracket P\}$ & $\bwdundertripleset\triangleq\{\bwdundertriple{P}{\mathsf{c}}{Q} \;|\; P \subseteq \llbracket \overleftarrow{\mathsf{c}}\rrbracket Q\}$ \\[3pt]
         \hline
    \end{tabular}
    \caption{Set of valid semantics triple for each direction of the analysis and sense of approximation.}\label{tab:triples}
\end{table}

\revision{According to the choice in Step (1), in Step (2) we select three closure operators: \emph{exists}, \emph{frame}, and \emph{consequence}.\footnote{
\revision{The names of the closure operators comes from the inference rules Exists, Frame and Consequence of existing separation logics \cite{DBLP:conf/cav/RaadBDDOV20, DBLP:journals/pacmpl/RaadBDO22, DBLP:journals/pacmpl/RaadVO24, DBLP:conf/lics/Reynolds02, DBLP:journals/pacmpl/OHearn20, DBLP:journals/pacmpl/AscariBGL25}, from which they have been inspired.}}
The \emph{exists} closure establishes that proving correctness for some concrete values of variables automatically lifts to the family of states where those variables are existentially quantified.}

\begin{definition}[Exists Closure]\label{def:exists_cl}
    $\closure{exists}{\leftrightleftrightarrows} T \triangleq \{\triple{\exists \mathsf{X}.P}{\mathsf{c}}{\exists \mathsf{X}.Q} \;|\; \triple{P}{\mathsf{c}}{Q}\in T, \mathsf{X}\subseteq \logicvar\}$
\end{definition}

\revision{For closure operators, superscripts can be combined to cover all analyses supported by the closure; e.g., the superscript $\leftrightleftrightarrows$ indicates that the closure operator $\rho_{\mathsf{exists}}$ is valid for over- and under\hyp{}approximation of the forward and backward semantics.
Table~\ref{tab:arrows} in Appendix provides a glossary of all superscripts used, for the reviewer's convenience.}
The correctness of the closure $\closure{exists}{\leftrightleftrightarrows}$ immediately follows from the fact that logical variables cannot appear in the program. 

The second closure embodies the essence of the locality principle.
It states that any valid triple that refers to a portion of the memory can be automatically lifted to every other memory extended with a compatible frame. In other words, it captures the role of the \emph{frame} rule in separation logics. 
However, we need to find an appropriate notion of compatibility. 
The classical condition ensuring correctness was $\mathsf{mod(r)\cap fv(R)=\emptyset}$. 
As discussed in Section \ref{sec:introduction}, we argue that this condition works properly for over-approximation but not for under-approximation. 
This is because when dealing with over-approximation, the soundness requirement for framing a triple $\mathsf{\hl{P}{r}{Q}}$ with $\mathsf{R}$ is expressed by the condition $\mathsf{\llbracket r\rrbracket (R\ast P)\subseteq (R\ast Q)}$, which is not guaranteed to be sound also for under-approximations. 
Here, instead, we proceed with a different approach by defining a new notion of \emph{heap compatibility}, where the framing is allowed only if the semantics is preserved. 
The new notion of heap compatibility is driven by the need for a closure that serves only two purposes: 
\begin{enumerate}
    \item to restrict the set of admissible stores in both the pre- and postconditions,
    \item to extend the heap in both pre- and postconditions whenever no conflicts arise.
\end{enumerate}

\revision{For example, let us frame with $\mathsf{x\mapsto\_}$ the triple derived from the axiom of deallocation:}
\begin{equation} \label{eq:wrong_frame}
    \vcenter{\hbox{
        \infer
            {\mathsf{\hl{x\mapsto \_ \ast x\mapsto \_}{free(x)}{\mathsf{x\mapsto\_}}}}
            {\mathsf{\hl{x\mapsto \_}{free(x)}{\mathsf{\emp}}}}
    }}
\end{equation}

Even though the conclusion is a valid over-approximation triple, the exact semantics of $\mathsf{free(x)}$ is not preserved by framing, because, by definition of $\ast$, the pre is $\mathsf{\false}$ but the post is not.
Moreover, extending a triple with an already specified region yields no additional semantic meaning.

In our setting, we want to preclude this kind of framing, because it falls into neither case (1) nor case (2). 
Since approximation is the scope of our third closure operator, we require that the \emph{frame} closure operator just preserve approximation without introducing it.
\revision{To this aim, we say that two sets of memories $M_1$ and $M_2$ are heap compatible} if for each state $(s,h)\in M_1$ and $(s',h')\in M_2$, if they match on the store, then the domains of $h$ and $h'$ are disjoint, i.e. $h\#h'$.

\begin{definition}[Heap Compatibility]\label{def:heap_compatibility}
    Given $M_1, M_2 \revision{\;\subseteq\;} \mathbb{M}^{\circled{i}}$, we define heap compatibility between $M_1$ and $M_2$ as follows:
    $M_1 \heapcomp M_2 \triangleq \forall (s,h)\in M_1, \forall (s',h')\in M_2, s=s' \Rightarrow h\#h'$
\end{definition}

\revision{We note that, although the predicate $\heapcomp$ may resemble the operator $\bullet$, they are quite distinct. The latter requires the stores to be identical and the heaps to be disjoint, whereas the former requires heap disjointness only when the stores coincide. This stronger property is critically important, as it is instrumental in the following results, which guarantee semantic preservation when a local proof is lifted to a heap-compatible memory.} 
Specifically, they establish the preservation of forward and backward semantics.
We recall that $\mathbb{F}$ is the set of universal frames (see Definition~\ref{def:universalframe}).

\begin{theorem}[Forward Preservation] \label{th:forward_preservation}
    Given $P,\mathsf{r}, Q$ such that $\llbracket \mathsf{r} \rrbracket P = Q$ and $R\in \mathbb{F}$ such that $P\heapcomp R$, then $(\llbracket \mathsf{r} \rrbracket P) \sast R = \llbracket \mathsf{r} \rrbracket (P\sast R) = Q\sast R$.
\end{theorem}
\begin{theorem}[Backward Preservation] \label{th:backward_preservation}
    Given $P,\mathsf{r}, Q$ such that $\llbracket \overleftarrow{\mathsf{r}} \rrbracket Q = P$ and $R\in \mathbb{F}$ such that $Q\heapcomp R$, then $(\llbracket \overleftarrow{\mathsf{r}} \rrbracket Q) \sast R = \llbracket \overleftarrow{\mathsf{r}} \rrbracket (Q\sast R) = P\sast R$.
\end{theorem}

We can now define two closure operators \emph{frame}, according to the direction of the analysis.

\begin{definition}[Frame Closures] \label{def:frame_cl}
    The frame closure operators are defined as follows:
    \begin{itemize}
        \item $\closure{frame}{\forwardarrows} \;T \triangleq \{\fwdtriple{P\sast R}{\mathsf{c}}{Q\sast R} \;|\;\fwdtriple{P}{\mathsf{c}}{Q}\in T, R\in \mathbb{F}, P \heapcomp R \}$ (forward semantics),
        \item $\closure{frame}{\backwardarrows} \;T \triangleq \{\bwdtriple{P\sast R}{\mathsf{c}}{Q\sast R} \;|\;\bwdtriple{P}{\mathsf{c}}{Q}\in T, R\in \mathbb{F}, Q \heapcomp R \}$ (backward semantics).
    \end{itemize}
\end{definition}

The third and last family of closure operators account for the admissible weakening and strengthening of pre- and postcondition. 
We note that these closures are the only sources of approximation among the closure operators we consider. 

\begin{definition}[Consequence Closures] \label{def:cons_cl}
The cons closure operators are defined as follows:
    \begin{itemize}
        \item $\closure{cons}{\rightarrowoverdot} \;T \triangleq \{\fwdovertriple{P}{\mathsf{c}}{Q'}\;|\;\fwdovertriple{P}{\mathsf{c}}{Q} \in T, Q\subseteq Q'\}$ (forward-over), 
        \item $\closure{cons}{\rightarrowbelowdot} \;T \triangleq \{\fwdundertriple{P}{\mathsf{c}}{Q'}\;|\;\fwdundertriple{P}{\mathsf{c}}{Q} \in T, Q'\subseteq Q\}$ (forward-under),
        \item $\closure{cons}{\leftarrowoverdot} \;T \triangleq \{\bwdovertriple{P'}{\mathsf{c}}{Q}\;|\;\bwdovertriple{P}{\mathsf{c}}{Q} \in T, P\subseteq P'\}$ (backward-over),
        \item $\closure{cons}{\leftarrowbelowdot} \;T \triangleq \{\bwdundertriple{P'}{\mathsf{c}}{Q}\;|\;\bwdundertriple{P}{\mathsf{c}}{Q} \in T, P'\subseteq P\}$ (backward-under).
    \end{itemize}
\end{definition}

In line with the general methodology, the set of valid axioms for each logic under consideration can be characterized in terms of the closure operators introduced above (as summarized in Table~\ref{tab:closure_combination}).  

\begin{table}[t]
    \centering
    \begin{tabular}{|c|c||c|c|c|}\hline
        \textbf{Direction} & \textbf{Sense} & \textbf{Exists closure} & \textbf{Frame closure} & \textbf{Cons closure} \\ 
        \hline \rule{0pt}{1em} 
        \multirow{2}{*}{\textit{Forward}} & \textit{Over} & \multirow{4}{*}{$\closure{exists}{\leftrightleftrightarrows}$} & \multirow{2}{*}{$\closure{frame}{\forwardarrows}$} & $\closure{cons}{\rightarrowoverdot}$ \\ \cline{2-2}\cline{5-5}
        & \textit{Under} & & & \rule{0pt}{1em}$\closure{cons}{\rightarrowbelowdot}$ \\ \cline{1-2}\cline{4-5}
        \multirow{2}{*}{\textit{Backward}} & \textit{Over} & & \multirow{2}{*}{$\closure{frame}{\backwardarrows}$} & \rule{0pt}{1em}$\closure{cons}{\leftarrowoverdot}$ \\ \cline{2-2}\cline{5-5}
        & \textit{Under} & & & \rule{0pt}{1em} $\closure{cons}{\leftarrowbelowdot}$ \\ \hline
    \end{tabular}
    \caption{Characterization of each logic under consideration in terms of closures.}
    \label{tab:closure_combination}
\end{table}

For the sake of conciseness, the following material concerns only over-approximation logics of the forward semantics; material for all remaining combinations is collected in Appendix \ref{sec:systematic_program_logics_appendix}.

Let us define the closure $\closure{}{\rightarrowoverdot}$ as the closure accounting for any finite combination of closures 
$\closure{exists}{\leftrightleftrightarrows}$,
$\closure{frame}{\forwardarrows}$,  and 
$\closure{cons}{\rightarrowoverdot}$, i.e., we let
$\closure{}{\rightarrowoverdot} \triangleq (\closure{exists}{\leftrightleftrightarrows} \cup \closure{frame}{\forwardarrows} \cup \closure{cons}{\rightarrowoverdot})^\ast$.

\revision{As required by Step (3),} the next result allows us to express $\closure{}{\rightarrowoverdot}$ in terms of a single application (and in very precise order) of $\closure{exists}{\leftrightleftrightarrows}$,
$\closure{frame}{\forwardarrows}$,  and 
$\closure{cons}{\rightarrowoverdot}$: first the \emph{frame} closure then the \emph{exists} closure and finally the \emph{consequence} closure. The proof relies on \cite[\S4.2]{DBLP:books/hal/Cousot78}.

\begin{theorem}[Normalization] \label{th:reorder}
    $\closure{}{\rightarrowoverdot}
     = \closure{cons}{\rightarrowoverdot} \circ \closure{exists}{\leftrightleftrightarrows} \circ \closure{frame}{\forwardarrows}$
\end{theorem}

\revision{At Step (4)} we identify the actual set of axioms and prove that all valid triples can be derived using closure operators.

\begin{definition}[Axioms] \label{def:axioms}
    A set of axioms for $\fwdovertripleset$ is any minimal set $\fwdoveraxiomsset\subseteq \fwdovertripleset$ such that
    \(\fwdovertripleset \triangleq 
    \closure{}{\rightarrowoverdot} (\fwdoveraxiomsset)\). 
\end{definition}

\begin{corollary}[Completeness] \label{cor:triples}
    Any valid set of axioms $\fwdoveraxiomsset$ for $\fwdovertripleset$ is such that:
    \[\fwdovertripleset = (\closure{cons}{\rightarrowoverdot} \circ \closure{exists}{\leftrightleftrightarrows} \circ \closure{frame}{\forwardarrows}) (\fwdoveraxiomsset)\ .\]   
\end{corollary}

\revision{
The very same strategy inspired by Corollary~\ref{cor:triples} is adopted, e.g., in the proof of Theorem~\ref{th:sl_completeness}.}

\begin{figure}[t]
    \renewcommand{\arraystretch}{3}
    \centering
    \resizebox{.99\linewidth}{!}{
    \begin{tabular}{c|cccc|}
         Rule & Fwd-over & Fwd-under & \textcolor{noveltycolor}{Bwd-over} & Bwd-under\\
         \hline \rule{0pt}{3em}
         $\mathsf{Transfer}$ & 
         \textcolor{noveltycolor}{$\infer{\Vdash\hl[noveltycolor][noveltycolor]{P}{\mathsf{c}}{Q}}{\fwdovertriple{P}{\mathsf{c}}{Q}\in \fwdoveraxiomsset}$} &
         \textcolor{noveltycolor}{$\infer{\Vdash\il[noveltycolor][noveltycolor]{P}{\mathsf{c}}{Q}}{\fwdundertriple{P}{\mathsf{c}}{Q}\in \fwdunderaxiomsset}$} &
         \textcolor{noveltycolor}{$\infer{\Vdash\nc[noveltycolor][noveltycolor]{P}{\mathsf{c}}{Q}}{\bwdovertriple{P}{\mathsf{c}}{Q}\in \bwdoveraxiomsset}$} &
         \textcolor{noveltycolor}{$\infer{\Vdash\sil[noveltycolor][noveltycolor]{P}{\mathsf{c}}{Q}}{\bwdundertriple{P}{\mathsf{c}}{Q}\in \bwdunderaxiomsset}$} \\
         \hline
         $\mathsf{Frame}$ &
         \textcolor{noveltycolor}{$\infer{\Vdash\hl[noveltycolor][noveltycolor]{P\sast R}{\mathsf{r}}{Q\sast R}}{R\in \mathbb{F}&\Vdash\hl[noveltycolor][noveltycolor]{P}{\mathsf{r}}{Q}&P\heapcomp R}$} &
         \textcolor{noveltycolor}{$\infer{\Vdash\il[noveltycolor][noveltycolor]{P\sast R}{\mathsf{r}}{Q\sast R}}{R\in \mathbb{F}&\Vdash\il[noveltycolor][noveltycolor]{P}{\mathsf{r}}{Q}&P\heapcomp R}$} &
         \textcolor{noveltycolor}{$\infer{\Vdash\nc[noveltycolor][noveltycolor]{P\sast R}{\mathsf{r}}{Q\sast R}}{R\in \mathbb{F}&\Vdash\nc[noveltycolor][noveltycolor]{P}{\mathsf{r}}{Q}&Q\heapcomp R}$} &
         \textcolor{noveltycolor}{$\infer{\Vdash\sil[noveltycolor][noveltycolor]{P\sast R}{\mathsf{r}}{Q\sast R}}{R\in \mathbb{F}&\Vdash\sil[noveltycolor][noveltycolor]{P}{\mathsf{r}}{Q}&Q\heapcomp R}$} \\
         $\mathsf{Cons}$ &
         $\infer{\Vdash\hl{P}{\mathsf{r}}{Q}}{\Vdash\hl{P}{\mathsf{r}}{Q'}&Q'\subseteq Q}$ &
         $\infer{\Vdash\il{P}{\mathsf{r}}{Q}}{\Vdash\il{P}{\mathsf{r}}{Q'}&Q\subseteq Q'}$ &
         \textcolor{noveltycolor}{$\infer{\Vdash\nc[noveltycolor][noveltycolor]{P}{\mathsf{r}}{Q}}{\Vdash\nc[noveltycolor][noveltycolor]{P'}{\mathsf{r}}{Q}&P'\subseteq P}$} &
         $\infer{\Vdash\sil{P}{\mathsf{r}}{Q}}{\Vdash\sil{P'}{\mathsf{r}}{Q}&P\subseteq P'}$ \\
         $\mathsf{Exists}$ &
         $\infer{\Vdash\hl{\exists \mathsf{X}.P}{\mathsf{r}}{\exists \mathsf{X}.Q}}{\Vdash\hl{P}{\mathsf{r}}{Q}&\mathsf{X}\subseteq \logicvar}$ &
         $\infer{\Vdash\il{\exists \mathsf{X}.P}{\mathsf{r}}{\exists \mathsf{X}.Q}}{\Vdash\il{P}{\mathsf{r}}{Q}&\mathsf{X}\subseteq \logicvar}$ &
         \textcolor{noveltycolor}{$\infer{\Vdash\nc[noveltycolor][noveltycolor]{\exists \mathsf{X}.P}{\mathsf{r}}{\exists \mathsf{X}.Q}}{\Vdash\nc[noveltycolor][noveltycolor]{P}{\mathsf{r}}{Q}&\mathsf{X}\subseteq \logicvar}$} &
         $\infer{\Vdash\sil{\exists \mathsf{X}.P}{\mathsf{r}}{\exists \mathsf{X}.Q}}{\Vdash\sil{P}{\mathsf{r}}{Q}&\mathsf{X}\subseteq \logicvar}$ \\
         $\mathsf{Disj}$ &
         $\infer{\Vdash\hl{\cup_{i\in I} P_i}{\mathsf{r}}{\cup_{i\in I} Q_i}}{\forall i\in I.\Vdash\hl{P_i}{\mathsf{r}}{Q_i}}$ &
         $\infer{\Vdash\il{\cup_{i\in I} P_i}{\mathsf{r}}{\cup_{i\in I} Q_i}}{\forall i\in I.\Vdash\il{P_i}{\mathsf{r}}{Q_i}}$ &
         \textcolor{noveltycolor}{$\infer{\Vdash\nc[noveltycolor][noveltycolor]{\cup_{i\in I} P_i}{\mathsf{r}}{\cup_{i\in I} Q_i}}{\forall i\in I.\Vdash\nc[noveltycolor][noveltycolor]{P_i}{\mathsf{r}}{Q_i}}$} &
         $\infer{\Vdash\sil{\cup_{i\in I} P_i}{\mathsf{r}}{\cup_{i\in I} Q_i}}{\forall i\in I.\Vdash\sil{P_i}{\mathsf{r}}{Q_i}}$\\
         \hline
         $\mathsf{Seq}$ &
         $\infer{\Vdash\hl{P}{\mathsf{r_1;r_2}}{Q}}{\Vdash\hl{P}{\mathsf{r_1}}{S}&\Vdash\hl{S}{\mathsf{r_2}}{Q}}$ &
         $\infer{\Vdash\il{P}{\mathsf{r_1;r_2}}{Q}}{\Vdash\il{P}{\mathsf{r_1}}{S}&\Vdash\il{S}{\mathsf{r_2}}{Q}}$ &
         \textcolor{noveltycolor}{$\infer{\Vdash\nc[noveltycolor][noveltycolor]{P}{\mathsf{r_1;r_2}}{Q}}{\Vdash\nc[noveltycolor][noveltycolor]{P}{\mathsf{r_1}}{S}&\Vdash\nc[noveltycolor][noveltycolor]{S}{\mathsf{r_2}}{Q}}$} &
         $\infer{\Vdash\sil{P}{\mathsf{r_1;r_2}}{Q}}{\Vdash\sil{P}{\mathsf{r_1}}{S}&\Vdash\sil{S}{\mathsf{r_2}}{Q}}$ \\
         $\mathsf{Choice}$ &
         $\infer{\Vdash\hl{P}{\mathsf{r_1+r_2}}{Q}}{\forall i\in \{1,2\}.\Vdash\hl{P}{r_i}{Q}}$ &
         $\infer{\Vdash\il{P}{\mathsf{r_1+r_2}}{Q_1\cup Q_2}}{\forall i\in \{1,2\}.\Vdash\il{P}{r_i}{Q_i}}$ &
         \textcolor{noveltycolor}{$\infer{\Vdash\nc[noveltycolor][noveltycolor]{P}{\mathsf{r_1+r_2}}{Q}}{\forall i\in \{1,2\}.\Vdash\nc[noveltycolor][noveltycolor]{P}{r_i}{Q}}$} &
         $\infer{\Vdash\sil{P_1\cup P_2}{\mathsf{r_1+r_2}}{Q}}{\forall i\in \{1,2\}.\Vdash\sil{P_i}{r_i}{Q}}$ \\ 
         $\mathsf{Iter}$ &
         $\infer{\Vdash\hl{P}{\mathsf{r^\ast}}{P}}{\Vdash\hl{P}{\mathsf{r}}{P}}$ &
         $\infer{\Vdash\il{P_0}{\mathsf{r^\ast}}{\cup_{n\geq0}P_n}}{\forall n\geq0. \Vdash\il{P_n}{\mathsf{r}}{P_{n+1}}}$ &
         \textcolor{noveltycolor}{$\infer{\Vdash\nc[noveltycolor][noveltycolor]{Q}{\mathsf{r^\ast}}{Q}}{\Vdash\nc[noveltycolor][noveltycolor]{Q}{\mathsf{r}}{Q}}$} &
         $\infer{\Vdash\sil{\cup_{n\geq0}Q_n}{\mathsf{r^\ast}}{Q_0}}{\forall n\geq0. \Vdash\sil{Q_{n+1}}{\mathsf{r}}{Q_n}}$\\
         \hline
         $\mathsf{Empty}$ &
         $\infer{\Vdash\hl{\emptyset}{\mathsf{r}}{Q}}{}$ &
         $\infer{\Vdash\il{P}{\mathsf{r}}{\emptyset}}{}$ &
         \textcolor{noveltycolor}{$\infer{\Vdash\nc[noveltycolor][noveltycolor]{P}{\mathsf{r}}{\emptyset}}{}$} &
         $\infer{\Vdash\sil{\emptyset}{\mathsf{r}}{Q}}{}$ \\
         $\mathsf{Cons2}$ &
         $\infer{\Vdash\hl{P}{\mathsf{r}}{Q}}{\Vdash\hl{P'}{\mathsf{r}}{Q}&P\subseteq P'}$ &
         $\infer{\Vdash\il{P}{\mathsf{r}}{Q}}{\Vdash\il{P'}{\mathsf{r}}{Q}&P'\subseteq P}$ &
         \textcolor{noveltycolor}{$\infer{\Vdash\nc[noveltycolor][noveltycolor]{P}{\mathsf{r}}{Q}}{\Vdash\nc[noveltycolor][noveltycolor]{P}{\mathsf{r}}{Q'}&Q\subseteq Q'}$} &
         $\infer{\Vdash\sil{P}{\mathsf{r}}{Q}}{\Vdash\sil{P}{\mathsf{r}}{Q'}&Q'\subseteq Q}$ \\
    \end{tabular}}
    \caption{\revision{Semantic proof systems.  The figure is divided into four groups. First, we list the rules for axioms, according to Corollary~\ref{cor:triples}. In the second group, we list rules induced by closure operators (for $\mathsf{Disj}$ see Section~\ref{sec:disj}). The third group collects standard rules for command compositions, while auxiliary rules are in the fourth group. Key novelties are in \textcolor{noveltycolor}{light blue}. Except for the new transfer and frame rules, the proof systems listed in Columns 1, 2, and 4 have already been proposed in literature \cite{DBLP:journals/pacmpl/AscariBGL25, DBLP:conf/RelMiCS/MollerOH21}. By contrast, Column 3 introduces the first proof system for separation NC.}}
    \Description{Semantic proof systems of program logics}
    \label{fig:semantic_proof_systems}
\end{figure}

In Figure \ref{fig:semantic_proof_systems} we report our general semantic proof system for each logic under consideration, which is applicable for any choice of atomic commands and their semantics. We use the symbol $\Vdash$ for assertions in the semantic proof systems in order to distinguish them from assertion in the logical proof systems that will be discussed in the next sections.
The next result guarantees that the semantic proof systems are sound and complete within their semantics, i.e., a triple is derivable if and only if it is valid in the considered semantics. 
Soundness can be proved by induction on the derivation tree. 
Completeness of atomic commands follows directly from the definitions of semantic axioms.
Furthermore, completeness of compositions for Fwd-over, Fwd-under and Bwd-under can be proved respectively as in \cite{DBLP:journals/siamcomp/Cook78}, \cite{DBLP:journals/pacmpl/OHearn20}, and \cite{ DBLP:journals/pacmpl/AscariBGL25}. 
Finally, completeness of compositions for Bwd-over is proved in the Appendix \ref{sec:appendix_proofs}.

\begin{theorem}[Soundness and completeness] \label{th:sound_complete_semantic}
    All proof systems in Fig. \ref{fig:semantic_proof_systems} are sound and complete.
\end{theorem}

\revision{
Given the assertion language of Section~\ref{sec:al}, 
Step (5) requires to identify a set of pre/postcondition pairs such that every triple in $\mathsf{Axioms}_{\llbracket \mathsf{c} \rrbracket}$ is captured by some of the pre/postcondition pair (and such that every axiom instance is sound too, i.e., it is in $\mathsf{Triples}_{\llbracket \mathsf{c} \rrbracket}$). Interestingly, Corollary~\ref{cor:triples} provides a systematic method to derive logical axioms:
\begin{itemize}
    \item due to $\closure{cons}{\rightarrowoverdot}$, axioms can provide the strongest postcondition for a given precondition;
    \item thanks to $\closure{exists}{\leftrightleftrightarrows}$, logical variables can be used to separate the part of the heap $\mathsf{H}$ that is modified by the atomic command $\mathsf{c}$ from the rest, seen as a suitable universal frame $\mathsf{R}$;
    \item thanks to $\closure{frame}{\forwardarrows}$, each axiom can focus on the precondition $\mathsf{H}$, leaving the frame $\mathsf{R}$ aside.
\end{itemize}
}

\revision{
Since the semantic frame rule requires, as a side condition, the heap compatibility check $P\heapcomp R$ (see Definition~\ref{def:heap_compatibility}), we need to distill a logical characterization of heap compatibility:}

\begin{definition}[Logical Heap Compatibility]\label{def:logic_heap_comp}
    \revision{Given the assertions $\mathsf{M_1=\vee_i(P_i \wedge H_i)\in Ast}$ and $\mathsf{M_2=\vee_j(P_j\wedge H_j)\in Ast}$ we define heap compatibility $\mathsf{M_1 \logicheapcomp M_2}$ between $\mathsf{M_1}$ and $\mathsf{M_2}$ as follows:
    \[\mathsf{M_1\logicheapcomp M_2 \triangleq \forall i,j. (P_i\wedge P_j)\Rightarrow (H_i\ast H_j)}\ .\]}
\end{definition}
\begin{proposition}\label{th:compatibility}
    \revision{For any $\mathsf{M_1},\mathsf{M_2}\in \mathsf{Ast}$ we have:
     \(\mathsf{M_1} \logicheapcomp \mathsf{M_2} \iff \llbrace \mathsf{M_1}\rrbrace \heapcomp \llbrace \mathsf{M_2}\rrbrace\).}
\end{proposition}

\revision{The correspondence in Proposition~\ref{def:logic_heap_comp} grants us that logical heap compatibility $\logicheapcomp$ can safely replace the semantic relation $\heapcomp$ in all frame rules of Figure~\ref{fig:semantic_proof_systems}.}

\paragraph{Mechanization}
\revision{
The main technical results of this section have been mechanized in Isabelle/HOL as a proof-of-concept for our systematic design procedure. Concretely, we have formalized the semantic setting (memories, universal frames, and heap compatibility) together with the three closure operators (consequence, existential generalization, and frame). 
Within this framework, we replicated the central normalization proof, showing that any finite combination of closures can be rearranged into a canonical sequence (frame, then existential generalization, then consequence), thereby confirming the validity of Theorem~\ref{th:reorder}. We then moved to address Step~(5). We fixed an assertion language in the spirit of our separation-logic assertions, we introduced a corresponding normal form for assertions, implemented the procedure that distill local axioms by extracting the assertion that represents the idle frame and verified its correctness in Isabelle.
}

\subsubsection{On Disjunction}\label{sec:disj}
The choices of closure operators need some discussion. 
Indeed, one may think of adding a fourth closure operator to capture the additivity of the collecting semantics. 
In the following, we give a definition of the closure operator \emph{disj} which is independent from the direction of the analysis and the sense of approximation.

\begin{definition}[Disj closure] \label{def:disj_closure}
    $\closure{disj}{\leftrightleftrightarrows} \;T\triangleq \{\triple{\cup_{i\in I} P_i}{\mathsf{c}}{\cup_{i\in I} Q_i}\;|\;\forall i\in I. \triple{P_i}{\mathsf{c}}{Q_i}\in T\}$.
\end{definition}

The reason why we treat the closure operator \emph{disj} separately comes from order theory.
When dealing only with closures \emph{exists}, \emph{cons} and \emph{frame}, we can guide the choice of the axioms by minimality: we can define an order $\leq$ induced by the closure and looking for the minimal elements of that order. Although applying the three closures can result in some preorder (e.g., the \emph{frame} closure can introduce vacuous logical variables that are then eliminated by applying the \emph{exists} closure), we can take the order induced by the equivalent classes modulo $\leq$. Each element of the minimal equivalence class is a candidate axiom. 
On the other hand, when considering also the closure operator \emph{disj}, it is difficult to define a suitable order: the same resulting element can be composed by starting from more than one choice of axioms. 
Notwithstanding this issue, the subsequent results enable us to present and treat every proof system with the disjunction rule in Figure \ref{fig:semantic_proof_systems}, which is the inference rule induced by the closure operator $\closure{disj}{\leftrightleftrightarrows}$. In fact, Definition~\ref{def:disj_closure} allows for a similar result obtained for the previous closures: any finite, arbitrarily ordered use of the four closures can be rewritten as a single application of each, ordered \emph{frame} $\rightarrow$ \emph{exists} $\rightarrow$ \emph{disj} $\rightarrow$ \emph{cons}.
In the following, we just report the result for over-approximation logics of the forward semantics; the corresponding material for all the logics examined is provided in the Appendix \ref{sec:systematic_program_logics_appendix}.

\begin{theorem}\label{th:disj_closure}
    $(\closure{cons}{\rightarrowoverdot} \cup \closure{disj}{\leftrightleftrightarrows} \cup \closure{exists}{\leftrightleftrightarrows} \cup \closure{frame}{\forwardarrows})^\ast = \closure{cons}{\rightarrowoverdot} \circ \closure{disj}{\leftrightleftrightarrows} \circ \closure{exists}{\leftrightleftrightarrows} \circ \closure{frame}{\forwardarrows}$.
\end{theorem}

The following theorem establishes a correspondence between a proof that employs the disjunction closure and one that does not. Specifically, starting from a collection of triples $(P_i, \mathsf{c}, Q_i)$ and allowing the use of the disjunction closure yields the same result as starting from their unified form ${(\bigcup_{i \in I} (P_i \cap d' = i), \mathsf{c}, \bigcup_{i \in I} (Q_i \cap d' = i))}$ and avoiding the use of the disjunction closure. The variable $d'$ serves to distinguish among the different disjuncts, thereby enabling a connection between the pre- and the postcondition. In particular, by assuming the unified triple ${(\bigcup_{i \in I} (P_i \cap d' = i), \mathsf{c}, \bigcup_{i \in I} (Q_i \cap d' = i))}$, one can selectively focus on a subset of cases $K$ by applying a frame  $R=\bigcup_{k \in K} d' = k$. We note that the rule induced by this closure was also used by \citet{DBLP:journals/pacmpl/RaadVO24} with the rule DisjTrack, a non-lossy version of the classical Disj, which preserves the pre-post correspondence and allows a disjunct to be dropped from both the pre- and postcondition.

\begin{theorem} \label{th:disj_no_disj}
    For every index set $I$, 
    \[(\closure{disj}{\leftrightleftrightarrows} \circ \closure{exists}{\leftrightleftrightarrows} \circ \closure{frame}{\forwardarrows})\{(P_i, \mathsf{c}, Q_i)\}_{i\in I} = (\closure{exists}{\leftrightleftrightarrows} \circ \closure{frame}{\forwardarrows})\{(\cup_{i\in I} (P_i\cap d'=i), \mathsf{c}, \cup_{i\in I} (Q_i\cap d'=i))\}\ .\]
\end{theorem}

\begin{remark}
    Theorem \ref{th:disj_no_disj} allows us to include the disjunction rule in all proof systems in Figure \ref{fig:semantic_proof_systems}.
    However, it guarantees that we can find axioms by minimality ignoring the disjunction closure.
\end{remark}

\section{\revision{Systematic Design of Forward Separation Logics}} \label{sec:systematic_forward}
In this section, we instantiate our semantic proof systems for forward analysis (Figure~\ref{fig:semantic_proof_systems}, first two columns) by considering the program semantics of SL \cite{DBLP:conf/csl/OHearnRY01, DBLP:conf/lics/Reynolds02} and ISL \cite{DBLP:conf/cav/RaadBDDOV20}.
\revision{As discussed in Section~\ref{sec:background}, SL and ISL consider two different memory models; however we will synthesize SL and ISL proof systems for both memory models, using the assertion language $\mathsf{Ast}$ defined in~\eqref{eq:ast}.}
Following SL and ISL, we consider atomic commands $\mathsf{c}$ to be:
\begin{equation} \label{eq:atom}
    \mathsf{Atom} \ni \mathsf{c} ::= \; \mathsf{b?} \;|\; \mathsf{error()}\;|\; \mathsf{x:=e} \;|\; \mathsf{x:=alloc()} \;|\; \mathsf{free(x)} \;|\; \mathsf{x:=[y]} \;|\; \mathsf{[x]:=y} 
\end{equation}
where $\mathsf{b?}$ filters the stores by discarding those that do not satisfy the Boolean expression $\mathsf{b}$; $\mathsf{error()}$ halts the execution and raises an error; $\mathsf{x:=e}$ is the usual assignment; $\mathsf{x:=alloc()}$ allocates a new location on the heap with a nondeterministic value and assigns the location to $\mathsf{x}$; $\mathsf{free(x)}$ deallocates the location on the heap pointed by $\mathsf{x}$; $\mathsf{x:=[y]}$ assigns to $\mathsf{x}$ the value contained in the location on the heap pointed by $\mathsf{y}$; and $\mathsf{[x]:=y}$ assigns the value of $\mathsf{y}$ to the location on the heap pointed by $\mathsf{x}$. 
Without loss of generality, we assume the variables $\mathsf{x}$ and $\mathsf{y}$ to be different in the above grammar productions.
Classical no-op command $\mathsf{\skipexp}$ can be defined as the macro $\mathsf{\skipexp \triangleq \true?}$.

\subsection{\revision{Design of SL+}} \label{sec:sl+}
As a first instance of our methodology, we consider the semantics used in SL \cite{DBLP:conf/csl/OHearnRY01,DBLP:conf/lics/Reynolds02}, 
\revision{with closure operators $\closure{exists}{\leftrightleftrightarrows}$,
$\closure{frame}{\forwardarrows}$, and 
$\closure{cons}{\rightarrowoverdot}$ from Table~\ref{tab:closure_combination}.}
For the memory model $\circled{i}\in \{\circled{1},\circled{2}\}$, we rely on a functional version \revision{$\semantics{\cdot}{SL}{\circled{i}}: \mathsf{Reg} \rightarrow \mem \rightarrow \wp(\mem\uplus \{\abort\})$}. Both are formalized in the Appendix (Figure~\ref{fig:sl_semantics}) for reviewers' convenience. The distinguished outcome $\abort$ serves to indicate a generic fault.\footnote{We assume that the abort state is propagated during the computation, i.e., we let $\semantics{r}{SL}{\circled{i}} \abort \triangleq \{\abort\}$ with $\circled{i}\in \{\circled{1},\circled{2}\}$.}
For notational convenience, we use the same symbol $\semantics{\cdot}{SL}{\circled{i}}$ to denote the lifted version to set of states, defined by union as usual.\footnote{$\semantics{\mathsf{r}}{SL}{\circled{i}} P = \bigcup_{(s,h)\in P} \semantics{\mathsf{r}}{SL}{\circled{i}} (s,h)$.}
In Figure~\ref{fig:axioms_sl+} we introduce the \revision{logical axioms of SL+ obtained applying Step (5) to the semantic axioms for the two memory models}, where we characterize pre- and postconditions using the assertion language of \eqref{eq:ast} (see Figure \ref{fig:sl+} in Appendix for the full proof system).  
\begin{figure*}
    \centering
    \resizebox{.99\linewidth}{!}{\fbox{\renewcommand{\arraystretch}{2.5}
    \begin{tabular}{c}
        \begin{tabular}{c}
            \begin{tabular}{@{}l@{\hspace{10pt}}l@{\hspace{10pt}}l@{\hspace{10pt}}l@{}}
                \hllabel{Alloc2}{\circled{2}} & \hllabel{Free}{\circled{1}} & \hllabel{Free2}{\circled{2}}\\
                $\infer{\hl{\mathsf{\emp_{\progvar} \ast l'\not\mapsto}}{\mathsf{x:=alloc()}}{\mathsf{\emp_{\progvar\setminus \{x\}} \ast l'\mapsto z'\wedge x=l'}}}{}$ &
                $\infer{\hl{\mathsf{\emp_{\progvar}\ast x\mapsto z'}}{\mathsf{free(x)}}{\mathsf{\emp_{\progvar}}}}{}$ &
                $\infer{\hl{\mathsf{\emp_{\progvar}\ast x\mapsto z'}}{\mathsf{free(x)}}{\mathsf{\emp_{\progvar}\ast x\not\mapsto}}}{}$ 
            \end{tabular}
        \end{tabular}\\
    
        \begin{tabular}{c}
            \begin{tabular}{@{}l@{\hspace{10pt}}l@{\hspace{10pt}}l@{\hspace{10pt}}l@{}}
                \hllabel{Alloc}{\circled{i}} & \hllabel{Assign}{\circled{i}} & \hllabel{Assume}{\circled{i}} & \\
                $\infer{\hl{\mathsf{\emp_{\progvar}}}{\mathsf{x:=alloc()}}{\mathsf{\emp_{\progvar\setminus \{x\}} \ast x\mapsto z'}}}{}$ &
                $\infer{\hl{\mathsf{\emp_{\progvar\setminus \{x\}}\wedge x=x'}}{\mathsf{x:=e}}{\mathsf{\emp_{\progvar\setminus \{x\}}\wedge x=e[x'/x]}}}{}$ &
                $\infer{\hl{\mathsf{\emp_{\progvar}}}{\mathsf{b?}}{\mathsf{\emp_{\progvar}\wedge b}}}{}$ &
            \end{tabular}
        \end{tabular}\\

        \begin{tabular}{c}
            \begin{tabular}{@{}l@{\hspace{10pt}}l@{\hspace{10pt}}l@{\hspace{10pt}}l@{}}
                \hllabel{Load}{\circled{i}} & \hllabel{Store}{\circled{i}} \\
                $\infer{\hl{\mathsf{\emp_{\progvar}\ast y\mapsto z'}}{\mathsf{x:=[y]}}{\mathsf{\emp_{\progvar\setminus \{x\}} \ast y\mapsto z' \wedge x=z'}}}{}$ &
                $\infer{\hl{\mathsf{\emp_{\progvar}\ast x\mapsto z'}}{\mathsf{[x]:=y}}{\mathsf{\emp_{\progvar}\ast x\mapsto y}}}{}$ 
            \end{tabular}
        \end{tabular}\\

    \end{tabular}
     }}
     \caption{SL+: Synthesized axioms for forward correctness analysis in both memory models $\circled{i}\in \{\circled{1}, \circled{2}\}$.}
    \Description{Logical proof systems SL+}
    \label{fig:axioms_sl+}
\end{figure*}
For compactness and to highlight commonalities, we superscript every rule's name with \circled{1} and \circled{2}, according to the memory model we are considering. Rules that are common to both semantics are superscripted with \circled{i}. 
For example \hllabel[\underline]{Free}{\circled{1}} is valid only in the memory model \circled{1}, $\hllabel[\underline]{Alloc2}{\circled{2}}$ is valid only in the memory model \circled{2} and $\hllabel{Alloc}{\circled{i}}$ is valid for both memory models. 
\revision{We denote by $\vdash_{\mathsf{SL+}}^\text{\circled{i}}$ the proof system related to the memory model $\circled{i}$.}

Before delving into the proof system, we look at Floyd's axioms for assignment in the two proposed versions. On the left, the classical global version \cite{Floyd1967Flowcharts}, and on the right the local version proposed for SL \cite{DBLP:conf/csl/OHearnRY01}, and also used in ISL \cite{DBLP:conf/cav/RaadBDDOV20}:
\begin{align*}
    \infer{\mathsf{\hl{P}{x:=e}{\exists x'.P[x'/x]\wedge x=e[x'/x]}}}{} &&
    \infer{\mathsf{\hl{x\doteq x'}{x:=e}{x\doteq e[x'/x]}}}{} 
\end{align*}

The global version applies to any precondition $\mathsf{P}$ and the postcondition uses the logical variable $\mathsf{x'}$ to describe how to manipulate $\mathsf{P}$ when we assign the expression $\mathsf{e}$ to $\mathsf{x}$.  
For example, starting with a precondition that requires $\mathsf{x}$ to be odd, we can derive the following triple:
\begin{equation} \label{eq:global}
    \mathsf{\hl{odd(x)}{x:=x+1}{\exists x'. odd(x')\wedge x=x'+1}}
\end{equation}

On the other hand, the local version adheres to the \emph{principle of locality} of the separation logics where the premise accounts for the minimum heap needed to execute the assignment, which is the empty heap (recall that $\mathsf{x\doteq x'}$ requires the heap to be empty). 
Furthermore, the premise uses a logical variable that captures the value of $\mathsf{x}$ before the assignment, which will be used in the postcondition. 
Using the frame rule to constrain the value of $\mathsf{x'}$, thereby limiting the value of $\mathsf{x}$ before the assignment, yields the same result as \eqref{eq:global}:
\[
    \infer
        {\mathsf{\hl{x\doteq x'\ast odd(x')}{x:=x+1}{x=x'+1\ast odd(x')}}}
        {\mathsf{\hl{x\doteq x'}{x:=x+1}{x=x'+1}}}
\]

The proof system in Figure~\ref{fig:axioms_sl+} heavily employs logical variables akin to $\mathsf{x'}$ in the preceding example; in SL+, however, this design choice emerges directly from the memory model detailed in Section~\ref{sec:semantic_model}.
Specifically, axioms are selected by minimality w.r.t. (also) the frame closure, which uses only universal frames (Definition~\ref{def:universalframe}). 
Both the pre- and postcondition must include a mechanism that relies exclusively on universal frames to embed the local reasoning within a broader memory. 
The existence of a mirroring logical variable for each program variable in the precondition guarantees a way to constrain the value of $\mathsf{x}$ before executing the command.
To reflect this in our new model, we introduce a shorthand assertion, denoted $\emp_\mathsf{X}$. 
It is analogous to the classical assertion $\emp$ but it imposes the existence of a logical variable $\mathsf{x}'$ for each variable $\mathsf{x}\in\mathsf{X}$:
\[
    \emp_\mathsf{X} \triangleq \textstyle\bigwedge_{\mathsf{x}\in \mathsf{X}}\mathsf{x}\doteq \mathsf{x}' 
    \label{eq:empX}
\]

Each SL+ axiom incorporates the above shorthand, thereby ensuring the presence of a distinct logical variable for every program variable and making these variables explicit whenever required. For instance, in rule \hllabel{Assign}{\circled{i}}, the logical variable $\mathsf{x'}$ associated with $\mathsf{x}$ is made explicit in the post. 

In the following, we discuss the proof system in more detail, by highlighting its main features and advantages w.r.t. the classical SL.

\paragraph{Frame rule and logical variables}
The peculiarity of our derived proof system lies in the combination of logical variables with the frame rule, which can predicate only about them. 
In fact, the logical frame rule (on the left) immediately follows from its semantic version in Figure~\ref{fig:semantic_proof_systems} (on the right), the notion of universal frames (Definition~\ref{def:universalframe}) and the heap compatibility result in Proposition~\ref{th:compatibility}:
\[\infer[\revision{\mathsf{(Frame)}}]{\hl{\mathsf{P}\ast \mathsf{R}}{\mathsf{r}}{\mathsf{Q}\ast \mathsf{R}}} 
        {\mathsf{fv(R)\subseteq \logicvar} & \hl{\mathsf{P}}{\mathsf{r}}{\mathsf{Q}} & \mathsf{P\logicheapcomp R}}
  \qquad\qquad
  \infer{\Vdash\hl{P\sast R}{\mathsf{r}}{Q\sast R}}
        {R\in \mathbb{F}&\Vdash\hl{P}{\mathsf{r}}{Q}&P\heapcomp R}\]

This represents one of the main differences between SL and SL+, because we can abandon the syntactic side condition about modified variables in $\mathsf{r}$. 
In standard SL, logical variables are used only to refer to and constrain the values of variables before the assignments, in the way we described before. 
Instead, in SL+, logical variables are also used to avoid framing with modified variables. 
To see how the axioms change, let us consider the case of \hllabel{Assign}{\circled{i}} in SL+:
\[\infer{\hl{\mathsf{\emp_{\progvar\setminus \{x\}} \wedge x\doteq x'}}{\mathsf{x:=e}}{\mathsf{\emp_{\progvar\setminus\{x\}} \wedge x\doteq e[x'/x]}}} {}\]

In the precondition, a logical variable is associated with each program variable. 
However, in the postcondition, the constraint $\mathsf{x\doteq x'}$ is no longer present because $\mathsf{x}$ has been modified. 
As a result, it is not possible to constrain the new value of $\mathsf{x}$ with a frame, but only its previous value. We notice that to achieve the same result in the original frame rule, the syntactic side condition is needed, which is an over-approximation of the actual modified variables and so a source of incompleteness. Instead, SL+ can capture exactly when a variable is actually modified: if the variable $\mathsf{x}$ is actually modified, we lose the ability to constrain $\mathsf{x}$ in the postcondition, because there is no logical variable that can serve as an alias to $\mathsf{x}$ for framing. Conversely, if the program does not modify any variable, each relation with logical variables is preserved. 
\revision{As discussed in Section \ref{sec:systematic_program_logics}, we observe that the above Frame rule is also valid for under-approximation, since the framing depends solely on the direction of the analysis and not on its orientation; consequently, the same reasoning applies to ISL+, which will be introduced in the next section.}
To highlight the potential of our approach, we show that SL+ allows to derive more triples than SL. 

\begin{example}[$\vdash_{\mathsf{SL+}}^\text{\circled{1}}$
 is more expressive than SL]\label{ex:sl+vssl}
\revision{The (forward, over-approximation) valid triple 
$$\hl{\mathsf{x\mapsto \_}}{\mathsf{(\false?; x:=[y]) + (\true?;\skipexp)}}{\mathsf{x\mapsto \_}}$$ 
discussed in Section~\ref{sec:introduction} can be derived in $\vdash_{\mathsf{SL+}}^{\circled{1}}$ but not in classical SL (see Example~\ref{eg:sl_derivation} for details).}
\end{example}

Vice versa, it can be shown that any SL derivation can be rephrased using SL+ inference rules---roughly, all SL axioms can be derived from the ones in SL+ by using \hllabel{Exists}{\circled{i}}, while the frame rule of SL can be recovered from the one in SL+ by using \hllabel{Cons}{\circled{i}}.

\paragraph{Different memory models}
As suggested by \citet{DBLP:conf/cav/RaadBDDOV20}, we make explicit the proof system for over-approximation in the memory model \circled{2}. 
The only difference between the proof systems $\vdash_{SL+}^\text{\circled{1}}$ 
 and $\vdash_{SL+}^\text{\circled{2}}$ is the first row of Figure~\ref{fig:axioms_sl+}, involving the axioms for deallocation (\hllabel{Free}{\circled{1}} and \hllabel{Free2}{\circled{2}}) and for the reallocation of deallocated locations (\hllabel{Alloc2}{\circled{2}}).
All the other rules remain identical.
Note that, for the memory model \circled{2}, \revision{the systematic design} would produce one single axiom for $\mathsf{x:=alloc()}$, whose pre- and postconditions for the two distinct cases---new allocation or reallocation---are indexed by logical variables, according to Theorem~\ref{th:disj_no_disj} (cf. the proof of Theorem~\ref{th:sl_completeness}, in Appendix). However, in Figure~\ref{fig:axioms_sl+} we list one separate axiom for each case to improve readability.

\paragraph{Different approximations}
A distinctive feature of our proof system is that axioms are independent of the sense of approximation. This can be seen below, by comparing rule \hllabel{Free}{\circled{1}} (on the left), whose precondition requires $\mathsf{x\mapsto z'}$, instead of $\mathsf{x\mapsto \_}$ like SL rule does (on the right):
\begin{align*}
    \infer{\mathsf{\hl{\emp_{\progvar}\ast x\mapsto z'}{free(x)}{\emp_{\progvar}}}}{} &&
    \infer{\mathsf{\hl{x\mapsto \_}{free(x)}{\emp}}}{} 
\end{align*}

In SL, the value pointed to by $\mathsf{x}$ is irrelevant. This is because we can always shrink the precondition and select a subset of the possible values. Instead, in SL+, we require $\mathsf{x}$ to point to a logical variable $\mathsf{z'}$. By using the rule of \hllabel{Exists}{\circled{i}} and \hllabel{Frame}{\circled{i}}, it is possible to derive the same axiom of SL. On the other hand, it turns out that our axioms are valid both for over- and under-approximation, allowing a unique set of axioms for both senses of approximation (see Section~\ref{sec:isl+} for more details). The same applies for rules \hllabel{Free2}{\circled{2}}, \hllabel{Store}{\circled{i}}, \hllabel{Load}{\circled{i}}, \hllabel{Alloc}{\circled{i}} and \hllabel{Alloc2}{\circled{2}}. \\

The following theorem guarantees that SL+ is sound and relative complete for loop-free programs.
\begin{theorem}[SL+ Sound and Relative Complete for Loop-Free Programs] \label{th:sl_completeness}
    The proof system SL+ (sketched in Figure \ref{fig:axioms_sl+}) is sound. Moreover, for loop-free programs, it is also relative complete.
\end{theorem}

\begin{remark}
    We want to note that, in line with \cite[\S 3]{DBLP:conf/lics/CalcagnoOY07} soundness and relative completeness results do not account for $\abort$ state as it is not expressible in the assertion language.  
\end{remark}

\subsection{\revision{Design of ISL+}} \label{sec:isl+}
In this section, we \revision{focus on} proof systems for forward under-approximation, \revision{i.e., subject to closure operators $\closure{exists}{\leftrightleftrightarrows}$,
$\closure{frame}{\forwardarrows}$, and
$\closure{cons}{\rightarrowbelowdot}$ from Table~\ref{tab:closure_combination}.}
Consistently with the principles of under-approximation logics \cite{DBLP:journals/pacmpl/OHearn20, DBLP:conf/cav/RaadBDDOV20, DBLP:journals/pacmpl/RaadBDO22}, we adopt semantics that explicitly account for errors, where the flags \ok and \er track, respectively, normal and erroneous terminations in the postcondition. 
While \citet{DBLP:journals/pacmpl/OHearn20} and later works \cite{DBLP:conf/cav/RaadBDDOV20, DBLP:journals/pacmpl/RaadBDO22, DBLP:conf/concur/RaadVBO23} rely on two different semantics, one accounting for normal termination and one for erroneous termination, here we follow the alternative exploited by \citet{DBLP:journals/jacm/BruniGGR23} and also used in other works \cite{DBLP:journals/pacmpl/AscariBGL25} where the domain is enriched with such flags, i.e., $\{\ok,\er\}\times \mem$. In this way, the only difference between the semantics functions of SL and ISL is the treatment of errors: while SL returns the $\abort$ state, ISL returns the whole state tagged as $\er$.
As in the previous section, we present slightly different semantics $\semantics{\cdot}{ISL}{\circled{i}}: \mathsf{Reg} \rightarrow \{\ok, \er\}\times \mem \rightarrow \wp(\{\ok,\er\}\times \mem)$ depending on the memory model $\circled{i}\in\{\circled{1},\circled{2}\}$. 
Both are formalised in Figure \ref{fig:isl_semantics} for readers’ convenience.\footnote{We assume that error states are propagated during the computation, i.e., we let $\semantics{r}{ISL}{\circled{i}} \er :(s,h) \triangleq \{\er:(s,h)\}$.}
Finally, we lift both definitions to sets of states by union as usual.
Note that the semantics $\semantics{\cdot}{ISL}{\circled{2}}$ is the functional version of ISL relational semantics  \cite{DBLP:conf/cav/RaadBDDOV20}. 
For readability, in proof systems that deal with explicit error handling, we color assertions in green for states tagged with $\ok$, in red if tagged with $\er$, and in blue when they can be both. 

\begin{figure*}[t]
    \centering
    \resizebox{.99\linewidth}{!}{\fbox{\renewcommand{\arraystretch}{2.5}
        \begin{tabular}{c}
            \begin{tabular}{c}
                \begin{tabular}{@{}l@{\hspace{10pt}}l@{\hspace{10pt}}l@{}}
                    \illabel{Alloc2}{\circled{2}} & \illabel{Free1}{\circled{1}} & \illabel{Free2}{\circled{2}} \\
                    $\infer{\il[blue][mygreen]{\mathsf{\emp_{\progvar}\ast l'\not\mapsto}}{\mathsf{x:=alloc}()}{\mathsf{\emp_{\progvar\setminus \{x\}}\ast l'\mapsto z'\wedge x=l'}}} {}$ &
                    $\infer{\il[blue][mygreen]{\mathsf{\emp_{\progvar} \ast x\mapsto z'}}{\mathsf{free(x)}}{\emp_{\progvar}}} {}$ &
                    $\infer{\il[blue][mygreen]{\mathsf{\emp_{\progvar} \ast x\mapsto z'}}{\mathsf{free(x)}}{\emp_{\progvar} \ast\mathsf{x\not\mapsto}}} {}$
                \end{tabular} \\[-10pt]

                \begin{tabular}{@{}l@{\hspace{10pt}}l@{\hspace{10pt}}l@{}}
                    \illabel{FreeEr3}{\circled{2}} & \illabel{LoadEr3}{\circled{2}} & \illabel{StoreEr3}{\circled{2}}  \\
                    $\infer{\il[blue][red]{\mathsf{\emp_{\progvar}\ast x\not\mapsto}}{\mathsf{free(x)}}{\mathsf{\emp_{\progvar}\ast x\not\mapsto}}} {}$ &
                    $\infer{\il[blue][red]{\mathsf{\emp_{\progvar} \ast y \not \mapsto}}{\mathsf{x:=[y]}}{\mathsf{\emp_{\progvar} \ast y \not \mapsto}}} {}$  &
                    $\infer{\il[blue][red]{\mathsf{\emp_{\progvar}\ast x \not \mapsto}}{\mathsf{[x]:=y}}{\mathsf{\emp_{\progvar}\ast x \not \mapsto}}} {}$  
                \end{tabular} \\[-10pt]
                
                \begin{tabular}{@{}l@{\hspace{10pt}}l@{\hspace{10pt}}l@{}}
                    \illabel{Assign}{\circled{i}} & \illabel{Assume}{\circled{i}} & \illabel{Error}{\circled{i}} \\
                    $\infer{\il[blue][mygreen]{\mathsf{\emp_{\progvar\setminus \{x\}} \wedge x=x'}}{\mathsf{x:=e}}{\mathsf{\emp_{\progvar\setminus\{x\}} \wedge x=e[x'/x]}}} {}$ &
                    $\infer{\il[blue][mygreen]{\mathsf{\emp_{\progvar}}}{\mathsf{b?}}{\mathsf{\emp_{\progvar}\wedge b}}} {}$ &
                    $\infer{\il[blue][red]{\mathsf{\emp_{\progvar}}}{\mathsf{error()}}{\mathsf{\emp_{\progvar}}}} {}$ 
                \end{tabular} \\[-10pt]

                \begin{tabular}{@{}l@{\hspace{10pt}}l@{\hspace{10pt}}l@{}}
                    \illabel{Alloc}{\circled{i}} & \illabel{Load}{\circled{i}} & \illabel{Store}{\circled{i}}\\
                    $\infer{\il[blue][mygreen]{\mathsf{\emp_{\progvar}}}{\mathsf{x:=alloc}()}{\mathsf{\emp_{\progvar\setminus \{x\}}\ast x\mapsto z'}}} {}$ &
                    $\infer{\il[blue][mygreen]{\emp_{\progvar} \ast\mathsf{y \mapsto z'}}{\mathsf{x:=[y]}}{\emp_{\progvar\setminus \{x\}} \ast\mathsf{y \mapsto z' \wedge x= z'}}} {}$ &
                    $\infer{\il[blue][mygreen]{\mathsf{\emp_{\progvar} \ast x \mapsto z'}}{\mathsf{[x]:=y}}{\mathsf{\emp_{\progvar} \ast x\mapsto y}}} {}$ 
                \end{tabular} \\[-10pt]

                \begin{tabular}{@{}l@{\hspace{10pt}}l@{\hspace{10pt}}l@{}}
                    \illabel{FreeEr1}{\circled{i}} & \illabel{LoadEr1}{\circled{i}} & \illabel{StoreEr1}{\circled{i}}\\
                    $\infer{\il[blue][red]{\mathsf{\emp_{\progvar}\ast x\doteq null}}{\mathsf{free(x)}}{\mathsf{\emp_{\progvar}\ast x\doteq null}}} {}$  &
                    $\infer{\il[blue][red]{\mathsf{\emp_{\progvar} \ast y\doteq null}}{\mathsf{x:=[y]}}{\mathsf{\emp_{\progvar} \ast y\doteq null}}} {}$  &
                    $\infer{\il[blue][red]{\mathsf{\emp_{\progvar}\ast x\doteq null}}{\mathsf{[x]:=y}}{\mathsf{\emp_{\progvar}\ast x\doteq null}}} {}$  
                \end{tabular} \\[-10pt]

                \begin{tabular}{@{}l@{\hspace{10pt}}l@{\hspace{10pt}}l@{}}
                    \illabel{FreeEr2}{\circled{i}} & \illabel{LoadEr2}{\circled{i}} & \illabel{StoreEr2}{\circled{i}}\\
                    $\infer{\il[blue][red]{\mathsf{\emp_{\progvar}\ast x\nnmapsto}}{\mathsf{free(x)}}{\mathsf{\emp_{\progvar}\ast x\nnmapsto}}} {}$  &
                    $\infer{\il[blue][red]{\mathsf{\emp_{\progvar} \ast y\nnmapsto}}{\mathsf{x:=[y]}}{\mathsf{\emp_{\progvar} \ast y\nnmapsto}}} {}$  &
                    $\infer{\il[blue][red]{\mathsf{\emp_{\progvar}\ast x\nnmapsto}}{\mathsf{[x]:=y}}{\mathsf{\emp_{\progvar}\ast x\nnmapsto}}} {}$  
                \end{tabular} \\
                
            \end{tabular}\\
        \end{tabular}
     }}       
     \caption{ISL+: Synthesized axioms for forward incorrectness analysis in both memory models $\circled{i}\in \{\circled{1}, \circled{2}\}$.}
    \Description{Logical proof systems ISL+}
    \label{fig:axioms_isl+}
\end{figure*}

\paragraph{Underspecified behaviour}
Examining the ISL proof system \cite{DBLP:conf/cav/RaadBDDOV20} \revision{and its semantics}, we can identify some underspecification glitches: no axiom accounts for the execution of a command that presupposes a location to be allocated in an initially empty heap. \revision{For example, in the deallocation command, we can identify five minimal pre on the heap under which the command may be executed:} 

\noindent
{\setlength{\multicolsep}{5pt} 
    \begin{multicols}{3}
    \begin{enumerate}
        \item[(1)] $s(\mathsf{x})=null$,
        \item[(4)] $s(\mathsf{x})\in \valdomain\setminus \locdomain$,
        \item[(2)] $h(s(\mathsf{x}))=\bot$,
        \item[(5)] $s(\mathsf{x}) \in \locdomain\setminus dom(h)$.
        \item[(3)] $h(s(\mathsf{x}))\in \valdomain$,
    \end{enumerate} 
    \end{multicols}
}

\revision{ISL covers (1)---(3) via `FreeNull', `FreeEr' and `Free'. Case (4) can be readily prevented by an appropriate type system. Case (5) is problematic: although $\mathsf{x}$ is currently not allocated, it may appear in a future heap extension through framing, potentially leading to an unsound triple. In separation logic, such \emph{"permanently absent"} locations can be expressed via separating implication $\sepimp$ \cite{DBLP:conf/popl/IshtiaqO01, DBLP:journals/bsl/OHearnP99}. The assertion $\mathsf{P \sepimp Q}$ is satisfied by those memories that, if extended with a disjoint part where $\mathsf{P}$ holds, then $\mathsf{Q}$ will hold in the extended heap. 
Accordingly, the axiom that captures case (5) can be stated as $\mathsf{\il[blue][red]{(x \mapsto v\vee x\not\mapsto) \sepimp \false}{free(x)}{(x \mapsto v\vee x\not\mapsto) \sepimp \false}}$, which asserts that whenever we use a frame with $\mathsf{x\mapsto v}$ or $\mathsf{x\not\mapsto}$, we obtain a triple with both pre- and postcondition $\mathsf{\false}$.
Since we only need the special case $\mathsf{P \sepimp \false}$, we avoid adding full $\sepimp$ by extending heaps with a marker $\reserved$ for \emph{reserved} locations: if $h(l)=\reserved$, then $l$ is absent now and remains absent in all future extensions. We capture this with a new assertion $\mathsf{x\nnmapsto}$, which composes as expected (e.g., $\mathsf{x\nnmapsto \ast x\mapsto v \equiv \false}$), obtaining the axiom $\mathsf{\il[blue][red]{x\nnmapsto}{free(x)}{x\nnmapsto}}$ for the case (5). Although $\mathsf{x\nnmapsto}$ resembles $\mathsf{x\not\mapsto}$, they have different meanings. The assertion $\mathsf{x\not\mapsto}$ denotes deallocated locations that may be reallocated; it is necessary in the memory model $\circled{2}$ since such states are directly used in its semantics. By contrast, $\mathsf{x\nnmapsto}$ forbids the location from ever appearing and is needed only when reasoning with explicit error handling, to avoid unsoundness. If errors are not handled explicitly, all errors collapse to $\abort$ (which is not expressible in the assertion language), making $\mathsf{x\nnmapsto}$ superfluous. }

\revision{
A related solution \citep{DBLP:conf/ecoop/CarrottAR25} introduces an exit condition $\textcolor{myorange}{\mathsf{miss}}(l)$ and enforces soundness via a side condition forbidding framing with $l$. In contrast, $\mathsf{x\nnmapsto}$ is a compositional assertion (e.g., $\mathsf{x\nnmapsto \ast y\nnmapsto}$ is meaningful), while $\textcolor{myorange}{\mathsf{miss}}(l)$ is an exit label that does not support such a combination.
}

\paragraph{Proof system}
\revision{In Figure~\ref{fig:axioms_isl+} we outline the axioms of the proof system ISL+, where we characterize every pre- and postcondition with the assertion language~\eqref{eq:ast}, here extended with $\mathsf{x \nnmapsto}$ (see Figure \ref{fig:isl+} in Appendix for the full proof system). 
As before, rule names are annotated with their reference memory model.
We denote by $\vdash_{\mathsf{ISL+}}^\text{\circled{i}}$ the proof system related to the memory model $\circled{i}$.}


\paragraph{Frame rule and logical variables}
\revision{As previously discussed, the frame rule is independent from the sense of the analysis, thereby ISL+ uses the same frame rule of SL+ and exploits logical variables in the same manner, allowing for deriving more triples w.r.t. ISL. In the next example, we show a triple derivable in ISL+ but not in ISL.}

\begin{example}[$\vdash_{\mathsf{ISL+}}^\text{\circled{2}}$ is more expressive than ISL]\label{ex:isl+vsisl}
\revision{As discussed in Section~\ref{sec:introduction}, Eq.~\eqref{eq:triple_ex}, the triple $\il[blue][red]{\mathsf{y\mapsto v}}{\mathsf{x:=alloc();free(y);x:=[y]}}{{\mathsf{y\not\mapsto}}\ast  \mathsf{x\mapsto \_}}$ is provable in $\vdash_{\mathsf{ISL+}}^\text{\circled{2}}$ but not in ISL (see Ex.~\ref{eg:isl_derivation}).}
\end{example}

\paragraph{Different memory models}
To the best of our knowledge, this work presents the first axiomatization of a separation logic for under-approximation that dispenses with the negative heap assertion, originally introduced by \citet{DBLP:conf/cav/RaadBDDOV20} to guarantee soundness. 
We argue that the negative heap assertion is not required in general to define an under-approximation logic; rather, it becomes necessary when using the standard frame rule of SL. 
Here, instead, we take another point of view by defining a new frame rule where the semantic preservation is exactly captured (Theorem~\ref{th:forward_preservation} and Theorem~\ref{th:backward_preservation}). This implies that frames that would lead to an over-approximation are denied by Definition~\ref{def:heap_compatibility}. 
Indeed, derivations like \eqref{eq:wrong_frame} are not possible in ISL+ \revision{as shown in the next example.}

\begin{example} [$\vdash_{\mathsf{ISL+}}^\text{\circled{1}}$ forbids wrong framing]
\revision{Let us consider the triple $\mathsf{\il[blue][mygreen]{x\mapsto \_}{free(x)}{\emp}}$, used as motivation for the introduction of the assertion $\mathsf{x\not\mapsto}$ in ISL. 
The classical frame rule~\eqref{eq:classicalframe} would allow the framing with the assertion $\mathsf{x\mapsto\_}$, leading to derive the unsound triple $\mathsf{\il[blue][mygreen]{\false}{free(x)}{x\mapsto\_}}$ (see derivation~\eqref{eq:wrong_frame}). 
Instead, in ISL+ such framing is prevented by the condition $P \logicheapcomp R$.}
\end{example}

\paragraph{Different approximations}
We highlight that the only difference between SL+ and ISL+ is the axioms for explicitly handling error termination: every axiom for normal termination remains the same. 
Notably, this is because we consider different semantics. 
Apart from this, our approach allows to reuse the same axioms for both over- and under-approximation. 

\begin{theorem}[ISL+ Sound and Relative Complete for Loop-Free Programs] \label{th:relative_completeness_loopfree_isl}
    The proof system ISL+ (sketched in Figure~\ref{fig:axioms_isl+}) is sound. Moreover, for loop-free programs, it is also relative complete.
\end{theorem}

In the case of under-approximation, we can also extend the completeness result to any program $\mathsf{r}$ by taking advantage of the reachability requirement of the postcondition.

\begin{theorem}[State-wise completeness] \label{th:statewise_completeness_isl}
    Given a program $\mathsf{r}$ and two states $(\oker: (s,h))$ and $(\oker:(s',h'))$ such that, $(\oker:(s',h'))\in \semantics{r}{ISL}{\circled{i}} (\oker:(s,h))$, for every assertion $\mathsf{P}$ such that $(s,h) \in \llbrace \mathsf{P}\rrbrace$ there exists an assertion $\mathsf{Q}$ such that $(s',h') \in \llbrace \mathsf{Q}\rrbrace$ and $\vdash_{\mathsf{ISL+}}^{\circled{i}}\il{\mathsf{P}}{\mathsf{r}}{\mathsf{Q}}$ is provable. 
\end{theorem}

\revision{It is worth mentioning here the relatively complete proof system for incorrectness separation logic defined by \citet{DBLP:conf/aplas/LeeN24}. Their axioms, however, require global preconditions, which violate the locality principle and can harm scalability and compositionality.}

\paragraph{On the footprint theorem}
Let us consider the footprint relation $\mathsf{foot(\cdot)}$ defined by \citet{DBLP:conf/cav/RaadBDDOV20}, where $\mathsf{foot(r)}$ describes the minimal state needed for some execution of $\mathsf{r}$.\footnote{The original definition takes into account also exit conditions that here we omit for the sake of presentation.}
This means that if $((s,h),(s',h'))\in \mathsf{foot(r)}$, then $h$ only contains locations accessed during the execution of $\mathsf{r}$ and its execution on the state $(s,h)$ terminates in the state $(s',h')$. 
When $\mathsf{c \in Atom}$, we observe some correspondences between $\mathsf{foot(c)}$ and $\fwdunderaxiomsset$. In particular, if two states are related by $\mathsf{foot(c)}$, then there is an axiom whose precondition contains the initial state and the postcondition contains the final state (independently from the values of logical variables). Formally,
\begin{equation} \label{eq:foot_relation}
    ((s,h),(s',h'))\in \mathsf{foot(c)} \implies \exists (P,\mathsf{c}, Q) \in \closure{exists}{\rightarrowbelowdot}(\fwdunderaxiomsset). (s,h)\in P, (s',h')\in Q
\end{equation}

We point out that, while the relation $\mathsf{foot(c)}$ must be given ad hoc for the specific semantics we are considering, the set of axioms $\fwdunderaxiomsset$ can be found by minimality considerations. 

Let us also consider the relation $\mathsf{frame(r)}\triangleq \{(\sigma_p \bullet \sigma_r, \sigma_q \bullet \sigma_r)\;|\; (\sigma_p, \sigma_q)\in R\}$ for which \citet[Theorem 2 (Footprint)]{DBLP:conf/cav/RaadBDDOV20} proved that $\forall \mathsf{c}. \semantics{c}{ISL}{\circled{2}} = \mathsf{frame(foot(c))}$. 
In other terms, the semantics of a program is given by every possible framing of the footprint of the command $\mathsf{c}$. A similar correspondence to \eqref{eq:foot_relation} can be observed if we also consider the closure operator frame:
\[
\resizebox{.99\linewidth}{!}{
$((s,h),(s',h'))\in \mathsf{frame(foot(c))} \implies \exists (P,\mathsf{c}, Q) \in \closure{frame}{\rightarrowbelowdot} \circ \closure{exists}{\rightarrowbelowdot}(\fwdunderaxiomsset). (s,h)\in P, (s',h')\in Q$.
}
\]

Although the Footprint Theorem secures an elegant property, there is no corresponding logical principle in ISL.
In fact, some states in $\mathsf{frame(foot(c))}$ cannot be derived in ISL due to the impossibility of applying the frame rule because of the side condition (see~\eqref{eq:triple_ex}, again). 
On the other hand, conditions imposed by the closure operator $\closure{frame}{\rightarrowbelowdot}$ and then by the rule of frame of the semantic proof system are translated as a one-to-one correspondence in the logical proof system ISL+. This means that (up to the expressiveness of the assertion language), the frame rule of ISL+ is complete.

\subsubsection{Design of CISL+.}
\revision{CISL \citep{DBLP:journals/pacmpl/RaadBDO22} is a general framework for multiple incorrectness analysis in concurrent settings (e.g., detecting memory safety bugs, data races on shared memory, deadlock issues) and avoids the syntactic side-condition by adopting a memory model where separation ranges uniformly over store and heap. Accordingly, a slightly different notion of separating conjunction is used:
\[
P \sastconcur Q \;\triangleq\; \{(s_p \bullet s_q,\, h_p \bullet h_q)\mid (s_p,h_p)\in P \land (s_q,h_q)\in Q \land s_p \# s_q \land h_p \# h_q\}.
\]
Applying our systematic design methodology to $\text{CISL}_{\text{DC}}$ (i.e. the instance of CISL for detecting memory safety bugs) yields the logic CISL+, that is analogous to ISL+, but where $\closure{frame}{\forwardarrows}$ is defined in terms of $\sastconcur$ rather than $\bullet$. Although CISL+ is similar to ISL+\footnote{Technically, CISL also exploits parallel composition of commands $\mathsf{r \parallel r}$; being orthogonal to our study, we leave the rules for handling parallel composition in Figure~\ref{fig:cisl+} in Appendix.}, its axioms reflect the different definition of the separating conjunction. For example, letting $\circledast$ be the assertional counterpart of $\sastconcur$ (i.e., $\llbrace \mathsf{P \circledast Q}\rrbrace \triangleq \llbrace\mathsf{P}\rrbrace \sastconcur \llbrace\mathsf{Q}\rrbrace$), the axiom for allocation can be written as follows:
\[
\infer
    {\il[blue][mygreen]{\mathsf{x\mapsto z'}}{\mathsf{x:=alloc()}}{\mathsf{\exists l. \exists v. x\mapsto l \circledast l\mapsto v}}}
    {}
\]
The axiom above itself provides an example in which CISL+ is more expressive than $\text{CISL}_{\text{DC}}$, which is intrinsically incomplete, as discussed in Section \ref{sec:introduction} (see Issue \#3). }
\revisiontmp{Notably, in the extended model $\islmem$ with explicit deallocation considered by CISL, the side condition of our frame rule $\mathsf{P \logicheapcomp R}$ is vacuously true. However, thanks to that side condition, we can also apply the CISL settings to the classical model $\slmem$, in the same manner described in Section \ref{sec:isl+}.}

\section{\revision{Systematic Design of Backward Separation Logics}} \label{sec:systematic_backward}

\revision{Our methodology is not limited to reconstructing and refining existing logics; it also enables the synthesis of entirely new ones. In this section, we address the design of \emph{backward} separation logics, for any sense of approximation and any memory model. While a backward, under-approximation separation logic, called \emph{SepSIL}, was recently introduced in~\cite{DBLP:journals/pacmpl/AscariBGL25} for $\islmem$, the over-approximation variants presented here are entirely novel contributions of this work.
For simplicity, as reference semantics we consider $\semantics{\cdot}{SL}{\circled{i}}: \mathsf{Reg} \rightarrow \mem \rightarrow \wp(\mem\uplus \{\abort\})$, i.e., without explicit error handling (see Figure~\ref{fig:sl_semantics} in Appendix).
The \revision{resulting} logic \emph{SIL+}, for backward under-approximation and of the logic \emph{NC+} for backward over-approximation exploit the closure operators already presented in Table~\ref{tab:closure_combination}, namely 
$\closure{exists}{\leftrightleftrightarrows}$,
$\closure{frame}{\backwardarrows}$, and 
$\closure{cons}{\leftarrowbelowdot}$ for SIL+, and
$\closure{exists}{\leftrightleftrightarrows}$,
$\closure{frame}{\backwardarrows}$, and 
$\closure{cons}{\leftarrowoverdot}$ for NC+.
}

\revision{We highlight that NC+ and SIL+ share a unique set of axioms, valid for both senses of approximation.
This is evident by comparing the axioms for SIL+ in Figure~\ref{fig:axioms_sil+} with those for NC+ in (the topmost part of) Figure~\ref{fig:nc+}, which carry the same pre- and postconditions except for the use of surrounding brackets: SIL+ triples are written using $\llangle\cdot\rrangle$, while NC+ triples using $(\cdot)$.
In these axioms, both pre- and postconditions are expressed using the assertion language defined in~\eqref{eq:ast}.}

\revision{
As an illustrative example, consider the axiom for allocation:}

\[\infer{\sil{\mathsf{\emp_{\mathbb{X}_\mathsf{p}\setminus \{x\}}\wedge x'=l'}}{\mathsf{x:=alloc}()}{\mathsf{\emp_{\mathbb{X}_\mathsf{p}\setminus \{x\}}\wedge x=x'\ast x\mapsto z'}}} {}\]

\revision{
The axiom is written in SIL+ syntax, but also valid for NC+. 
Due to the locality principle, the pre- and postconditions of the conclusion are restricted to the footprint of the allocated cell, i.e., to the heap fragment directly affected by the allocation.
In contrast with SL+ and ISL+, the analysis of NC+ and SIL+ is best read backward: the axiom is applicable to a generic heap-minimal postcondition. 
}

\subsection{\revision{Design of SIL+}}

\revisiontmp{
The goal of backward under-approximation logics, such as SepSIL, is to expose \emph{some} initial states that have at least one terminating run in one state that satisfies the given postcondition $Q$.}
\revision{Figure~\ref{fig:axioms_sil+} just showcases the axioms of SIL+, but the full proof system can be found in Appendix, Figure \ref{fig:sil+}.
Even if SepSIL is as expressive as SIL+, we observe that some axioms of SepSIL (see \cite[Figure~5]{DBLP:journals/pacmpl/AscariBGL25}) were defined globally, whereas our methodology enforces the introduction of local axioms only.
As a simple instance, compare the global assignment axiom in SepSIL---a.k.a., Hoare’s assignment axiom \cite{DBLP:journals/cacm/Hoare69}, as opposed to Floyd’s axiom~\eqref{eq:floyd}---with SIL+ local axiom:}

\[
    \infer[\mathsf{(SepSIL)}]{\sil{\mathsf{Q[e/x]}}{\mathsf{x:=e}}{\mathsf{Q}}}{}
    \qquad
    \infer[\mathsf{(SIL+)}]{\sil{\mathsf{\emp_{\mathbb{X}_\mathsf{p}\setminus \{x\}} \wedge x'=e}}{\mathsf{x:=e}}{\mathsf{\emp_{\mathbb{X}_\mathsf{p}\setminus\{x\}} \wedge x=x'}}} {}\
\]

\begin{figure*}[t]
    \centering
    \resizebox{.99\linewidth}{!}{\fbox{\renewcommand{\arraystretch}{2.5}
        \begin{tabular}{c}
            \begin{tabular}{c}
                \begin{tabular}{@{}l@{\hspace{10pt}}l@{\hspace{10pt}}l@{}}
                    \sillabel{Alloc2}{\circled{2}} & \sillabel{Free1}{\circled{1}} & \sillabel{Free2}{\circled{2}}  \\
                    $\infer{\sil{\mathsf{\emp_{\mathbb{X}_\mathsf{p}\setminus \{x\}}\wedge x'=l'\ast l'\not\mapsto}}{\mathsf{x:=alloc}()}{\mathsf{\emp_{\mathbb{X}_\mathsf{p}\setminus \{x\}}\wedge x=x'\ast x\mapsto z'}}} {}$ &
                    $\infer{\sil{\mathsf{\emp_{\mathbb{X}_\mathsf{p}}\ast x\mapsto z'}}{\mathsf{free(x)}}{\mathsf{\emp_{\mathbb{X}_\mathsf{p}}}}} {}$ &
                    $\infer{\sil{\mathsf{\emp_{\mathbb{X}_\mathsf{p}}\ast x\mapsto z'}}{\mathsf{free(x)}}{\mathsf{\emp_{\mathbb{X}_\mathsf{p}} \ast x\not\mapsto}}} {}$
                \end{tabular} \\[-10pt]
            
                \begin{tabular}{@{}l@{\hspace{10pt}}l@{\hspace{10pt}}l@{}}
                    \sillabel{Alloc}{\circled{i}} & \sillabel{Assign}{\circled{i}} & \sillabel{Assume}{\circled{i}} \\
                    $\infer{\sil{\mathsf{\emp_{\mathbb{X}_\mathsf{p}\setminus \{x\}}\wedge x'=l'}}{\mathsf{x:=alloc}()}{\mathsf{\emp_{\mathbb{X}_\mathsf{p}\setminus \{x\}}\wedge x=x'\ast x\mapsto z'}}} {}$ &
                    $\infer{\sil{\mathsf{\emp_{\mathbb{X}_\mathsf{p}\setminus \{x\}} \wedge x'=e}}{\mathsf{x:=e}}{\mathsf{\emp_{\mathbb{X}_\mathsf{p}\setminus\{x\}} \wedge x=x'}}} {}$ &
                    $\infer{\sil{\mathsf{\emp_{\mathbb{X}_\mathsf{p}}\wedge b}}{\mathsf{b?}}{\mathsf{\emp_{\mathbb{X}_\mathsf{p}}}}} {}$ 
                \end{tabular} \\[-10pt]            

                \begin{tabular}{@{}l@{\hspace{10pt}}l@{\hspace{10pt}}l@{}}
                    \sillabel{Load}{\circled{i}} & \sillabel{Store}{\circled{i}}\\
                    $\infer{\sil{\mathsf{\emp_{\mathbb{X}_\mathsf{p}\setminus \{x\}} \wedge x'=z' \ast y \mapsto z'}}{\mathsf{x:=[y]}}{\mathsf{\emp_{\mathbb{X}_\mathsf{p}\setminus \{x\}} \wedge x=x' \ast y \mapsto z'}}} {}$ &
                    $\infer{\sil{\mathsf{\emp_{\mathbb{X}_\mathsf{p}}\ast x \mapsto z'}}{\mathsf{[x]:=y}}{\mathsf{\emp_{\mathbb{X}_\mathsf{p}}\ast x\mapsto y}}} {}$
                \end{tabular} \\
                
            \end{tabular}\\

        \end{tabular}
     }}         
     \caption{SIL+: Synthesized axioms for backward incorrectness analysis in both memory models $\circled{i}\in \{\circled{1}, \circled{2}\}$.}
    \Description{Logical proof systems SIL+}
    \label{fig:axioms_sil+}
\end{figure*}

\revision{Our method leads to exploit logical variables to disregard the idle part of the postcondition, which can be added at a later time using the  frame rule.   
More precisely, using logical variables, we can rewrite the SepSIL axiom for assignment as $\sil{\mathsf{Q[x'/x]\ast x'\doteq e}}{\mathsf{x:=e}}{\mathsf{Q[x'/x]\ast x\doteq x'}}$. 
The local axiom of SIL+ lifts this approach to the whole set of program variables $\emp_{\mathbb{X}_\mathsf{p}}$ in order to remove the idle frame $\mathsf{Q[x'/x]}$ from the picture.
The same difference can be observed by looking at the axioms for $\mathsf{x:=e}$, $\mathsf{b?}$, and $\mathsf{x:=[y]}$, where SepSIL still considers global post rather than local ones.
Finally, thanks to universal frames,  SIL+ can be applied to the memory model without reallocation.} 

\subsection{\revision{Design of NC+}}

\begin{figure*}[t]
    \centering
    \resizebox{.99\linewidth}{!}{\fbox{\renewcommand{\arraystretch}{2.5}
        \begin{tabular}{c}
            \begin{tabular}{c}
                \begin{tabular}{@{}l@{\hspace{10pt}}l@{\hspace{10pt}}l@{}}
                    \nclabel{Alloc2}{\circled{2}} & \nclabel{Free1}{\circled{1}} & \nclabel{Free2}{\circled{2}}  \\
                    $\infer{\nc{\mathsf{\emp_{\mathbb{X}_\mathsf{p}\setminus \{x\}}\wedge x'=l'\ast l'\not\mapsto}}{\mathsf{x:=alloc}()}{\mathsf{\emp_{\mathbb{X}_\mathsf{p}\setminus \{x\}}\wedge x=x'\ast x\mapsto z'}}} {}$ &
                    $\infer{\nc{\mathsf{\emp_{\mathbb{X}_\mathsf{p}}\ast x\mapsto z'}}{\mathsf{free(x)}}{\mathsf{\emp_{\mathbb{X}_\mathsf{p}}}}} {}$ &
                    $\infer{\nc{\mathsf{\emp_{\mathbb{X}_\mathsf{p}}\ast x\mapsto z'}}{\mathsf{free(x)}}{\mathsf{\emp_{\mathbb{X}_\mathsf{p}} \ast x\not\mapsto}}} {}$
                \end{tabular} \\[-10pt]
            
                \begin{tabular}{@{}l@{\hspace{10pt}}l@{\hspace{10pt}}l@{}}
                    \nclabel{Alloc}{\circled{i}} & \nclabel{Assign}{\circled{i}} & \nclabel{Assume}{\circled{i}} \\
                    $\infer{\nc{\mathsf{\emp_{\mathbb{X}_\mathsf{p}\setminus \{x\}}\wedge x'=l'}}{\mathsf{x:=alloc}()}{\mathsf{\emp_{\mathbb{X}_\mathsf{p}\setminus \{x\}}\wedge x=x'\ast x\mapsto z'}}} {}$ &
                    $\infer{\nc{\mathsf{\emp_{\mathbb{X}_\mathsf{p}\setminus \{x\}} \wedge x'=e}}{\mathsf{x:=e}}{\mathsf{\emp_{\mathbb{X}_\mathsf{p}\setminus\{x\}} \wedge x=x'}}} {}$ &
                    $\infer{\nc{\mathsf{\emp_{\mathbb{X}_\mathsf{p}}\wedge b}}{\mathsf{b?}}{\mathsf{\emp_{\mathbb{X}_\mathsf{p}}}}} {}$ 
                \end{tabular} \\[-10pt]            

                \begin{tabular}{@{}l@{\hspace{10pt}}l@{\hspace{10pt}}l@{}}
                    \nclabel{Load}{\circled{i}} & \nclabel{Store}{\circled{i}}\\
                    $\infer{\nc{\mathsf{\emp_{\mathbb{X}_\mathsf{p}\setminus \{x\}} \wedge x'=z' \ast y \mapsto z'}}{\mathsf{x:=[y]}}{\mathsf{\emp_{\mathbb{X}_\mathsf{p}\setminus \{x\}} \wedge x=x' \ast y \mapsto z'}}} {}$ &
                    $\infer{\nc{\mathsf{\emp_{\mathbb{X}_\mathsf{p}}\ast x \mapsto z'}}{\mathsf{[x]:=y}}{\mathsf{\emp_{\mathbb{X}_\mathsf{p}}\ast x\mapsto y}}} {}$
                \end{tabular} \\
            \end{tabular}\\

            \hline 
            
            \begin{tabular}{@{}l@{\hspace{10pt}}l@{\hspace{10pt}}l@{\hspace{10pt}}l@{}}
                \nclabel{Frame}{\circled{i}} & \nclabel{Exists}{\circled{i}} & \nclabel{Disj}{\circled{i}} & \nclabel{Seq}{\circled{i}} \\
                $\infer
                    {\nc{\mathsf{P}\ast \mathsf{R}}{\mathsf{r}}{\mathsf{Q}\ast \mathsf{R}}} 
                    {\mathsf{fv(R)\subseteq\mathbb{X}_\mathsf{l}} & \nc{\mathsf{P}}{\mathsf{r}}{\mathsf{Q}} & \mathsf{Q\logicheapcomp R}}$ &
                $\infer{\nc{\mathsf{\exists \mathsf{X}.P}}{\mathsf{r}}{\mathsf{\exists \mathsf{X}.Q}}}{\nc{\mathsf{P}}{\mathsf{r}}{\mathsf{Q}} & \mathsf{X}\subseteq \mathbb{X}_\mathsf{l}}$ &
                $\infer
                    {\nc{\mathsf{\cup_{i\in I}P_i}}{\mathsf{r}}{\mathsf{\cup_{i\in I}Q_i}}} 
                    {\forall i\in I. \nc{\mathsf{P_i}}{\mathsf{r}}{\mathsf{Q_i}}}$ &
                $\infer
                    {\nc{\mathsf{P}}{\mathsf{r_1;r_2}}{\mathsf{Q}}} 
                    {\nc{\mathsf{P}}{\mathsf{r_1}}{\mathsf{S}} & \nc{{S}}{\mathsf{r_2}}{\mathsf{Q}}}$ 
            \end{tabular} \\
            \begin{tabular}{@{}l@{\hspace{10pt}}l@{\hspace{10pt}}l@{\hspace{10pt}}l@{}}
                \nclabel{Cons}{\circled{i}} & \nclabel{Iterate}{\circled{i}} & \nclabel{Choice}{\circled{i}} \\
                $\infer
                    {\nc{\mathsf{P}}{\mathsf{r}}{\mathsf{Q}}} 
                    {\mathsf{P'} \Rightarrow \mathsf{P} & \nc{\mathsf{P}'}{\mathsf{r}}{\mathsf{Q}'} & \mathsf{Q} \Rightarrow \mathsf{Q}'}$ &
                $\infer{\nc{\mathsf{Q}}{\mathsf{r}^\ast}{\mathsf{Q}}} {\nc{\mathsf{Q}}{\mathsf{r}}{\mathsf{Q}}}$ &
                $\infer
                    {\nc{\mathsf{P}}{\mathsf{r_1+r_2}}{\mathsf{Q}}} 
                    {\forall i\in \{1,2\}. \nc{\mathsf{P}}{\mathsf{r_i}}{\mathsf{Q}}}$ 
            \end{tabular}
            \end{tabular}
            }}
     \caption{NC+: Synthesized proof system for backward correctness analysis in both memory models $\circled{i}\in \{\circled{1}, \circled{2}\}$.}
    \Description{Logical proof systems NC+}
    \label{fig:nc+}
\end{figure*}

\revision{
The goal of backward over-approximation logics, such as NC, is to characterise, for a given postcondition $Q$, a precondition that covers \emph{all} initial states with at least one terminating run in one state that satisfies $Q$.
Equivalently, this property ensure that no state in $Q$ is reachable from outside $P$.
Although necessary condition triples were already introduced more than a decade ago by \citet{DBLP:conf/vmcai/CousotCFL13}, a proof system for establishing such triples for either regular programs or heap-manipulating ones had remained missing until now. Thanks to our methodology, such a proof system can be straightforwardly obtained from the common axioms of backward separation logic by simply adding the structural inference rules and the right consequence rule, since in this setting we are over-approximating the preconditions.
The resulting proof system is listed in Fig.~\ref{fig:nc+} and it is parametric to the underlying memory model.
}

\revisiontmp{
The next example illustrates how backward analyses can largely differ from forward ones.}

\begin{example}
\revision{Consider the triples
    \(
    \vdash_{\mathsf{SIL+}}^{\circled{i}} 
    \mathsf{\sil{x\mapsto 1}{\mathsf{r}}{Q}}
    \) and \(
    \vdash_{\mathsf{NC+}}^{\circled{i}}
    \mathsf{\nc{x\mapsto z'}{\mathsf{r}}{Q}}
    \)
    where 
    $\mathsf{r} = (\mathsf{free(x)} + \skipexp); \mathsf{y:=[x]}$
    and 
    $\mathsf{Q=x\mapsto z' \wedge y>0 \wedge z'<3}$ is a correctness specification.
The SIL+ triple is only valid for backward, under-approximation analysis and asserts that any initial state satisfying $\mathsf{x\mapsto 1}$ has a terminating execution in $\mathsf{Q}$.
The NC+ triple is only valid for backward, over-approximation analysis and asserts that only initial states that satisfy $\mathsf{x\mapsto z'}$ could have a terminating execution in $\mathsf{Q}$.
Neither of the triples specifies which final states are reachable.}
\end{example}

\section{Conclusions and future work} \label{sec:conclusion}
We presented a methodology for designing \revision{separation} logics, grounded in collecting semantics, the locality principle, and closure operators. 
By combining semantic closure reasoning with universal frames over logical variables, the method synthesizes sound and expressive proof systems for forward/backward and over/under-approximate analyses across different memory models. 
\revision{The approach unifies, generalizes, and improves SL, ISL, CISL, and SepSIL, and offers a disciplined way to deal with language extensions and to derive new logics with limited additional effort. 
At the core of the approach lies a closure-based characterization of valid triples, which allows local axioms to be derived as minimal elements from the program semantics.
The main outcomes of our systematic design are summarized in Table~\ref{tab:summary+}.}
Because the treatment is semantic, the approach is fully general: beyond the two memory models studied, it applies to any semantics and evaluation model. Moreover, the proof systems introduced for the forward analysis use forward axioms (i.e., axioms applicable to any precondition) to highlight the duality of the analysis, yet by disregarding this duality and focusing on over-approximations, one can still define a Hoare-style assignment axiom incorporating the principle of locality, even when logical variables are employed. We believe our work contributes to a deeper understanding of the semantic foundations of separation logics, also offering practical tools for the scalable and modular verification of heap-manipulating programs.

\revision{
Our methodology has been validated in two ways: by a prototype tool developed as part of a master's course project, which is able to perform forward reasoning, and by the mechanization of several core results in Isabelle/HOL. 
However, they were not the final goal.
A more promising direction is to instrument existing verification tools with our closure-guided principle for deriving axioms. By largely reusing, or only lightly extending, current codebases, one can obtain a principled selection of axiom sets and adapt bi-abduction procedures to match the direction and sense of the analysis (forward vs.\ backward), thereby broadening both expressiveness and applicability. In this setting, the ability to share axioms across different analyses, together with the notion of universal frames, can play a key role in integrating heterogeneous analyses within a unified reasoning framework and in deriving more expressive code summaries. From this perspective, our framework finds potential applications in improving the capabilities of tools such as Infer \cite{DBLP:journals/cacm/DistefanoFLO19}, Pulse \cite{DBLP:journals/pacmpl/LeRVBDO22} and Iris \cite{DBLP:journals/jfp/JungKJBBD18,DBLP:journals/pacmpl/VistrupSJ25}.
}

\revision{Future work also includes the investigation of richer assertion languages, concurrency beyond heap safety, and the extension of bi-abduction to backward reasoning. 
While in this paper the focus is mostly on the systematic design of local axioms for atomic commands, the next steps will be to refine the concept of minimality further and to see if it can be incorporated within the calculational approach of \citet{DBLP:journals/pacmpl/Cousot24}. 
We also plan to study the exploitation of logical variables within probabilistic separation logics \cite{DBLP:journals/pacmpl/BatzKKMN19, DBLP:journals/pacmpl/BartheHL20,DBLP:journals/pacmpl/ZilbersteinST26} and hyper Hoare logics \cite{DBLP:journals/pacmpl/DardinierM24}, along the lines of \citet{DBLP:journals/pacmpl/CousotW25}.}

\section{Data-Availability Statement} \label{sec:data_availability}
\revision{We have developed two supporting tools based on our findings in the paper.}

\revision{A prototype implementing the core ideas of SL+ and ISL+ was developed within a master’s course project. Its objective was to show that our methodology can effectively support both over- and under-approximation analyses, and that logical variables can be practically integrated.} 

\revision{Moreover, the crucial steps of our systematic design methodology have been mechanized in Isabelle/HOL to gain confidence in their correctness and to verify that the locality principle was correctly addressed in our logical axioms.}

\revision{Both artifacts constitute just proof of concepts and they are not intended to achieve competitive performance, precision, or usability.}

\begin{acks}
The research has been supported by the Italian Ministero dell'Università e della Ricerca under Grant No. P2022HXNSC, PRIN 2022 PNRR -- \emph{Resource Awareness in Programming: Algebra, Rewriting, and Analysis}.
\end{acks}

\bibliographystyle{ACM-Reference-Format}
\bibliography{bibliography}

\clearpage

\appendix
\section{Glossary}\label{sec:glossary}
In this appendix, we provide a set of tables that can guide the reader through the various notational conventions used throughout the main text.
Table~\ref{tab:ast} recalls the meaning of the structural assertions.
Table~\ref{tab:arrows} summarizes the superscripts employed to indicate the orientation and interpretation of closure operators and semantic triples, as introduced in Section~\ref{sec:systematic_program_logics}.
Finally, Table~\ref{tab:proof_system_validity} lists all the proof systems discussed in this work, along with their corresponding validity conditions and classifications according to the analytical dimensions considered.

\begin{table}[h]
    \centering
    \renewcommand{\arraystretch}{1.3}
    \begin{tabular}{|c|c|}
        \hline
        \textbf{Structural assertion} & \textbf{Meaning} \\
        \hline
        $\mathsf{\emp}$ & Empty heap \\
        $\mathsf{x\mapsto e}$ &  Single location in the heap pointed by $\mathsf{x}$ whose content is $\mathsf{e}$ \\
        $\mathsf{x\not\mapsto}$ & Single location in the heap pointed by $\mathsf{x}$ which has been deallocated \\
        $\mathsf{P\ast Q}$ & $\mathsf{P}$ and $\mathsf{Q}$ hold in two disjoint portions of the heap \\
        \hline
        \textcolor{noveltycolor}{$\mathsf{x \nnmapsto}$} & \textcolor{noveltycolor}{Single location pointed by $\mathsf{x}$ that cannot be framed} \\
        \hline
    \end{tabular}
    \caption{Intuitive meaning of structural assertions. The assertion \textcolor{noveltycolor}{$\mathsf{x \nnmapsto}$} is novel to this contribution.}
    \label{tab:ast}
\end{table}

\begin{table}[h]
    \centering
    \renewcommand{\arraystretch}{1.3}
    \begin{tabular}{|c|c|c|c|c|c|c|c|}
        \hline
         & $\rightarrowoverdot$ & $\rightarrowbelowdot$ & $\leftarrowoverdot$ & $\leftarrowbelowdot$ & $\forwardarrows$ & $\backwardarrows$ & $\leftrightleftrightarrows$ \\
         \hline
         Over &\checkmark&&\checkmark&&\checkmark&\checkmark& \checkmark\\ 
         Under &&\checkmark&&\checkmark&\checkmark&\checkmark&\checkmark\\
         Forward &\checkmark&\checkmark&&&\checkmark&& \checkmark\\
         Backward &&&\checkmark&\checkmark&&\checkmark& \checkmark\\
        \hline
    \end{tabular}
    \caption{Arrows validity (symbols appearing as superscripts in closures and triples).}
    \label{tab:arrows}
\end{table}

\begin{table}[h]
    \centering
    \renewcommand{\arraystretch}{1.2}
    \begin{tabular}{c|c|c|c|c|c|}
        Proof System & Validity & Direction & Sense & Memory model & Error Handling \\
        \hline
        $\proofsystem{SL+}{\circled{1}} \mathsf{\hl{P}{\mathsf{r}}{Q}}$ & $\semantics{r}{SL}{\circled{1}} \llbrace \mathsf{P}\rrbrace \subseteq \llbrace \mathsf{Q}\rrbrace$ & Forward & Over & No reallocation & $\times$ \\ 
        $\proofsystem{SL+}{\circled{2}} \mathsf{\hl{P}{\mathsf{r}}{Q}}$ & $\semantics{r}{SL}{\circled{2}} \llbrace \mathsf{P}\rrbrace \subseteq \llbrace \mathsf{Q}\rrbrace$ & Forward & Over & Reallocation & $\times$ \\
        $\proofsystem{ISL+}{\circled{1}} \mathsf{\il{P}{\mathsf{r}}{Q}}$ & $\llbrace \mathsf{Q}\rrbrace \subseteq \semantics{r}{ISL}{\circled{1}} \llbrace \mathsf{P}\rrbrace$ & Forward & Under & No reallocation & $\checkmark$ \\
        $\proofsystem{ISL+}{\circled{2}} \mathsf{\il{P}{\mathsf{r}}{Q}}$ & $\llbrace \mathsf{Q}\rrbrace \subseteq \semantics{r}{ISL}{\circled{2}} \llbrace \mathsf{P}\rrbrace$ & Forward & Under & Reallocation & $\checkmark$ \\
        $\proofsystem{NC+}{\circled{1}} \mathsf{\nc{P}{\mathsf{r}}{Q}}$ & $\semantics{\overleftarrow{\mathsf{r}}}{SL}{\circled{1}} \llbrace \mathsf{Q}\rrbrace \subseteq \llbrace \mathsf{P}\rrbrace$ & Backward & Over & No reallocation & $\times$ \\
        $\proofsystem{NC+}{\circled{2}} \mathsf{\nc{P}{\mathsf{r}}{Q}}$ & $\semantics{\overleftarrow{\mathsf{r}}}{SL}{\circled{2}} \llbrace \mathsf{Q}\rrbrace \subseteq \llbrace \mathsf{P}\rrbrace$ & Backward & Over & Reallocation & $\times$ \\
        $\proofsystem{SIL+}{\circled{1}} \mathsf{\sil{P}{\mathsf{r}}{Q}}$ & $\llbrace \mathsf{P}\rrbrace \subseteq \semantics{\overleftarrow{\mathsf{r}}}{SL}{\circled{1}} \llbrace \mathsf{Q}\rrbrace$ & Backward & Under & No reallocation & $\times$ \\
        $\proofsystem{SIL+}{\circled{2}} \mathsf{\sil{P}{\mathsf{r}}{Q}}$ & $\llbrace \mathsf{P}\rrbrace \subseteq \semantics{\overleftarrow{\mathsf{r}}}{SL}{\circled{2}} \llbrace \mathsf{Q}\rrbrace$ & Backward & Under & Reallocation & $\times$ \\
    \end{tabular}
    \caption{Validity conditions for each of the considered proof systems, detailing their direction of the analysis, sense of approximation, memory model, and whether or not they explicitly account for errors.}
    \label{tab:proof_system_validity}
\end{table}

\clearpage
\section{Program semantics}\label{sec:program_semantics}
In this appendix we spell out the full details of the collecting semantics considered in the paper.
The subscripts $\circled{1}$ and $\circled{2}$ denote the reference memory model.
As usual, we use the subscript $\circled{i}$ whenever the definition is independent from the memory model.

\begin{figure}[h]
    \centering
        \[
            \begin{aligned}
            \semantics{b?}{SL}{\circled{i}} (s,h)&\triangleq
            \begin{cases}
              \{(s,h)\} & \text{if } \llparenthesis \mathsf{b}\rrparenthesis s = \true \\
              \emptyset & \text{otherwise}
            \end{cases}\\
            \semantics{error()}{SL}{\circled{i}} (s,h)&\triangleq \{\abort\} \\
            \semantics{x:=e}{SL}{\circled{i}} (s,h)&\triangleq \{(s[\mathsf{x} \mapsto \llparenthesis \mathsf{e}\rrparenthesis s], h)\} \\
            \semantics{x:=[y]}{SL}{\circled{i}} (s,h)&\triangleq 
            \begin{cases}
              \{(s[\mathsf{x}\mapsto h(s(\mathsf{y}))],h)\} & \text{if } h(s(\mathsf{y}))\in \mathbb{V} \\
              \{\abort\} & \text{otherwise}
            \end{cases} \\
            \semantics{[x]:=y}{SL}{\circled{i}} (s,h)&\triangleq 
            \begin{cases}
              \{(s, h[s(\mathsf{x})\mapsto s(\mathsf{y})])\} & \text{if } h(s(\mathsf{x}))\in \mathbb{V} \\
              \{\abort\} & \text{otherwise}
            \end{cases}\\
            \semantics{x:=alloc()}{SL}{\circled{1}} (s,h) &\triangleq \{(s[\mathsf{x}\mapsto l], h[l\mapsto v])\;|\;v\in \mathbb{V}, l\notin dom(h)\} \\
            \semantics{x:=alloc()}{SL}{\circled{2}} (s,h) &\triangleq \{(s[\mathsf{x}\mapsto l], h[l\mapsto v])\;|\;v\in \mathbb{V}, l\notin dom(h) \vee h(l)=\bot\} \\
            \semantics{free(x)}{SL}{\circled{1}} (s,h)&\triangleq 
            \begin{cases}
              \{(s,h')\} & \text{if } h=h'\bullet [s(\mathsf{x})\mapsto v], v \in \mathbb{V} \\
              \{\abort\} & \text{otherwise}
            \end{cases}\\
            \semantics{free(x)}{SL}{\circled{2}} (s,h)&\triangleq 
            \begin{cases}
              \{(s,h'\bullet [s(\mathsf{x})\mapsto \bot])\} & \text{if } h=h'\bullet [s(\mathsf{x})\mapsto v], v \in \mathbb{V} \\
              \{\abort\} & \text{otherwise}
            \end{cases}\\
            \semantics{r_1+r_2}{SL}{\circled{i}} (s,h)&\triangleq \semantics{r_1}{SL}{\circled{i}} (s,h) \cup \semantics{r_2}{SL}{\circled{i}} (s,h) \\
            \semantics{r_1;r_2}{SL}{\circled{i}} (s,h) &\triangleq \{(s',h')\;|\; (s'',h'')\in \semantics{r_1}{SL}{\circled{i}} (s,h), (s',h')\in \semantics{r_2}{SL}{\circled{i}} (s'',h'')\} \\
            \semantics{r^\ast}{SL}{\circled{i}} (s,h) &\triangleq \bigcup_{\mathsf{i} \in \mathbb{N}} \semantics{r^i}{SL}{\circled{i}} (s,h) \quad \text{with} \quad \mathsf{r^0} \triangleq\skipexp \quad\text{and} \quad \mathsf{r^{i+1}}\triangleq \mathsf{r;r^i} \\
        \end{aligned}
        \]
    \caption{The SL denotational semantics \cite{DBLP:conf/lics/Reynolds02}.}
    \Description{Denotational semantics of SL}
    \label{fig:sl_semantics}
\end{figure}

\begin{figure}
    \centering
        \[
            \begin{aligned}
            \semantics{b?}{ISL}{\circled{i}} \oker:(s,h)&\triangleq
            \begin{cases}
              \{\ok:(s,h)\} & \text{if } \llparenthesis \mathsf{b}\rrparenthesis s = \true \\
              \emptyset & \text{otherwise}
            \end{cases}\\
            \semantics{error()}{ISL}{\circled{i}} \oker:(s,h)&\triangleq \{\er:(s,h)\} \\
            \semantics{x:=e}{ISL}{\circled{i}} \oker:(s,h)&\triangleq \{\ok:(s[\mathsf{x} \mapsto \llparenthesis \mathsf{e}\rrparenthesis s], h)\} \\
            \semantics{x:=[y]}{ISL}{\circled{i}} \oker:(s,h)&\triangleq 
            \begin{cases}
              \{\ok:(s[\mathsf{x}\mapsto h(s(\mathsf{y}))],h)\} & \text{if } h(s(\mathsf{y}))\in \mathbb{V} \\
              \{\er:(s,h)\} & \text{otherwise}
            \end{cases} \\
            \semantics{[x]:=y}{ISL}{\circled{i}} \oker:(s,h)&\triangleq 
            \begin{cases}
              \{\ok:(s, h[s(\mathsf{x})\mapsto s(\mathsf{y})])\} & \text{if } h(s(\mathsf{x}))\in \mathbb{V} \\
              \{\er:(s,h)\} & \text{otherwise}
            \end{cases}\\
            \semantics{x:=alloc()}{ISL}{\circled{1}} \oker:(s,h) &\triangleq \{\ok:(s[\mathsf{x}\mapsto l], h[l\mapsto v])\;|\;v\in \mathbb{V}, l\notin dom(h)\} \\
            \semantics{x:=alloc()}{ISL}{\circled{2}} \oker:(s,h) &\triangleq \{\ok:(s[\mathsf{x}\mapsto l], h[l\mapsto v])\;|\;v\in \mathbb{V}, l\notin dom(h) \vee h(l)=\bot\} \\
            \semantics{free(x)}{ISL}{\circled{1}} \oker:(s,h)&\triangleq 
            \begin{cases}
              \{\ok:(s,h')\} & \text{if } h=h'\bullet [s(\mathsf{x})\mapsto v], v \in \mathbb{V} \\
              \{\er:(s,h)\} & \text{otherwise}
            \end{cases}\\
            \semantics{free(x)}{ISL}{\circled{2}} \oker:(s,h)&\triangleq 
            \begin{cases}
              \{\ok:(s,h'\bullet [s(\mathsf{x})\mapsto \bot])\} & \text{if } h=h'\bullet [s(\mathsf{x})\mapsto v], v \in \mathbb{V} \\
              \{\er:(s,h)\} & \text{otherwise}
            \end{cases}\\
            \semantics{r_1+r_2}{ISL}{\circled{i}} \oker:(s,h)&\triangleq \semantics{r_1}{ISL}{\circled{i}} \oker:(s,h) \cup \semantics{r_2}{ISL}{\circled{i}} \oker:(s,h) \\
            \semantics{r_1;r_2}{ISL}{\circled{i}} \oker:(s,h) &\triangleq 
            \{\er: (s',h')\;|\; \er: (s',h')\in \semantics{r_1}{ISL}{\circled{i}} \oker:(s,h)\} \\
            &\cup \{\oker:(s',h')\;|\; \ok:(s'',h'')\in \semantics{r_1}{ISL}{\circled{i}} \oker:(s,h), \\
            &\phantom{\cup \{\textcolor{blue}{\epsilon: (s',h')}\;|\; \;\;} \oker:(s',h')\in \semantics{r_2}{ISL}{\circled{i}} \ok: (s'',h'')\} \\
            \semantics{r^\ast}{ISL}{\circled{i}} \oker:(s,h) &\triangleq \bigcup_{\mathsf{i} \in \mathbb{N}} \semantics{r^i}{ISL}{\circled{i}} \oker:(s,h) \quad \text{with} \quad \mathsf{r^0} \triangleq\skipexp \quad\text{and} \quad \mathsf{r^{i+1}}\triangleq \mathsf{r;r^i} \\
        \end{aligned}
        \]
    \caption{The ISL denotational semantics \cite{DBLP:conf/cav/RaadBDDOV20}.}
    \Description{Denotational semantics of ISL}
    \label{fig:isl_semantics}
\end{figure}

\clearpage
\section{Proof systems}\label{sec:proof_systems}
In this appendix, we report the proof systems omitted from the main text.
In particular, Figures~\ref{fig:sl+},~\ref{fig:isl+},~\ref{fig:sil+} and~\ref{fig:nc+} contain the full set of inference rules for forward and backward analyses. We can immediately observe that proof systems for forward/backward analysis share the same axioms for atomic commands for both over- and under-approximation.

\begin{figure*}
    \centering
    \resizebox{.99\linewidth}{!}{\fbox{\renewcommand{\arraystretch}{2.5}
    \begin{tabular}{c}
        \begin{tabular}{c}
            \begin{tabular}{@{}l@{\hspace{10pt}}l@{\hspace{10pt}}l@{\hspace{10pt}}l@{}}
                \hllabel{Alloc2}{\circled{2}} & \hllabel{Free}{\circled{1}} & \hllabel{Free2}{\circled{2}}\\
                $\infer{\hl{\mathsf{\emp_{\progvar} \ast l'\not\mapsto}}{\mathsf{x:=alloc()}}{\mathsf{\emp_{\progvar\setminus \{x\}} \ast l'\mapsto z'\wedge x=l'}}}{}$ &
                $\infer{\hl{\mathsf{\emp_{\progvar}\ast x\mapsto z'}}{\mathsf{free(x)}}{\mathsf{\emp_{\progvar}}}}{}$ &
                $\infer{\hl{\mathsf{\emp_{\progvar}\ast x\mapsto z'}}{\mathsf{free(x)}}{\mathsf{\emp_{\progvar}\ast x\not\mapsto}}}{}$ 
            \end{tabular}
        \end{tabular}\\
    
        \begin{tabular}{c}
            \begin{tabular}{@{}l@{\hspace{10pt}}l@{\hspace{10pt}}l@{\hspace{10pt}}l@{}}
                \hllabel{Alloc}{\circled{i}} & \hllabel{Assume}{\circled{i}} & \hllabel{Assign}{\circled{i}} & \\
                $\infer{\hl{\mathsf{\emp_{\progvar}}}{\mathsf{x:=alloc()}}{\mathsf{\emp_{\progvar\setminus \{x\}} \ast x\mapsto z'}}}{}$ &
                $\infer{\hl{\mathsf{\emp_{\progvar}}}{\mathsf{b?}}{\mathsf{\emp_{\progvar}\wedge b}}}{}$ &
                $\infer{\hl{\mathsf{\emp_{\progvar\setminus \{x\}}\wedge x=x'}}{\mathsf{x:=e}}{\mathsf{\emp_{\progvar\setminus \{x\}}\wedge x=e[x'/x]}}}{}$ &
            \end{tabular}
        \end{tabular}\\

        \begin{tabular}{c}
            \begin{tabular}{@{}l@{\hspace{10pt}}l@{\hspace{10pt}}l@{\hspace{10pt}}l@{}}
                \hllabel{Load}{\circled{i}} & \hllabel{Store}{\circled{i}} \\
                $\infer{\hl{\mathsf{\emp_{\progvar}\ast y\mapsto z'}}{\mathsf{x:=[y]}}{\mathsf{\emp_{\progvar\setminus \{x\}} \ast y\mapsto z' \wedge x=z'}}}{}$ &
                $\infer{\hl{\mathsf{\emp_{\progvar}\ast x\mapsto z'}}{\mathsf{[x]:=y}}{\mathsf{\emp_{\progvar}\ast x\mapsto y}}}{}$ 
            \end{tabular}
        \end{tabular}\\[0.4cm]
        
        \hline 
        
        \begin{tabular}{@{}l@{\hspace{10pt}}l@{\hspace{10pt}}l@{\hspace{10pt}}l@{}}
            \hllabel{Frame}{\circled{i}} & \hllabel{Cons}{\circled{i}} & \hllabel{Iterate}{\circled{i}} & \hllabel{Exists}{\circled{i}} \\
            $\infer
                {\hl{\mathsf{P}\ast \mathsf{R}}{\mathsf{r}}{\mathsf{Q}\ast \mathsf{R}}} 
                {\mathsf{fv(R)\subseteq \logicvar} & \hl{\mathsf{P}}{\mathsf{r}}{\mathsf{Q}} & \mathsf{P\logicheapcomp R}}$ &
            $\infer
                {\hl{\mathsf{P}}{\mathsf{r}}{\mathsf{Q}}} 
                {\mathsf{P} \Rightarrow \mathsf{P}' & \hl{\mathsf{P}'}{\mathsf{r}}{\mathsf{Q}'} & \mathsf{Q}'  \Rightarrow \mathsf{Q}}$ &
            $\infer{\hl{\mathsf{P}}{\mathsf{r}^\ast}{\mathsf{P}}} 
                   {\hl{\mathsf{P}}{\mathsf{r}}{\mathsf{P}}}$ &
            $\infer{\hl{\mathsf{\exists X.P}}{\mathsf{r}}{\mathsf{\exists X.Q}}}{\hl{\mathsf{P}}{\mathsf{r}}{\mathsf{Q}}& \mathsf{X}\subseteq \logicvar}$
        \end{tabular} \\[-10pt]

        \begin{tabular}{@{}l@{\hspace{10pt}}l@{\hspace{10pt}}l@{}}
            \hllabel{Seq}{\circled{i}} & \hllabel{Choice}{\circled{i}} & \hllabel{Disj}{\circled{i}} \\
            $\infer
                {\hl{\mathsf{P}}{\mathsf{r_1;r_2}}{\mathsf{Q}}} 
                {\hl{\mathsf{P}}{\mathsf{r_1}}{\mathsf{S}} & \hl{\mathsf{S}}{\mathsf{r_2}}{\mathsf{Q}}}$ &
            $\infer{\hl{\mathsf{P}}{\mathsf{r_1+r_2}}{\mathsf{Q}}} 
                   {\forall i\in \{1,2\}. \hl{\mathsf{P}}{\mathsf{r_i}}{\mathsf{Q}}}$ &
            $\infer{\hl{\vee_{i\in I} \mathsf{P_i}}{\mathsf{r}}{\vee_{i\in I}\mathsf{Q_i}}} 
                   {\forall i\in I. \hl{\mathsf{P_i}}{\mathsf{r}}{\mathsf{Q_i}}}$
        \end{tabular} 
    \end{tabular}
     }}
     \caption{SL+: Synthesized proof system for forward correctness analysis in both memory models $\circled{i}\in \{\circled{1}, \circled{2}\}$.}
    \Description{Logical proof systems SL+}
    \label{fig:sl+}
\end{figure*}

\begin{figure*}[t]
    \centering
    \resizebox{.99\linewidth}{!}{\fbox{\renewcommand{\arraystretch}{2.5}
        \begin{tabular}{c}
            \begin{tabular}{c}
                \begin{tabular}{@{}l@{\hspace{10pt}}l@{\hspace{10pt}}l@{}}
                    \illabel{Alloc2}{\circled{2}} & \illabel{Free1}{\circled{1}} & \illabel{Free2}{\circled{2}} \\
                    $\infer{\il[blue][mygreen]{\mathsf{\emp_{\progvar}\ast l'\not\mapsto}}{\mathsf{x:=alloc}()}{\mathsf{\emp_{\progvar\setminus \{x\}}\ast l'\mapsto z'\wedge x=l'}}} {}$ &
                    $\infer{\il[blue][mygreen]{\mathsf{\emp_{\progvar} \ast x\mapsto z'}}{\mathsf{free(x)}}{\emp_{\progvar}}} {}$ &
                    $\infer{\il[blue][mygreen]{\mathsf{\emp_{\progvar} \ast x\mapsto z'}}{\mathsf{free(x)}}{\emp_{\progvar} \ast\mathsf{x\not\mapsto}}} {}$
                \end{tabular} \\[-10pt]

                \begin{tabular}{@{}l@{\hspace{10pt}}l@{\hspace{10pt}}l@{}}
                    \illabel{FreeEr3}{\circled{2}} & \illabel{LoadEr3}{\circled{2}} & \illabel{StoreEr3}{\circled{2}}  \\
                    $\infer{\il[blue][red]{\mathsf{\emp_{\progvar}\ast x\not\mapsto}}{\mathsf{free(x)}}{\mathsf{\emp_{\progvar}\ast x\not\mapsto}}} {}$ &
                    $\infer{\il[blue][red]{\mathsf{\emp_{\progvar} \ast y \not \mapsto}}{\mathsf{x:=[y]}}{\mathsf{\emp_{\progvar} \ast y \not \mapsto}}} {}$  &
                    $\infer{\il[blue][red]{\mathsf{\emp_{\progvar}\ast x \not \mapsto}}{\mathsf{[x]:=y}}{\mathsf{\emp_{\progvar}\ast x \not \mapsto}}} {}$  
                \end{tabular} \\[-10pt]
                
                \begin{tabular}{@{}l@{\hspace{10pt}}l@{\hspace{10pt}}l@{}}
                    \illabel{Assign}{\circled{i}} & \illabel{Assume}{\circled{i}} & \illabel{Error}{\circled{i}} \\
                    $\infer{\il[blue][mygreen]{\mathsf{\emp_{\progvar\setminus \{x\}} \wedge x=x'}}{\mathsf{x:=e}}{\mathsf{\emp_{\progvar\setminus\{x\}} \wedge x=e[x'/x]}}} {}$ &
                    $\infer{\il[blue][mygreen]{\mathsf{\emp_{\progvar}}}{\mathsf{b?}}{\mathsf{\emp_{\progvar}\wedge b}}} {}$ &
                    $\infer{\il[blue][red]{\mathsf{\emp_{\progvar}}}{\mathsf{error()}}{\mathsf{\emp_{\progvar}}}} {}$ 
                \end{tabular} \\[-10pt]

                \begin{tabular}{@{}l@{\hspace{10pt}}l@{\hspace{10pt}}l@{}}
                    \illabel{Alloc}{\circled{i}} & \illabel{Load}{\circled{i}} & \illabel{Store}{\circled{i}}\\
                    $\infer{\il[blue][mygreen]{\mathsf{\emp_{\progvar}}}{\mathsf{x:=alloc}()}{\mathsf{\emp_{\progvar\setminus \{x\}}\ast x\mapsto z'}}} {}$ &
                    $\infer{\il[blue][mygreen]{\emp_{\progvar} \ast\mathsf{y \mapsto z'}}{\mathsf{x:=[y]}}{\emp_{\progvar\setminus \{x\}} \ast\mathsf{y \mapsto z' \wedge x= z'}}} {}$ &
                    $\infer{\il[blue][mygreen]{\mathsf{\emp_{\progvar} \ast x \mapsto z'}}{\mathsf{[x]:=y}}{\mathsf{\emp_{\progvar} \ast x\mapsto y}}} {}$ 
                \end{tabular} \\[-10pt]

                \begin{tabular}{@{}l@{\hspace{10pt}}l@{\hspace{10pt}}l@{}}
                    \illabel{FreeEr1}{\circled{i}} & \illabel{LoadEr1}{\circled{i}} & \illabel{StoreEr1}{\circled{i}}\\
                    $\infer{\il[blue][red]{\mathsf{\emp_{\progvar}\ast x\doteq null}}{\mathsf{free(x)}}{\mathsf{\emp_{\progvar}\ast x\doteq null}}} {}$  &
                    $\infer{\il[blue][red]{\mathsf{\emp_{\progvar} \ast y\doteq null}}{\mathsf{x:=[y]}}{\mathsf{\emp_{\progvar} \ast y\doteq null}}} {}$  &
                    $\infer{\il[blue][red]{\mathsf{\emp_{\progvar}\ast x\doteq null}}{\mathsf{[x]:=y}}{\mathsf{\emp_{\progvar}\ast x\doteq null}}} {}$  
                \end{tabular} \\[-10pt]

                \begin{tabular}{@{}l@{\hspace{10pt}}l@{\hspace{10pt}}l@{}}
                    \illabel{FreeEr2}{\circled{i}} & \illabel{LoadEr2}{\circled{i}} & \illabel{StoreEr2}{\circled{i}}\\
                    $\infer{\il[blue][red]{\mathsf{\emp_{\progvar}\ast x\nnmapsto}}{\mathsf{free(x)}}{\mathsf{\emp_{\progvar}\ast x\nnmapsto}}} {}$  &
                    $\infer{\il[blue][red]{\mathsf{\emp_{\progvar} \ast y\nnmapsto}}{\mathsf{x:=[y]}}{\mathsf{\emp_{\progvar} \ast y\nnmapsto}}} {}$  &
                    $\infer{\il[blue][red]{\mathsf{\emp_{\progvar}\ast x\nnmapsto}}{\mathsf{[x]:=y}}{\mathsf{\emp_{\progvar}\ast x\nnmapsto}}} {}$  
                \end{tabular} \\
                
            \end{tabular}\\
            \hline 
            
            \begin{tabular}{@{}l@{\hspace{10pt}}l@{\hspace{10pt}}l@{\hspace{10pt}}l@{\hspace{10pt}}l@{}}
                \illabel{Frame}{\circled{i}} & \illabel{Cons}{\circled{i}} & \illabel{Iterate-zero}{\circled{i}} & \illabel{Iterate}{\circled{i}} & \illabel{Exists}{\circled{i}}\\
                $\infer
                    {\il[blue][blue]{\mathsf{P}\ast \mathsf{R}}{\mathsf{r}}{\mathsf{Q}\ast \mathsf{R}}} 
                    {\mathsf{P\logicheapcomp R} & \il[blue][blue]{\mathsf{P}}{\mathsf{r}}{\mathsf{Q}} & \mathsf{fv(R)\subseteq \logicvar}}$ &
                $\infer
                    {\il[blue][blue]{\mathsf{P}}{\mathsf{r}}{\mathsf{Q}}} 
                    {\mathsf{P}' \Rightarrow \mathsf{P} & \il[blue][blue]{\mathsf{P}'}{\mathsf{r}}{\mathsf{Q}'} & \mathsf{Q} \Rightarrow \mathsf{Q}'}$ &
                $\infer{\il[blue][blue]{\mathsf{P}}{\mathsf{r}^\ast}{\mathsf{P}}} {}$ &
                $\infer{\il[blue][blue]{\mathsf{P}}{\mathsf{r}^\ast}{\mathsf{Q}}} {\il[blue][blue]{\mathsf{P}}{\mathsf{r^\ast;r}}{\mathsf{Q}}}$ &
                $\infer{\il[blue][blue]{\mathsf{\exists \mathsf{X}.P}}{\mathsf{r}}{\mathsf{\exists \mathsf{X}.Q}}}{\il[blue][blue]{\mathsf{P}}{\mathsf{r}}{\mathsf{Q}} & \mathsf{X}\subseteq \logicvar}$
            \end{tabular} \\[-10pt]

            \begin{tabular}{@{}l@{\hspace{10pt}}l@{\hspace{10pt}}l@{\hspace{10pt}}l@{}}
                \illabel{SeqEr}{\circled{i}} & \illabel{Seq}{\circled{i}} & \illabel{Choice}{\circled{i}} & \illabel{Disj}{\circled{i}} \\
                $\infer
                    {\il[blue][red]{\mathsf{P}}{\mathsf{r_1;r_2}}{\mathsf{Q}}} 
                    {\il[blue][red]{\mathsf{P}}{\mathsf{r_1}}{\mathsf{Q}}}$ &
                $\infer
                    {\il[blue][blue]{\mathsf{P}}{\mathsf{r_1;r_2}}{\mathsf{Q}}} 
                    {\il[blue][mygreen]{\mathsf{P}}{\mathsf{r_1}}{\mathsf{S}} & \il[blue][blue]{{S}}{\mathsf{r_2}}{\mathsf{Q}}}$ &
                $\infer
                    {\il[blue][blue]{\mathsf{P}}{\mathsf{r_1+r_2}}{\mathsf{Q_1} \vee \mathsf{Q_2}}}
                    {\forall i\in \{1,2\}. \il[blue][blue]{\mathsf{P}}{\mathsf{r_i}}{\mathsf{Q_i}}}$ &
                $\infer
                    {\il[blue][blue]{\vee_{i\in I}\mathsf{P_i}}{\mathsf{r}}{\vee_{i\in I}\mathsf{Q_i}}}
                    {\forall i\in I. \il[blue][blue]{\mathsf{P_i}}{\mathsf{r}}{\mathsf{Q_i}}}$
            \end{tabular} 
        \end{tabular}
     }}       
     \caption{ISL+: Synthesized proof system for forward incorrectness analysis in both memory models $\circled{i}\in \{\circled{1}, \circled{2}\}$.}
    \Description{Logical proof systems ISL+}
    \label{fig:isl+}
\end{figure*}

\begin{figure*}[ht]
    \centering
    \resizebox{.99\linewidth}{!}{\fbox{\renewcommand{\arraystretch}{2.5}
        \begin{tabular}{c}
            \begin{tabular}{c}
                \begin{tabular}{@{}l@{\hspace{10pt}}l@{\hspace{10pt}}l@{}}
                    Par & ParEr & ParSeq \\
                    $\infer
                        {\il[blue][blue]{\mathsf{P_1\circledast P_2}}{\mathsf{r_1 \parallel r_2}}{\mathsf{Q_1 \circledast Q_2}}} 
                        {\il[blue][blue]{\mathsf{P_i}}{\mathsf{r_i}}{\mathsf{Q_i}} & \text{for all } \mathsf{i} \in\{1,2\}}$ &
                    $\infer
                        {\il[blue][red]{\mathsf{P}}{\mathsf{r_1 \parallel r_2}}{\mathsf{Q}}} 
                        {\il[blue][red]{\mathsf{P}}{\mathsf{r_i}}{\mathsf{Q}} & \text{for some } \mathsf{i} \in\{1,2\}}$ &
                    $\infer
                        {\il[blue][blue]{\mathsf{P}}{\mathsf{r_1 \parallel r_2}}{\mathsf{Q}}} 
                        {\il[blue][blue]{\mathsf{P}}{\mathsf{r_1 \parallel r_2}}{\mathsf{Q}} & \text{or} & \il[blue][blue]{\mathsf{P}}{\mathsf{r_2 \parallel r_1}}{\mathsf{Q}}}$  
                \end{tabular} \\[-10pt]

                \begin{tabular}{@{}l@{\hspace{10pt}}l@{}}
                    ParL & ParR \\
                    $\infer
                        {\il[blue][blue]{\mathsf{P}}{\mathsf{r_1 \parallel r_2}}{\mathsf{Q}}} 
                        {\mathsf{r_1=r_3;r_4} & \il[blue][mygreen]{\mathsf{P}}{\mathsf{r_3}}{\mathsf{S}} & \il[blue][blue]{\mathsf{S}}{\mathsf{r_4 \parallel r_2}}{\mathsf{Q}}}$ &
                    $\infer
                        {\il[blue][blue]{\mathsf{P}}{\mathsf{r_1 \parallel r_2}}{\mathsf{Q}}} 
                        {\mathsf{r_2=r_3;r_4} & \il[blue][mygreen]{\mathsf{P}}{\mathsf{r_3}}{\mathsf{S}} & \il[blue][blue]{\mathsf{S}}{\mathsf{r_1 \parallel r_4}}{\mathsf{Q}}}$ 
                \end{tabular}
            \end{tabular}
        \end{tabular}
     }}
     \caption{Additional rules for parallel composition in CISL+.}
     \Description{Additional rules for parallel composition in CISL+}
     \label{fig:cisl+}
\end{figure*}

\begin{figure*}[ht]
    \centering
    \resizebox{.99\linewidth}{!}{\fbox{\renewcommand{\arraystretch}{2.5}
        \begin{tabular}{c}
            \begin{tabular}{c}
                \begin{tabular}{@{}l@{\hspace{10pt}}l@{\hspace{10pt}}l@{}}
                    \sillabel{Alloc2}{\circled{2}} & \sillabel{Free1}{\circled{1}} & \sillabel{Free2}{\circled{2}}  \\
                    $\infer{\sil{\mathsf{\emp_{\mathbb{X}_\mathsf{p}\setminus \{x\}}\wedge x'=l'\ast l'\not\mapsto}}{\mathsf{x:=alloc}()}{\mathsf{\emp_{\mathbb{X}_\mathsf{p}\setminus \{x\}}\wedge x=x'\ast x\mapsto z'}}} {}$ &
                    $\infer{\sil{\mathsf{\emp_{\mathbb{X}_\mathsf{p}}\ast x\mapsto z'}}{\mathsf{free(x)}}{\mathsf{\emp_{\mathbb{X}_\mathsf{p}}}}} {}$ &
                    $\infer{\sil{\mathsf{\emp_{\mathbb{X}_\mathsf{p}}\ast x\mapsto z'}}{\mathsf{free(x)}}{\mathsf{\emp_{\mathbb{X}_\mathsf{p}} \ast x\not\mapsto}}} {}$
                \end{tabular} \\[-10pt]
            
                \begin{tabular}{@{}l@{\hspace{10pt}}l@{\hspace{10pt}}l@{}}
                    \sillabel{Alloc}{\circled{i}} & \sillabel{Assign}{\circled{i}} & \sillabel{Assume}{\circled{i}} \\
                    $\infer{\sil{\mathsf{\emp_{\mathbb{X}_\mathsf{p}\setminus \{x\}}\wedge x'=l'}}{\mathsf{x:=alloc}()}{\mathsf{\emp_{\mathbb{X}_\mathsf{p}\setminus \{x\}}\wedge x=x'\ast x\mapsto z'}}} {}$ &
                    $\infer{\sil{\mathsf{\emp_{\mathbb{X}_\mathsf{p}\setminus \{x\}} \wedge x'=e}}{\mathsf{x:=e}}{\mathsf{\emp_{\mathbb{X}_\mathsf{p}\setminus\{x\}} \wedge x=x'}}} {}$ &
                    $\infer{\sil{\mathsf{\emp_{\mathbb{X}_\mathsf{p}}\wedge b}}{\mathsf{b?}}{\mathsf{\emp_{\mathbb{X}_\mathsf{p}}}}} {}$ 
                \end{tabular} \\[-10pt]            

                \begin{tabular}{@{}l@{\hspace{10pt}}l@{\hspace{10pt}}l@{}}
                    \sillabel{Load}{\circled{i}} & \sillabel{Store}{\circled{i}}\\
                    $\infer{\sil{\mathsf{\emp_{\mathbb{X}_\mathsf{p}\setminus \{x\}} \wedge x'=z' \ast y \mapsto z'}}{\mathsf{x:=[y]}}{\mathsf{\emp_{\mathbb{X}_\mathsf{p}\setminus \{x\}} \wedge x=x' \ast y \mapsto z'}}} {}$ &
                    $\infer{\sil{\mathsf{\emp_{\mathbb{X}_\mathsf{p}}\ast x \mapsto z'}}{\mathsf{[x]:=y}}{\mathsf{\emp_{\mathbb{X}_\mathsf{p}}\ast x\mapsto y}}} {}$
                \end{tabular} \\
                
            \end{tabular}\\

            \hline 
            
            \begin{tabular}{@{}l@{\hspace{10pt}}l@{\hspace{10pt}}l@{\hspace{10pt}}l@{}}
                \sillabel{Frame}{\circled{i}} & \sillabel{Exists}{\circled{i}} & \sillabel{Disj}{\circled{i}} & \sillabel{Seq}{\circled{i}} \\
                $\infer
                    {\sil{\mathsf{P}\ast \mathsf{R}}{\mathsf{r}}{\mathsf{Q}\ast \mathsf{R}}} 
                    {\mathsf{fv(R)\subseteq\mathbb{X}_\mathsf{l}} & \sil{\mathsf{P}}{\mathsf{r}}{\mathsf{Q}} & \mathsf{Q\logicheapcomp R}}$ &
                $\infer{\sil{\mathsf{\exists \mathsf{X}.P}}{\mathsf{r}}{\mathsf{\exists \mathsf{X}.Q}}}{\sil{\mathsf{P}}{\mathsf{r}}{\mathsf{Q}} & \mathsf{X}\subseteq \mathbb{X}_\mathsf{l}}$ &
                $\infer
                    {\sil{\mathsf{\cup_{i\in I}P_i}}{\mathsf{r}}{\mathsf{\cup_{i\in I}Q_i}}} 
                    {\forall i\in I. \sil{\mathsf{P_i}}{\mathsf{r}}{\mathsf{Q_i}}}$ &
                $\infer
                    {\sil{\mathsf{P}}{\mathsf{r_1;r_2}}{\mathsf{Q}}} 
                    {\sil{\mathsf{P}}{\mathsf{r_1}}{\mathsf{S}} & \sil{{S}}{\mathsf{r_2}}{\mathsf{Q}}}$ 
            \end{tabular} \\

            \begin{tabular}{@{}l@{\hspace{10pt}}l@{\hspace{10pt}}l@{\hspace{10pt}}l@{}}
                \sillabel{Cons}{\circled{i}} & \sillabel{Iterate\text{-}zero}{\circled{i}} & \sillabel{Iterate}{\circled{i}} & \sillabel{Choice}{\circled{i}} \\
                $\infer
                    {\sil{\mathsf{P}}{\mathsf{r}}{\mathsf{Q}}} 
                    {\mathsf{P} \Rightarrow \mathsf{P}' & \sil{\mathsf{P}'}{\mathsf{r}}{\mathsf{Q}'} & \mathsf{Q}' \Rightarrow \mathsf{Q}}$ &
                $\infer{\sil{\mathsf{Q}}{\mathsf{r}^\ast}{\mathsf{Q}}} {}$ &
                $\infer{\sil{\mathsf{P}}{\mathsf{r}^\ast}{\mathsf{Q}}} {\sil{\mathsf{P}}{\mathsf{r}^\ast;\mathsf{r}}{\mathsf{Q}}}$ &
                $\infer
                    {\sil{\mathsf{P}}{\mathsf{r_1;r_2}}{\mathsf{Q}}} 
                    {\sil{\mathsf{P}}{\mathsf{r_1}}{\mathsf{S}} & \sil{{S}}{\mathsf{r_2}}{\mathsf{Q}}}$
            \end{tabular} \\
        \end{tabular}
     }}         
     \caption{SIL+: Synthesized proof system for backward incorrectness analysis in both memory models $\circled{i}\in \{\circled{1}, \circled{2}\}$.}
    \Description{Logical proof systems SIL+}
    \label{fig:sil+}
\end{figure*}

\clearpage
\section{Proofs}\label{sec:appendix_proofs}
This appendix gathers rather technical material along with the proofs of the main results in which it is applied. 
\revision{
    The results on the properties of the closure operators presented for forward analysis have been mechanized and formally verified in Isabelle/HOL. In addition, we have formalized the systematic derivation procedure described in Section~\ref{sec:systematic_program_logics}, which the completeness proof follows.
}

\subsection{Proofs of Sec.\ref{sec:systematic_program_logics}} \label{sec:systematic_program_logics_appendix}

Standard properties for the separating conjunction $\ast$ extend naturally to the operator $\sast$ when dealing with sets of states instead of assertions. In the following, we present these properties explicitly.

\begin{property}[Distributivity] \label{prop:distributivity}
    For $R\in \mathbb{F}$ and it holds $(\exists \mathsf{X}.P)\sast R = \exists \mathsf{X}.(P\sast R)$.
\end{property}
\begin{property}[Monotonicity] \label{prop:monotonicity}
    For $R\in \mathbb{F}$, if $Q'\subseteq Q$ then $Q'\sast R\subseteq Q\sast R$.
\end{property}

\setcounter{theorem}{0} 
\renewcommand{\thetheorem}{\ref{th:forward_preservation}}
\begin{theorem}[Forward Semantic Preservation] 
    Given $P,\mathsf{r}, Q$ such that $\llbracket \mathsf{r} \rrbracket P = Q$ and $R\in \mathbb{F}$ such that $P\heapcomp R$, then $(\llbracket \mathsf{r} \rrbracket P) \sast R = \llbracket \mathsf{r} \rrbracket (P\sast R) = Q\sast R$.
\end{theorem}
\begin{proof}
    We prove separately that $Q\sast R\subseteq \llbracket \mathsf{r}\rrbracket(P\sast R)$ and $\llbracket \mathsf{r}\rrbracket(P\sast R) \subseteq Q\sast R$.\\
    \textbf{Case $Q\sast R\subseteq \llbracket \mathsf{r}\rrbracket(P\sast R)$:}\\
    Suppose $(s,h)\in Q\sast R$, then by definition of $\sast$ we have:
    \begin{description}
        \item[(1)] $h=h_1\bullet h_2$,
        \item[(2)] $(s,h_1)\in Q$,
        \item[(3)] $(s,h_2) \in R$
    \end{description}
    By (2) and definition of $\llbracket \mathsf{r}\rrbracket$ we have that $\exists (s',h')\in P$ s.t. $(s,h_1)\in \llbracket \mathsf{r}\rrbracket \{(s',h')\}$. By hypothesis $R\in \mathbb{F}$, and since $s$ and $s'$ may differ only in modified variables, it implies that also $(s',h_2)\in R$. Hence, by using the hypothesis $P\heapcomp R$ we have that $h'\#h_2$, and so $(s', h'\bullet h_2)\in P\sast R$. Finally, since $(s,h_1)\in \llbracket \mathsf{r}\rrbracket\{(s',h')\}$ and $\mathsf{r}$ leaves $h_2$ untouched we have that $(s,h_1\bullet h_2)\in \llbracket \mathsf{r}\rrbracket(P\sast R)$, that by (1) $(s,h)\in \llbracket \mathsf{r}\rrbracket(P\sast R)$.\\
    \textbf{Case $\llbracket \mathsf{r}\rrbracket(P\sast R) \subseteq Q\sast R$:}\\
    Suppose $(s,h)\in \llbracket \mathsf{r}\rrbracket(P\sast R)$, then by definition of $\llbracket \mathsf{r}\rrbracket$ we have that $\exists (s',h')\in P\sast R$ s.t. $(s,h)\in \llbracket \mathsf{r}\rrbracket\{(s',h')\}$, and by definition of $\sast$, $\exists h_1,h_2$ s.t.:
    \begin{description}
        \item[(1)] $h_1\# h_2$
        \item[(2)] $h'=h_1\bullet h_2$
        \item[(3)] $(s',h_1)\in P$
        \item[(4)] $(s',h_2)\in R$
    \end{description}
    Since $(s,h) \in \llbracket \mathsf{r}\rrbracket\{(s',h_1\bullet h_2)\}$ and $R\in \mathbb{F}$, there exists $h_1'$ s.t. $(s,h_1')\in \llbracket \mathsf{r}\rrbracket\{(s',h_1)\}$ with $h=h_1'\bullet h_2$. By definition of $\llbracket \mathsf{r}\rrbracket$ and (3), we have that $(s,h_1')\in \llbracket \mathsf{r}\rrbracket P$. Finally, by (4) and the hypothesis of $R\in \mathbb{F}$, we have that $(s,h_2)\in R$, and so by hypothesis $P\heapcomp R$ we have that $h_1'\# h_2$. Hence $(s,h_1'\bullet h_2)\in Q\sast R$, which is $(s,h_1'\bullet h_2)\in Q\sast R$ since $h=h_1'\bullet h_2$.
\end{proof}

\setcounter{theorem}{0} 
\renewcommand{\thetheorem}{\ref{th:backward_preservation}}
\begin{theorem}[Backward Semantic Preservation]    
    Given $P,\mathsf{r}, Q$ such that $\llbracket \overleftarrow{\mathsf{r}} \rrbracket Q = P$ and $R\in \mathbb{F}$ such that $Q\heapcomp R$, then $(\llbracket \overleftarrow{\mathsf{r}} \rrbracket Q) \sast R = \llbracket \overleftarrow{\mathsf{r}} \rrbracket (Q\sast R) = P\sast R$. 
\end{theorem}
\begin{proof}
    Analogous to the proof of Theorem \ref{th:forward_preservation}.
\end{proof}

\begin{property} \label{prop:exists_indip}
    Given $M_1, M_2\in \mem$, $(\exists \mathsf{X}.M_1)\heapcomp M_2 \iff M_1\heapcomp M_2$
\end{property}
\begin{property} \label{prop:union_indip}
    Given $M_1 = (\cup_{i\in I}.M_i), M_2\in \mem$, $M_1\heapcomp M_2 \iff \forall i\in I. M_i\heapcomp M_2$
\end{property}

Let us define the following closures:
\begin{itemize}
    \item $\closure{}{\rightarrowoverdot} \triangleq (\closure{exists}{\leftrightleftrightarrows} \cup \closure{frame}{\forwardarrows} \cup \closure{cons}{\rightarrowoverdot})^\ast$
    \item $\closure{}{\rightarrowbelowdot} \triangleq (\closure{exists}{\leftrightleftrightarrows} \cup \closure{frame}{\forwardarrows} \cup \closure{cons}{\rightarrowbelowdot})^\ast$
    \item $\closure{}{\leftarrowoverdot} \triangleq (\closure{exists}{\leftrightleftrightarrows} \cup \closure{frame}{\backwardarrows} \cup \closure{cons}{\leftarrowoverdot})^\ast$
    \item $\closure{}{\leftarrowbelowdot} \triangleq (\closure{exists}{\leftrightleftrightarrows} \cup \closure{frame}{\backwardarrows} \cup \closure{cons}{\leftarrowbelowdot})^\ast$
\end{itemize}

In the following, we state Definition \ref{def:axioms}, Theorem \ref{th:reorder} and Corollary \ref{cor:triples} extensively for each logic under consideration. 

\setcounter{definition}{0} 
\renewcommand{\thedefinition}{\ref{def:axioms}}
\begin{definition}[Axioms] 
    A valid set of axioms for $\fwdovertripleset$/$\fwdundertripleset$/$\bwdovertripleset$/$\bwdundertripleset$ is any minimal set $\fwdoveraxiomsset$/$\fwdunderaxiomsset$/$\bwdoveraxiomsset$/$\bwdunderaxiomsset$ such that: 
    \begin{itemize}
        \item $\fwdovertripleset \triangleq \closure{}{\rightarrowoverdot} (\fwdoveraxiomsset)$
        \item $\fwdundertripleset \triangleq \closure{}{\rightarrowbelowdot} (\fwdunderaxiomsset)$
        \item $\bwdovertripleset \triangleq \closure{}{\leftarrowoverdot} (\bwdoveraxiomsset)$
        \item $\bwdundertripleset \triangleq \closure{}{\leftarrowbelowdot} (\bwdunderaxiomsset)$
    \end{itemize}    
\end{definition}

\setcounter{theorem}{0} 
\renewcommand{\thetheorem}{\ref{th:reorder}}
\begin{theorem} 
    It holds that:
    \begin{enumerate}
        \item $\closure{}{\rightarrowoverdot} = \closure{cons}{\rightarrowoverdot} \circ \closure{exists}{\leftrightleftrightarrows} \circ \closure{frame}{\forwardarrows}$
        \item $\closure{}{\rightarrowbelowdot} = \closure{cons}{\rightarrowbelowdot} \circ \closure{exists}{\leftrightleftrightarrows} \circ \closure{frame}{\forwardarrows}$
        \item $\closure{}{\leftarrowoverdot} = \closure{cons}{\leftarrowoverdot} \circ \closure{exists}{\leftrightleftrightarrows} \circ \closure{frame}{\backwardarrows}$
        \item $\closure{}{\leftarrowbelowdot} = \closure{cons}{\leftarrowbelowdot} \circ \closure{exists}{\leftrightleftrightarrows} \circ \closure{frame}{\backwardarrows}$
    \end{enumerate}
\end{theorem}
\begin{proof}
    We only prove the first two statement since the last two are analogous. 
    Moreover, we only prove the inclusion $\subseteq$ since $\supseteq$ trivially holds. 
    The proof proceeds by showing pairwise switches of consequent closures. 
    
    \textit{Over-approximation of forward semantics:}
    \begin{enumerate}
        \item $\closure{frame}{\forwardarrows} \circ \closure{cons}{\rightarrowoverdot} \subseteq \closure{cons}{\rightarrowoverdot} \circ \closure{frame}{\forwardarrows}$:
        \begin{align*}
            &\closure{frame}{\forwardarrows} \circ \closure{cons}{\rightarrowoverdot} T  
            & [\text{Def. of }\closure{cons}{\rightarrowoverdot}]\\
            =&\; \closure{frame}{\forwardarrows}\{\fwdovertriple{P}{\mathsf{c}}{Q'}\;|\;\fwdovertriple{P}{\mathsf{c}}{Q}\in T, Q\subseteq Q'\}
            & [\text{Def. of }\closure{frame}{\forwardarrows}]\\
            =&\; \{\fwdovertriple{P\sast R}{\mathsf{c}}{Q'\sast R}\;|\;\fwdovertriple{P}{\mathsf{c}}{Q}\in T, Q\subseteq Q', R \in \mathbb{F}, P\heapcomp R\} 
            & [\text{Prop.} \ref{prop:monotonicity}]\\
            \subseteq&\; \{\fwdovertriple{P\sast R}{\mathsf{c}}{Q'\sast R} \;|\; \fwdovertriple{P}{\mathsf{c}}{Q}\in T, Q \sast R\subseteq Q'\sast R, R\in \mathbb{F}, P\heapcomp R\} 
            & [\text{Def. of }\closure{cons}{\rightarrowoverdot}]\\
            \subseteq&\; \closure{cons}{\rightarrowoverdot}\{\fwdovertriple{P\sast R}{\mathsf{c}}{Q\sast R}\;|\;\fwdovertriple{P}{\mathsf{c}}{Q}\in T, R \in \mathbb{F}, P\heapcomp R\}
            & [\text{Def. of }\closure{frame}{\forwardarrows}]\\
            =&\; \closure{cons}{\rightarrowoverdot} \circ \closure{frame}{\forwardarrows} T
        \end{align*}
        \item $\closure{frame}{\forwardarrows} \circ \closure{exists}{\leftrightleftrightarrows} \subseteq \closure{exists}{\leftrightleftrightarrows} \circ \closure{frame}{\forwardarrows}$:
        \begin{align*}
            & \closure{frame}{\forwardarrows} \circ \closure{exists}{\leftrightleftrightarrows} T
            & [\text{Def. of }\closure{exists}{\leftrightleftrightarrows}]\\
            =&\; \closure{frame}{\forwardarrows}\{\fwdovertriple{\exists \mathsf{X}.P}{\mathsf{c}}{\exists \mathsf{X}.Q}\;|\; \fwdovertriple{P}{\mathsf{c}}{Q}\in T, \mathsf{X} \subseteq \mathbb{X}_\mathsf{l}\} 
            & [\text{Def. of }\closure{frame}{\forwardarrows}]\\
            =&\; \{\fwdovertriple{(\exists \mathsf{X}.P)\sast R}{\mathsf{c}}{(\exists \mathsf{X}.Q)\sast R}\;|\;\fwdovertriple{P}{\mathsf{c}}{Q}\in T, \mathsf{X} \subseteq \mathbb{X}_\mathsf{l}, R\in \mathbb{F}, (\exists \mathsf{X}.P)\heapcomp R\} 
            & [\text{Prop.} \ref{prop:exists_indip}] \\
            =&\; \{\fwdovertriple{(\exists \mathsf{X}.P)\sast R}{\mathsf{c}}{(\exists \mathsf{X}.Q)\sast R}\;|\;\fwdovertriple{P}{\mathsf{c}}{Q}\in T, \mathsf{X} \subseteq \mathbb{X}_\mathsf{l}, R\in \mathbb{F}, P \heapcomp R\}
            & [\text{Prop.} \ref{prop:distributivity}] \\
            =&\; \{\fwdovertriple{\exists \mathsf{X}.(P\sast R)}{\mathsf{c}}{\exists \mathsf{X}.(Q\sast R)}\;|\;\fwdovertriple{P}{\mathsf{c}}{Q}\in T, \mathsf{X} \subseteq \mathbb{X}_\mathsf{l}, R\in \mathbb{F}, P \heapcomp R\}
            & [\text{Def. of }\closure{exists}{\leftrightleftrightarrows}] \\
            =&\; \closure{exists}{\leftrightleftrightarrows} \{\fwdovertriple{P\sast R}{\mathsf{c}}{Q\sast R}\;|\; \fwdovertriple{P}{\mathsf{c}}{Q}\in T, R\in \mathbb{F}, P\heapcomp R\}
            & [\text{Def. of }\closure{frame}{\forwardarrows}] \\
            =&\;\closure{exists}{\leftrightleftrightarrows} \circ \closure{frame}{\forwardarrows} T
        \end{align*}
        \item $\closure{exists}{\leftrightleftrightarrows} \circ \closure{cons}{\rightarrowoverdot} \subseteq \closure{cons}{\rightarrowoverdot} \circ \closure{exists}{\leftrightleftrightarrows}$: 
        \begin{align*}
            &\closure{exists}{\leftrightleftrightarrows} \circ \closure{cons}{\rightarrowoverdot} T  
            & [\text{Def. of }\closure{cons}{\rightarrowoverdot}]\\
            =&\; \closure{exists}{\leftrightleftrightarrows} \{\fwdovertriple{P}{\mathsf{c}}{Q'}\;|\;\fwdovertriple{P}{\mathsf{c}}{Q}\in T, Q\subseteq Q'\}
            & [\text{Def. of }\closure{exists}{\leftrightleftrightarrows}]\\
            =&\; \{\fwdovertriple{\exists \mathsf{X}. P}{\mathsf{c}}{\exists \mathsf{X}. Q'}\;|\;\fwdovertriple{P}{\mathsf{c}}{Q}\in T, Q\subseteq Q', \mathsf{X}\subseteq \mathbb{X}_\mathsf{l}\}
            & [\exists \text{ monotone}]\\
            \subseteq&\; \{\fwdovertriple{\exists \mathsf{X}. P}{\mathsf{c}}{\exists \mathsf{X}. Q'}\;|\;\fwdovertriple{P}{\mathsf{c}}{Q}\in T, \exists \mathsf{X}. Q\subseteq \exists \mathsf{X}. Q', \mathsf{X}\subseteq \mathbb{X}_\mathsf{l}\}
            & [\text{Def. of }\closure{cons}{\rightarrowoverdot}]\\
            \subseteq&\; \closure{cons}{\rightarrowoverdot}(\{\fwdovertriple{\exists \mathsf{X}. P}{\mathsf{c}}{\exists \mathsf{X}. Q}\;|\;\fwdovertriple{P}{\mathsf{c}}{Q}\in T, \mathsf{X}\subseteq \mathbb{X}_\mathsf{l}\})
            & [\text{Def. of }\closure{exists}{\leftrightleftrightarrows}]\\
            =&\; \closure{cons}{\rightarrowoverdot} \circ \closure{exists}{\leftrightleftrightarrows} T
        \end{align*}
    \end{enumerate}
    
    \textit{Under-approximation of forward semantics:}
    \begin{enumerate}
        \item $\closure{frame}{\forwardarrows} \circ \closure{cons}{\rightarrowbelowdot}\subseteq \closure{cons}{\rightarrowbelowdot} \circ \closure{frame}{\forwardarrows}$:
        \begin{align*}
            &\closure{frame}{\forwardarrows} \circ \closure{cons}{\rightarrowbelowdot} T 
            & [\text{Def. of }\closure{cons}{\rightarrowbelowdot}]\\
            =&\; \closure{frame}{\forwardarrows}\{\fwdundertriple{P}{\mathsf{c}}{Q'}\;|\;\fwdundertriple{P}{\mathsf{c}}{Q}\in T, Q'\subseteq Q\} 
            & [\text{Def. of }\closure{frame}{\forwardarrows}]\\
            =&\; \{\fwdundertriple{P\sast R}{\mathsf{c}}{Q'\sast R}\;|\;\fwdundertriple{P}{\mathsf{c}}{Q}\in T, Q'\subseteq Q, R \in \mathbb{F}, P\heapcomp R\} 
            & [\text{Prop.} \ref{prop:monotonicity}]\\
            \subseteq&\; \{\fwdundertriple{P\sast R}{\mathsf{c}}{Q'\sast R} \;|\; \fwdundertriple{P}{\mathsf{c}}{Q}\in T, Q' \sast R\subseteq Q\sast R, R\in \mathbb{F}, P\heapcomp R\} 
            & [\text{Def. of }\closure{cons}{\rightarrowbelowdot}]\\
            \subseteq&\; \closure{cons}{\rightarrowbelowdot}\{\fwdundertriple{P\sast R}{\mathsf{c}}{Q\sast R}\;|\;\fwdundertriple{P}{\mathsf{c}}{Q}\in T, R \in \mathbb{F}, P\heapcomp R\}
            & [\text{Def. of }\closure{frame}{\forwardarrows}]\\
            =&\; (\closure{cons}{\rightarrowbelowdot} \circ \closure{frame}{\forwardarrows}) T
            & 
        \end{align*}
        \item $\closure{exists}{\leftrightleftrightarrows} \circ \closure{frame}{\forwardarrows} = \closure{frame}{\forwardarrows} \circ \closure{exists}{\leftrightleftrightarrows}$:
        \begin{align*}
            &\closure{frame}{\forwardarrows} \circ \closure{exists}{\leftrightleftrightarrows} T
            & [\text{Def. of }\closure{exists}{\leftrightleftrightarrows}]\\
            =&\; \closure{frame}{\forwardarrows}\{\fwdundertriple{\exists \mathsf{X}.P}{\mathsf{c}}{\exists \mathsf{X}.Q}\;|\; \fwdundertriple{P}{\mathsf{c}}{Q}\in T, \mathsf{X} \subseteq \mathbb{X}_\mathsf{l}\}
            & [\text{Def. of }\closure{frame}{\forwardarrows}]\\
            =&\; \{\fwdundertriple{(\exists \mathsf{X}.P)\sast R}{\mathsf{c}}{(\exists \mathsf{X}.Q)\sast R}\;|\;\fwdundertriple{P}{\mathsf{c}}{Q}\in T, \mathsf{X} \subseteq \mathbb{X}_\mathsf{l}, R\in \mathbb{F}, P\heapcomp R\} 
            & [\text{Prop.} \ref{prop:distributivity}]\\
            =&\; \{\fwdundertriple{\exists \mathsf{X}.(P\sast R)}{\mathsf{c}}{\exists \mathsf{X}.(Q\sast R)}\;|\;\fwdundertriple{P}{\mathsf{c}}{Q}\in T, \mathsf{X} \subseteq \mathbb{X}_\mathsf{l}, R\in \mathbb{F}, P\heapcomp R \}
            & [\text{Def. of }\closure{exists}{\leftrightleftrightarrows}]\\
            =&\; \closure{exists}{\leftrightleftrightarrows}\{\fwdundertriple{P\sast R}{\mathsf{c}}{Q\sast R}\;|\;\fwdundertriple{P}{\mathsf{c}}{Q}\in T, R\in \mathbb{F}, P\heapcomp R\}
            & [\text{Def. of }\closure{frame}{\forwardarrows}]\\
            =&\; \closure{exists}{\leftrightleftrightarrows} \circ \closure{frame}{\forwardarrows} T 
            & 
        \end{align*}
        \item $\closure{exists}{\leftrightleftrightarrows} \circ \closure{cons}{\rightarrowbelowdot} \subseteq \closure{cons}{\rightarrowbelowdot} \circ \closure{exists}{\leftrightleftrightarrows}$:
        \begin{align*}
            &(\closure{exists}{\leftrightleftrightarrows} \circ \closure{cons}{\rightarrowbelowdot}) T 
            & [\text{Def. of }\closure{cons}{\rightarrowbelowdot}]\\
            =&\; \closure{exists}{\leftrightleftrightarrows} (\{\fwdundertriple{P}{\mathsf{c}}{Q'}\;|\;\fwdundertriple{P}{\mathsf{c}}{Q}\in T \wedge Q'\subseteq Q\})
            & [\text{Def. of }\closure{exists}{\leftrightleftrightarrows}]\\
            =&\; \{\fwdundertriple{\exists \mathsf{X}. P}{\mathsf{c}}{\exists \mathsf{X}. Q'}\;|\;\fwdundertriple{P}{\mathsf{c}}{Q}\in T \wedge Q'\subseteq Q \wedge \mathsf{X}\subseteq \mathbb{X}_\mathsf{l}\}
            & [\exists \text{ monotone}]\\
            \subseteq&\; \{\fwdundertriple{\exists \mathsf{X}. P}{\mathsf{c}}{\exists \mathsf{X}. Q'}\;|\;\fwdundertriple{P}{\mathsf{c}}{Q}\in T \wedge \exists \mathsf{X}. Q'\subseteq \exists \mathsf{X}. Q \wedge \mathsf{X}\subseteq \mathbb{X}_\mathsf{l}\}
            & [\text{Def. of }\closure{cons}{\rightarrowbelowdot}]\\
            \subseteq&\; \closure{cons}{\rightarrowbelowdot}(\{\fwdundertriple{\exists \mathsf{X}. P}{\mathsf{c}}{\exists \mathsf{X}. Q}\;|\;\fwdundertriple{P}{\mathsf{c}}{Q}\in T \wedge \mathsf{X}\subseteq \mathbb{X}_\mathsf{l}\})
            & [\text{Def. of }\closure{exists}{\leftrightleftrightarrows}]\\
            =&\; (\closure{cons}{\rightarrowbelowdot} \circ \closure{exists}{\leftrightleftrightarrows}) T
            & 
        \end{align*}
    \end{enumerate}
\end{proof}

\setcounter{theorem}{0} 
\renewcommand{\thetheorem}{\ref{cor:triples}}
\begin{corollary}[Axioms] 
    A valid set of axioms $\fwdoveraxiomsset$/$\fwdunderaxiomsset$/$\bwdoveraxiomsset$/$\bwdunderaxiomsset$ is any minimal set  $\fwdovertripleset$/$\fwdundertripleset$/$\bwdovertripleset$/$\bwdundertripleset$ such that:
    \begin{itemize}
        \item $\fwdovertripleset = (\closure{cons}{\rightarrowoverdot} \circ \closure{exists}{\leftrightleftrightarrows} \circ \closure{frame}{\forwardarrows}) \fwdoveraxiomsset$
        \item $\fwdundertripleset = (\closure{cons}{\rightarrowbelowdot} \circ \closure{exists}{\leftrightleftrightarrows} \circ \closure{frame}{\forwardarrows}) \fwdunderaxiomsset$
        \item $\bwdovertripleset = (\closure{cons}{\leftarrowoverdot} \circ \closure{exists}{\leftrightleftrightarrows} \circ \closure{frame}{\backwardarrows}) \bwdoveraxiomsset$
        \item $\bwdundertripleset = (\closure{cons}{\leftarrowbelowdot} \circ \closure{exists}{\leftrightleftrightarrows} \circ \closure{frame}{\backwardarrows}) \bwdunderaxiomsset$
    \end{itemize}    
\end{corollary}

\setcounter{theorem}{0}
\renewcommand{\thetheorem}{\ref{th:sound_complete_semantic}}
\begin{theorem}[Soundness and completeness]
    Proof systems in Fig.\ref{fig:semantic_proof_systems} are sound and complete:
    \begin{enumerate}
        \item $\Vdash\hl{P}{\mathsf{r}}{Q} \iff \llbracket \mathsf{r} \rrbracket P \subseteq Q$
        \item $\Vdash\il{P}{\mathsf{r}}{Q} \iff Q \subseteq \llbracket \mathsf{r} \rrbracket P$ 
        \item $\Vdash\nc{P}{\mathsf{r}}{Q} \iff \llbracket \overleftarrow{\mathsf{r}} \rrbracket Q \subseteq P$
        \item $\Vdash\sil{P}{\mathsf{r}}{Q} \iff P\subseteq \llbracket \overleftarrow{\mathsf{r}} \rrbracket Q$
    \end{enumerate} 
\end{theorem}
\begin{proof}
    Soundness can be proved by induction on the derivation tree. 
    Completeness of atomic commands follows directly from the definitions of semantic axioms
    Furthermore, completeness of compositions for Fwd-over, Fwd-under and Bwd-under can be proved respectively as in \cite{DBLP:journals/siamcomp/Cook78}, \cite{DBLP:journals/pacmpl/OHearn20}, and \cite{ DBLP:journals/pacmpl/AscariBGL25}.
    In the following, we show the completeness of the composition of atomic commands for Bwd-over. 
    Particularly, we show that for any $Q$, we can derive the triple $\nc{\llbracket \overleftarrow{\mathsf{r}}\rrbracket Q}{\mathsf{r}}{Q}$. Completeness will follow by applying the rule of consequence. We proceed by induction on the structure of $\mathsf{r}$:
    \\\textbf{Case $\mathsf{r\triangleq r_1;r_2}$:}\\
    We show a derivation for $\nc{\llbracket \overleftarrow{\mathsf{r_1;r_2}}\rrbracket Q}{\mathsf{r_1;r_2}}{Q}$. 
    By property of backward semantics \cite[Lemma 3.1]{DBLP:journals/pacmpl/AscariBGL25} we know that $\llbracket \overleftarrow{\mathsf{r_1;r_2}}\rrbracket Q = \llbracket \overleftarrow{\mathsf{r_1}}\rrbracket\llbracket \overleftarrow{\mathsf{r_2}}\rrbracket Q$. Hence, by inductive hypothesis we assume that the triples $\nc{\llbracket \overleftarrow{\mathsf{r_1}}\rrbracket\llbracket \overleftarrow{\mathsf{r_2}}\rrbracket Q}{\mathsf{r_1}}{\llbracket \overleftarrow{\mathsf{r_2}}\rrbracket Q}$ and $\nc{\llbracket \overleftarrow{\mathsf{r_2}}\rrbracket Q}{\mathsf{r_2}}{Q}$ are derivable, then by applying the rule of seq we can derive:
    \[
    \infer[\mathsf{(Seq)}]
        {\nc{\llbracket \overleftarrow{\mathsf{r_1;r_2}}\rrbracket Q}{\mathsf{r_1;r_2}}{Q}}
        {\nc{\llbracket \overleftarrow{\mathsf{r_1}}\rrbracket\llbracket \overleftarrow{\mathsf{r_2}}\rrbracket Q}{\mathsf{r_1}}{\llbracket \overleftarrow{\mathsf{r_2}}\rrbracket Q}&\nc{\llbracket \overleftarrow{\mathsf{r_2}}\rrbracket Q}{\mathsf{r_2}}{Q}}
    \]
    \textbf{Case $\mathsf{r\triangleq r_1+r_2}$:}\\
    We show a derivation for $\nc{\llbracket \overleftarrow{\mathsf{r_1+r_2}}\rrbracket Q}{\mathsf{r_1+r_2}}{Q}$. By property of backward semantics \cite[Lemma 3.1]{DBLP:journals/pacmpl/AscariBGL25} we know that $\llbracket \overleftarrow{\mathsf{r_1+r_2}}\rrbracket Q = \llbracket \overleftarrow{\mathsf{r_1}}\rrbracket Q \cup \llbracket \overleftarrow{\mathsf{r_2}}\rrbracket Q$. Hence, by inductive hypothesis we assume that the triples $\nc{\llbracket \overleftarrow{\mathsf{r_1}}\rrbracket Q}{\mathsf{r_1}}{Q}$ and $\nc{\llbracket \overleftarrow{\mathsf{r_2}}\rrbracket Q}{\mathsf{r_2}}{Q}$ are derivable, then by applying the rule of choice we can derive:
    \[
    \infer[\mathsf{(Choice)}]
        {\nc{\llbracket \overleftarrow{\mathsf{r_1+r_2}}\rrbracket Q}{\mathsf{r_1+r_2}}{Q}}
        {\infer[\mathsf{(Cons)}]
            {\nc{\llbracket \overleftarrow{\mathsf{r_1}}\rrbracket Q \cup \llbracket \overleftarrow{\mathsf{r_2}}\rrbracket Q}{\mathsf{r_1}}{Q}}
            {\nc{\llbracket \overleftarrow{\mathsf{r_1}}\rrbracket Q}{\mathsf{r_1}}{Q}} 
        & \infer[\mathsf{(Cons)}]
            {\nc{\llbracket \overleftarrow{\mathsf{r_1}}\rrbracket Q \cup \llbracket \overleftarrow{\mathsf{r_2}}\rrbracket Q}{\mathsf{r_2}}{Q}}
            {\nc{\llbracket \overleftarrow{\mathsf{r_2}}\rrbracket Q}{\mathsf{r_2}}{Q}}}
    \]
    \textbf{Case $\mathsf{r\triangleq r_1^\ast}$:}\\
    We show a derivation for $\nc{\llbracket \overleftarrow{\mathsf{r_1^\ast}}\rrbracket Q}{\mathsf{r_1^\ast}}{Q}$. 
    We set $I \triangleq\llbracket \overleftarrow{\mathsf{r_1^\ast}}\rrbracket Q$. By property of backward semantics \cite[Lemma 3.1]{DBLP:journals/pacmpl/AscariBGL25} we know that $\llbracket \overleftarrow{\mathsf{r_1^\ast}}\rrbracket Q = \cup_{n\geq 0} \llbracket \overleftarrow{\mathsf{r_1}}\rrbracket^n Q$ and then $\llbracket \overleftarrow{\mathsf{r_1}}\rrbracket I = \llbracket \overleftarrow{\mathsf{r_1}} \rrbracket (\cup_{n\geq 0} \llbracket \overleftarrow{\mathsf{r_1}}\rrbracket^n Q) = \cup_{n\geq 0} \llbracket \overleftarrow{\mathsf{r_1}}\rrbracket^{n+1}Q = \cup_{n\geq 1} \llbracket \overleftarrow{\mathsf{r_1}}\rrbracket^{n}Q \subseteq \cup_{n\geq 0} \llbracket \overleftarrow{\mathsf{r_1}}\rrbracket^{n}Q = I$. Hence, by inductive hypothesis we assume that the triple $\nc{I}{\mathsf{r_1}}{I}$ is derivable.
    Furthermore, we have that $Q = \llbracket \mathsf{r_1}\rrbracket^0 Q \subseteq \llbracket \mathsf{r_1}^\ast\rrbracket Q = I$, then by applying rule of Iter and Cons we can derive:
    \[
    \infer[\mathsf{(Cons)}]
        {\nc{I}{\mathsf{r_1^\ast}}{Q}}
        {\infer[\mathsf{(Iter)}]
            {\nc{I}{\mathsf{r_1^\ast}}{I}}
            {\nc{I}{\mathsf{r_1}}{I}}
            & Q\subseteq I
        }
    \]
    That by definition of $I$ is analogous to the triple $\nc{\llbracket \overleftarrow{\mathsf{r_1^\ast}}\rrbracket Q}{\mathsf{r_1^\ast}}{Q}$.
\end{proof}

In the following, we state Theorem \ref{th:disj_closure} and Theorem \ref{th:disj_no_disj} extensively for each logic under consideration. 

\setcounter{theorem}{0} 
\renewcommand{\thetheorem}{\ref{th:disj_closure}}
\begin{theorem}
    It holds that:
    \begin{enumerate}
        \item $(\closure{cons}{\rightarrowoverdot} \cup \closure{disj}{\leftrightleftrightarrows} \cup \closure{exists}{\leftrightleftrightarrows} \cup \closure{frame}{\forwardarrows})^\ast = \closure{cons}{\rightarrowoverdot} \circ \closure{disj}{\leftrightleftrightarrows} \circ \closure{exists}{\leftrightleftrightarrows} \circ \closure{frame}{\forwardarrows}$
        \item $(\closure{cons}{\rightarrowbelowdot} \cup \closure{disj}{\leftrightleftrightarrows} \cup \closure{exists}{\leftrightleftrightarrows} \cup \closure{frame}{\forwardarrows})^\ast = \closure{cons}{\rightarrowbelowdot} \circ \closure{disj}{\leftrightleftrightarrows} \circ \closure{exists}{\leftrightleftrightarrows} \circ \closure{frame}{\forwardarrows}$
        \item $(\closure{cons}{\leftarrowoverdot} \cup \closure{disj}{\leftrightleftrightarrows} \cup \closure{exists}{\leftrightleftrightarrows} \cup \closure{frame}{\backwardarrows})^\ast = \closure{cons}{\leftarrowoverdot} \circ \closure{disj}{\leftrightleftrightarrows} \circ \closure{exists}{\leftrightleftrightarrows} \circ \closure{frame}{\backwardarrows}$
        \item $(\closure{cons}{\leftarrowbelowdot} \cup \closure{disj}{\leftrightleftrightarrows} \cup \closure{exists}{\leftrightleftrightarrows} \cup \closure{frame}{\backwardarrows})^\ast = \closure{cons}{\leftarrowbelowdot} \circ \closure{disj}{\leftrightleftrightarrows} \circ \closure{exists}{\leftrightleftrightarrows} \circ \closure{frame}{\backwardarrows}$
    \end{enumerate}
\end{theorem}
\begin{proof}
    We only prove the first statement, as the others are analogous. 
    Moreover, we only prove the inclusion $\subseteq$ since $\supseteq$ trivially holds. 
    The proof proceeds by showing pairwise switches of consequent closures, avoiding cases already proved in the proof of Theorem \ref{th:reorder}. 
    \begin{enumerate}
        \item $\closure{exists}{\leftrightleftrightarrows} \circ \closure{disj}{\leftrightleftrightarrows} = \closure{disj}{\leftrightleftrightarrows} \circ \closure{exists}{\leftrightleftrightarrows}$:
        \begin{align*}
            &\closure{exists}{\leftrightleftrightarrows} \circ \closure{disj}{\leftrightleftrightarrows} T \\
            =&\; \closure{exists}{\leftrightleftrightarrows}\{\fwdovertriple{\cup_{i\in I}P_i}{\mathsf{c}}{\cup_{i\in I}Q_i}\;|\; \forall i\in I. \fwdovertriple{P_i}{\mathsf{c}}{Q_i}\in T\}
            & [\text{Def. of }\closure{disj}{\leftrightleftrightarrows}]\\
            =&\; \{\fwdovertriple{\exists \mathsf{X}.\cup_{i\in I}P_i}{\mathsf{c}}{\exists \mathsf{X}.\cup_{i\in I}Q_i}\;|\; \forall i\in I. \fwdovertriple{P_i}{\mathsf{c}}{Q_i}\in T, \mathsf{X}\subseteq \mathbb{X}_\mathsf{l}\} 
            & [\text{Def. of }\closure{exists}{\leftrightleftrightarrows}]\\
            =&\; \{\fwdovertriple{\cup_{i\in I}\exists \mathsf{X}.P_i}{\mathsf{c}}{\cup_{i\in I}\exists \mathsf{X}.Q_i}\;|\; \forall i\in I. \fwdovertriple{P_i}{\mathsf{c}}{Q_i}\in T, \mathsf{X}\subseteq \mathbb{X}_\mathsf{l}\} 
            & [\text{Distributivity of $\cup$ over $\exists$}]\\
            =&\; \closure{exists}{\leftrightleftrightarrows}\{\fwdovertriple{\exists \mathsf{X}.P_i}{\mathsf{c}}{\exists \mathsf{X}.Q_i}\;|\; \fwdovertriple{P_i}{\mathsf{c}}{Q_i}\in T, \mathsf{X}\subseteq \mathbb{X}_\mathsf{l}\}
            & [\text{Def. of }\closure{disj}{\leftrightleftrightarrows}]\\
            =&\; \closure{exists}{\leftrightleftrightarrows} \circ \closure{disj}{\leftrightleftrightarrows} T 
            & [\text{Def. of }\closure{exists}{\leftrightleftrightarrows}]
        \end{align*}
        \item $\closure{frame}{\forwardarrows} \circ \closure{disj}{\leftrightleftrightarrows} = \closure{disj}{\leftrightleftrightarrows} \circ \closure{frame}{\forwardarrows}$:
        \begin{align*}
            &\closure{frame}{\forwardarrows} \circ \closure{disj}{\leftrightleftrightarrows} T \\
            =&\; \closure{frame}{\forwardarrows}\{\fwdovertriple{\cup_{i\in I}P_i}{\mathsf{c}}{\cup_{i\in I}Q_i}\;|\; \forall i\in I. \fwdovertriple{P_i}{\mathsf{c}}{Q_i}\in T\} 
            & [\text{Def. of }\closure{disj}{\leftrightleftrightarrows}]\\
            =&\; \{\fwdovertriple{(\cup_{i\in I}P_i)\sast R}{\mathsf{c}}{(\cup_{i\in I}Q_i)\sast R}\;|\; \forall i\in I. \fwdovertriple{P_i}{\mathsf{c}}{Q_i}\in T, R\in \mathbb{F}, (\cup_{i\in I}P_i)\heapcomp R\} 
            & [\text{Def. of }\closure{frame}{\forwardarrows}]\\
            =&\; \{\fwdovertriple{\cup_{i\in I}(P_i\sast R)}{\mathsf{c}}{\cup_{i\in I}(Q_i\sast R)}\;|\; \forall i\in I. \fwdovertriple{P_i}{\mathsf{c}}{Q_i}\in T, R\in \mathbb{F}, (\cup_{i\in I}P_i)\heapcomp R\}  
            & [\text{Distributivity}]\\
            =&\; \{\fwdovertriple{\cup_{i\in I}(P_i\sast R)}{\mathsf{c}}{\cup_{i\in I}(Q_i\sast R)}\;|\; \forall i\in I. \fwdovertriple{P_i}{\mathsf{c}}{Q_i}\in T, R\in \mathbb{F}, \forall i \in I. P_i\heapcomp R\}  
            & [\text{Prop.} \ref{prop:union_indip}]\\
            =&\; \closure{disj}{\leftrightleftrightarrows}\{\fwdovertriple{P_i\sast R}{\mathsf{c}}{Q_i\sast R}\;|\; \fwdovertriple{P_i}{\mathsf{c}}{Q_i}\in T, R\in \mathbb{F}, \forall i\in I. P_i\heapcomp R\}
            & [\text{Def. of }\closure{disj}{\leftrightleftrightarrows}]\\
            \subseteq &\; \closure{disj}{\leftrightleftrightarrows} \circ \closure{frame}{\forwardarrows} T 
            & [\text{Def. of }\closure{frame}{\forwardarrows}]
        \end{align*}
        \item $\closure{disj}{\leftrightleftrightarrows} \circ \closure{cons}{\rightarrowoverdot} = \closure{cons}{\rightarrowoverdot} \circ \closure{disj}{\leftrightleftrightarrows}$:
        \begin{align*}
            &\closure{disj}{\leftrightleftrightarrows} \circ \closure{cons}{\rightarrowoverdot} T \\
            =&\; \closure{disj}{\leftrightleftrightarrows}\{\fwdovertriple{P}{\mathsf{c}}{Q'}\;|\; \fwdovertriple{P}{\mathsf{c}}{Q}\in T, Q\subseteq Q'\} 
            & [\text{Def. of }\closure{cons}{\rightarrowoverdot}] \\
            =&\; \{\fwdovertriple{\cup_{i\in I}P_i}{\mathsf{c}}{\cup_{i\in I}Q'_i}\;|\; \forall i\in I.\fwdovertriple{P_i}{\mathsf{c}}{Q_i}\in T, Q_i\subseteq Q'_i\}
            & [\text{Def. of }\closure{disj}{\leftrightleftrightarrows}] \\
            \subseteq&\; \{\fwdovertriple{\cup_{i\in I}P_i}{\mathsf{c}}{\cup_{i\in I}Q'_i}\;|\; \forall i\in I.\fwdovertriple{P_i}{\mathsf{c}}{Q_i}\in T, \cup_{i\in I} Q_i\subseteq \cup_{i\in I} Q'_i\}
            & [\cup \text{ monotone}] \\
            \subseteq&\; \closure{cons}{\rightarrowoverdot}\{\fwdovertriple{\cup_{i\in I}P_i}{\mathsf{c}}{\cup_{i\in I}Q_i}\;|\; \forall i\in I.\fwdovertriple{P_i}{\mathsf{c}}{Q_i}\in T\}
            & [\text{Def. of }\closure{cons}{\rightarrowoverdot}] \\
            =&\; \closure{cons}{\rightarrowoverdot} \circ \closure{disj}{\leftrightleftrightarrows} T
            & [\text{Def. of }\closure{disj}{\leftrightleftrightarrows}]
        \end{align*}
    \end{enumerate}
\end{proof}

\setcounter{theorem}{0} 
\renewcommand{\thetheorem}{\ref{th:disj_no_disj}}
\begin{theorem}
    For every index set $I$,
    \begin{enumerate}
        \item \resizebox{\linewidth}{!}{$\closure{disj}{\leftrightleftrightarrows} \circ \closure{exists}{\leftrightleftrightarrows} \circ \closure{frame}{\forwardarrows}\{\fwdovertriple{P_i}{\mathsf{c}}{Q_i})\;|\;i\in I\} = 
        \closure{exists}{\leftrightleftrightarrows} \circ \closure{frame}{\forwardarrows}\{\fwdovertriple{\cup_{i\in I} (P_i\cap d'=i)}{\mathsf{c}}{\cup_{i\in I} (Q_i\cap d'=i)})\}$}
        \item \resizebox{\linewidth}{!}{$\closure{disj}{\leftrightleftrightarrows} \circ \closure{exists}{\leftrightleftrightarrows} \circ \closure{frame}{\forwardarrows}\{\fwdundertriple{P_i}{\mathsf{c}}{Q_i})\;|\;i\in I\} = 
        \closure{exists}{\leftrightleftrightarrows} \circ \closure{frame}{\forwardarrows}\{\fwdundertriple{\cup_{i\in I} (P_i\cap d'=i)}{\mathsf{c}}{\cup_{i\in I} (Q_i\cap d'=i)})\}$}
        \item \resizebox{\linewidth}{!}{$\closure{disj}{\leftrightleftrightarrows} \circ \closure{exists}{\leftrightleftrightarrows} \circ \closure{frame}{\backwardarrows}\{\bwdovertriple{P_i}{\mathsf{c}}{Q_i})\;|\;i\in I\} = 
        \closure{exists}{\leftrightleftrightarrows} \circ \closure{frame}{\backwardarrows}\{\bwdovertriple{\cup_{i\in I} (P_i\cap d'=i)}{\mathsf{c}}{\cup_{i\in I} (Q_i\cap d'=i)})\}$}
        \item \resizebox{\linewidth}{!}{$\closure{disj}{\leftrightleftrightarrows} \circ \closure{exists}{\leftrightleftrightarrows} \circ \closure{frame}{\backwardarrows}\{\bwdundertriple{P_i}{\mathsf{c}}{Q_i})\;|\;i\in I\} = 
        \closure{exists}{\leftrightleftrightarrows} \circ \closure{frame}{\backwardarrows}\{\bwdundertriple{\cup_{i\in I} (P_i\cap d'=i)}{\mathsf{c}}{\cup_{i\in I} (Q_i\cap d'=i)})\}$}
    \end{enumerate}
\end{theorem}
\begin{proof}
    We only show the proof for the first statement, as the others are analogous. We proceed by simply calculating the two closures separately and proving they are equal.
    \begin{align*}
        &\closure{exists}{\leftrightleftrightarrows} \circ \closure{frame}{\forwardarrows}\{(\cup_{i\in I} (P_i\cap d'=i), \mathsf{c}, \cup_{i\in I} (Q_i\cap d'=i))\} \\
        &=\closure{exists}{\leftrightleftrightarrows} \{((\cup_{i\in I} (P_i\cap d'=i))\sast R, \mathsf{c}, (\cup_{i\in I} (Q_i\cap d'=i))\sast R) \;|\; R\in \mathbb{F}, (\cup_{i\in I} (P_i\cap d'=i))\heapcomp R\} \\
        &=\closure{exists}{\leftrightleftrightarrows} \{((\cup_{i\in I} P_i)\sast R, \mathsf{c}, (\cup_{i\in I} Q_i)\sast R) \;|\; R\in \mathbb{F}, (\cup_{i\in I} P_i)\heapcomp R\} \\
        &=\closure{exists}{\leftrightleftrightarrows} \{(\cup_{i\in I} P_i\sast R, \mathsf{c}, \cup_{i\in I} Q_i\sast R) \;|\; R\in \mathbb{F}, i\in I. P_i\heapcomp R\} \\
        &=\{(\exists \mathsf{X}.\cup_{i\in I} P_i\sast R, \mathsf{c}, \exists \mathsf{X}.\cup_{i\in I} Q_i\sast R) \;|\; R\in \mathbb{F}, i\in I. P_i\heapcomp R, \mathsf{X}\subseteq \mathbb{X}_\mathsf{l}\} \\
        &=\{(\cup_{i\in I}\exists \mathsf{X}.P_i\sast R, \mathsf{c}, \cup_{i\in I}\exists \mathsf{X}.Q_i\sast R)\;|\;i\in I, R\in \mathbb{F}, P_i \heapcomp R, \mathsf{X}\subseteq \mathbb{X}_\mathsf{l}\} \\
        &=\closure{disj}{\leftrightleftrightarrows} \{(\exists \mathsf{X}.P_i\sast R, \mathsf{c}, \exists \mathsf{X}.Q_i\sast R)\;|\;i\in I, R\in \mathbb{F}, P_i \heapcomp R, \mathsf{X}\subseteq \mathbb{X}_\mathsf{l}\} \\
        &=\closure{disj}{\leftrightleftrightarrows} \circ \closure{exists}{\leftrightleftrightarrows} \{(P_i\sast R, \mathsf{c}, Q_i\sast R)\;|\;i\in I, R\in \mathbb{F}, P_i \heapcomp R\} \\
        &=\closure{disj}{\leftrightleftrightarrows} \circ \closure{exists}{\leftrightleftrightarrows} \circ \closure{frame}{\forwardarrows}\{(P_i, \mathsf{c}, Q_i)\;|\;i\in I\}
    \end{align*}
    where in the second step we used the fact that $d'=i \in \mathbb{F}$. 
\end{proof}

\subsection{Proofs of Sec.\ref{sec:systematic_forward}}
We note that generic heaps $\mathsf{H^{\circled{i}}}$ described by the assertion language within the memory model $\circled{i}\in \{\circled{1}, \circled{2}\}$ must be of the following form:
\begin{align*}
    \mathsf{H^{\circled{1}} = \emp \sepconj_{i\in \{1,...,n\}} x_{i,1}\mapsto x_{i,2}} 
    &&
    \mathsf{H^{\circled{2}} = \emp \sepconj_{i\in \{1,...,n\}} x_{i,1}\mapsto x_{i,2} \sepconj_{j\in \{1,...,m\}} x_j\not\mapsto}
\end{align*}

\noindent
i.e., a (finite) separating conjunction of ownership predicate and negative heap assertion (only in the model \circled{2}) along with the empty assertion. Then, a generic memory takes the following form:
\begin{align*}
    \mathsf{M^{\circled{i}} = \bigvee_{k\in K} \left(H^{\circled{i}}_k \wedge P_k\right)}
\end{align*}

\noindent
i.e., a disjunction of structural ($\mathsf{H^{\circled{i}}_k}$) and pure ($\mathsf{P_k}$) assertions. 
Although the assertion language only accounts for finite disjunction, infinite disjunctions are encoded using the existential quantifier. 

For each of the memory models \circled{1} and \circled{2}, we define, respectively, two memories $\mathsf{\heapcore^{\circled{1}}}$ and $\mathsf{\heapcore^{\circled{2}}}$:
\[
\begin{aligned}    
    \mathsf{\heapcore^{\circled{1}}} &\triangleq \mathsf{(a'\doteq 1 \ast \emp) \vee (a'\doteq 2 \ast k_1'\mapsto k_2')} \\
    \mathsf{\heapcore^{\circled{2}}} &\triangleq \mathsf{(a'\doteq 1 \ast \emp) \vee (a'\doteq 2 \ast k_1'\mapsto k_2') \vee (a'\doteq 3 \ast k_1'\not\mapsto)}
\end{aligned}
\]

\begin{lemma} \label{lm:rewrite}
    Each assertion $\mathsf{M^{\circled{i}}}$ can be rewritten in a normal form $\mathsf{\exists a'. \heapcore^{\circled{i}} \ast M'}$ for some $\mathsf{M'}$.
\end{lemma}
\begin{proof}
    We only show the proof for memory model \circled{2} since it includes the case for memory model \circled{1}. 
    We proceed by induction on the structure of $\mathsf{M^{\circled{2}}}$, using the letter $\mathsf{P}$ for pure assertions:
    \\\textbf{Case $\mathsf{M^{\circled{2}}\triangleq P}$:} \\
    We rewrite $\mathsf{P \equiv \exists a'.\heapcore^{\circled{2}} \ast P}$ and we prove the equivalence by proving separately the two implications. \begin{description}
        \item[$\Rightarrow$)] Suppose $(s,h)\in \llbrace \mathsf{P}\rrbrace$ then $(s,h)\in \llbrace \mathsf{\exists a'.a'\doteq 1\ast \emp \ast P}\rrbrace = \llbrace \mathsf{\emp \ast P}\rrbrace$ = $\llbrace \mathsf{P}\rrbrace$ which is one of the disjunct of $\mathsf{\heapcore^{\circled{2}}}$.
        \item[$\Leftarrow$)] Suppose that $(s,h)\in \llbrace \mathsf{\exists a'.\heapcore^{\circled{2}} \ast P}\rrbrace$, then by definition of $\ast$ we know that $\exists v,h_1,h_2. (s[\mathsf{a'}\mapsto v], h_1)\in \llbrace\mathsf{\heapcore^{\circled{2}}}\rrbrace$ and $(s[\mathsf{a'}\mapsto v], h_2)\in \llbrace\mathsf{P}\rrbrace$. Then since $\mathsf{a'}$ is a fresh variable $(s,h_2)\in \llbrace\mathsf{P}\rrbrace$ and since $\mathsf{P}$ is a pure assertion also $(s,h_1\bullet h_2)\in \llbrace\mathsf{P}\rrbrace$.
    \end{description}
    \textbf{Case $\mathsf{M^{\circled{2}}\triangleq H^{\circled{2}}}$:} \\
    To rewrite $\mathsf{H^{\circled{2}}}$ in the normal form, we proceed by induction on the structure of $\mathsf{H^{\circled{2}}}$:
    \begin{enumerate}
        \item If $\mathsf{H^{\circled{2}} = \emp}$ then $\mathsf{\emp \equiv \exists a'. \heapcore^{\circled{2}} \ast a'\doteq 1}$
        \item If $\mathsf{H^{\circled{2}} = x\mapsto y}$ then $\mathsf{x\mapsto y \equiv \exists a'. \heapcore^{\circled{2}} \ast (a'\doteq 2 \ast k_1'\doteq x \ast k_2'\doteq y)}$
        \item If $\mathsf{H^{\circled{2}} = x\not\mapsto}$ then $\mathsf{x\not\mapsto \equiv \exists a'. \heapcore^{\circled{2}} \ast (a'\doteq 3 \ast k_1'\doteq x)}$
        \item If $\mathsf{H^{\circled{2}} = H_1 \ast H_2}$ then $\mathsf{H_1 \ast H_2 \equiv \exists a'.\heapcore^{\circled{2}} \ast ((H_1\ast H_2')\vee(H_1'\ast H_2))}$ with $\mathsf{H_1 \equiv \exists a'.\heapcore^{\circled{2}}\ast H_1'}$ and $\mathsf{H_2 \equiv \exists a'.\heapcore^{\circled{2}}\ast H_2'}$.  
    \end{enumerate}
    The first three equivalences are trivial, so we only prove (4):
    \begin{description}
        \item[$\Rightarrow$)] Suppose $(s,h)\in \llbrace \mathsf{H_1\ast H_2}\rrbrace$. By definition of $\ast$, there exists $h_1,h_2$ s.t. $h=h_1\bullet h_2$, $(s,h_1)\in \llbrace\mathsf{H_1}\rrbrace$ and $(s,h_2)\in \llbrace\mathsf{H_2}\rrbrace$. By definition of $\mathsf{H_1}$ and $\mathsf{H_2}$ we have $(s,h_1)\in \llbrace\mathsf{\exists a'.\heapcore^{\circled{2}}\ast H_1'}\rrbrace$ and $(s,h_2)\in \llbrace\mathsf{\exists a'.\heapcore^{\circled{2}}\ast H_2'}\rrbrace$. By definition of $\ast$, there exists $h_1', h_1'', h_2', h_2''$ s.t. $h_1=h_1'\bullet h_1''$, $h_2 = h_2'\bullet h_2''$, $(s,h_1')\in \llbrace \mathsf{\heapcore^{\circled{2}}}\rrbrace$, $(s,h_1'')\in \llbrace \mathsf{H_1'}\rrbrace$, $(s,h_2')\in \llbrace \mathsf{\heapcore^{\circled{2}}}\rrbrace$, $(s,h_2'')\in \llbrace \mathsf{H_2'}\rrbrace$. Then $(s,h_1'\bullet h_1''\bullet h_2)\in \llbrace \mathsf{\heapcore^{\circled{2}}\ast H_1'\ast H_2}\rrbrace$ and so also $(s,h_1'\bullet h_1''\bullet h_2)\in \llbrace \mathsf{\exists a'. \heapcore^{\circled{2}}\ast H_1'\ast H_2}\rrbrace$. Furthermore, $(s,h_1\bullet h_2'\bullet h_2'')\in \llbrace \mathsf{H_1\ast \heapcore^{\circled{2}} \ast H_2'}\rrbrace$ and so also $(s,h_1\bullet h_2'\bullet h_2'')\in \llbrace \mathsf{\exists a'. H_1\ast \heapcore^{\circled{2}} \ast H_2'}\rrbrace$. 
        \item[$\Leftarrow$)] Suppose $(s,h)\in \llbrace \mathsf{\exists a'.\heapcore^{\circled{2}} \ast ((H_1\ast H_2')\vee(H_1'\ast H_2))}\rrbrace$. By definition of $\ast$ and $\exists$, there exists $v, h_1,h_2$ s.t. $h=h_1\bullet h_2$, $(s[\mathsf{a'}\mapsto v],h_1)\in \llbrace \mathsf{\heapcore^{\circled{2}}}\rrbrace$ and $(s[\mathsf{a'}\mapsto v],h_2)\in \llbrace \mathsf{(H_1\ast H_2')\vee(H_1'\ast H_2)}\rrbrace$. We proceed by case on the disjuncts of the assertion $\mathsf{(H_1\ast H_2')\vee(H_1'\ast H_2)}$: 
            \begin{itemize}
                \item Suppose $(s[\mathsf{a'}\mapsto v],h_2)\in \llbrace \mathsf{H_1\ast H_2'}\rrbrace$. 
                Then by definition of $\ast$ there exists $h_2', h_2''$ s.t. $(s[\mathsf{a'}\mapsto v],h_2') \in \llbrace \mathsf{H_1}\rrbrace$ and $(s[\mathsf{a'}\mapsto v],h_2'') \in \llbrace \mathsf{H_2'}\rrbrace$. Then $(s[\mathsf{a'}\mapsto v],h_1\bullet h_2'') \in \llbrace \mathsf{H_2}\rrbrace$. Then $(s[\mathsf{a'}\mapsto v],h_1\bullet h_2'\bullet h_2'') \in \llbrace \mathsf{H_1\ast H_2}\rrbrace$.
                \item The other case is dual.
            \end{itemize}
    \end{description}
    \textbf{Case $\mathsf{M^{\circled{2}}\triangleq P\wedge H^{\circled{2}}}$:}\\
    By inductive hypothesis $\mathsf{H^{\circled{2}} \equiv \exists a'.\heapcore^{\circled{2}}\ast H'}$, then we rewrite $\mathsf{P\wedge H^{\circled{2}} \equiv \exists a'.\heapcore^{\circled{2}} \ast P\wedge H'}$:
    \[\mathsf{P\wedge H^{\circled{2}} \equiv P \wedge \exists a'.\heapcore^{\circled{2}}\ast H' \equiv \exists a'. P \wedge \heapcore^{\circled{2}}\ast H' \equiv \exists a'.\heapcore^{\circled{2}} \ast P\wedge H'}\]
    where we used the fact that $\mathsf{a'}$ is a fresh variable.
    \\\textbf{Case $\mathsf{M^{\circled{2}}\triangleq P\ast H^{\circled{2}}}$:}\\
    By inductive hypothesis $\mathsf{H^{\circled{2}} \equiv \exists a'.\heapcore^{\circled{2}}\ast H'}$, then we rewrite $\mathsf{P\ast H^{\circled{2}} \equiv \exists a'.\heapcore^{\circled{2}} \ast ((P\ast H^{\circled{2}})\vee(P\ast H'))}$. The proof is analogous to the proof of the rewriting of $\mathsf{H_1\ast H_2}$.
    \\\textbf{Case $\mathsf{M^{\circled{2}}\triangleq M_1 \vee M_2}$:}\\
    By inductive hypothesis $\mathsf{M_1 \equiv \exists a'.\heapcore^{\circled{2}}\ast M_1'}$ and $\mathsf{M_2 \equiv \exists a'.\heapcore^{\circled{2}}\ast M_2'}$, then we rewrite \\$\mathsf{M_1\vee M_2 \equiv \exists a'.\heapcore^{\circled{2}}\ast (M_1'\vee M_2')}$:
    \[\mathsf{M_1\vee M_2 \equiv (\exists a'.\heapcore^{\circled{2}}\ast M_1') \vee (\exists a'.\heapcore^{\circled{2}}\ast M_2') \equiv \exists a'.(\heapcore^{\circled{2}}\ast M_1') \vee (\heapcore^{\circled{2}}\ast M_2') \equiv \exists a'.\heapcore^{\circled{2}}\ast (M_1'\vee M_2')}\]
    where we used the distributivity of the $\ast$ over $\vee$.
    \\\textbf{Case $\mathsf{M^{\circled{2}}\triangleq \exists X.M_1}$:}\\
    By inductive hypothesis $\mathsf{M_1 \equiv \exists a'.\heapcore^{\circled{2}}\ast M_1'}$, then we rewrite $\mathsf{\exists X.M_1 \equiv \exists a'. \heapcore^{\circled{2}} \ast (\exists X. M_1')}$:
    \[\mathsf{\exists X.M_1 \equiv \exists X.(\exists a'.\heapcore^{\circled{2}}\ast M_1') \equiv \exists a'.\exists X. \heapcore^{\circled{2}}\ast M_1' \equiv \exists a'. \heapcore^{\circled{2}} \ast (\exists X. M_1')}\]
\end{proof}

\begin{lemma} \label{lm:rewrite2}
    We can rewrite $\mathsf{\heapcore^{\circled{1}}}$ and $\mathsf{\heapcore^{\circled{2}}}$ in the following forms:
    \begin{enumerate}
        \item $\mathsf{\heapcore^{\circled{1}} \equiv \overline{H}\ast \heapcore}$ with $\mathsf{\overline{H}=\emp}$
        \item $\mathsf{\heapcore^{\circled{2}} \equiv \overline{H}\ast \heapcore^{\circled{2}}}$ with $\mathsf{\overline{H}=\emp}$
        \item $\mathsf{\heapcore^{\circled{2}} \equiv \overline{H}\ast (\emp \vee (a\doteq 3\ast k_1\not\mapsto))}$ with $\mathsf{\overline{H}=((a\doteq 1\ast \emp) \vee (a\doteq 2\ast k_1\mapsto k_2))}$
    \end{enumerate}
\end{lemma}

\setcounter{theorem}{0} 
\renewcommand{\thetheorem}{\ref{th:sl_completeness}}
\begin{theorem}[SL+ Sound and Relative Complete for Loop-Free Programs]
    Let $\mathsf{r\in Reg}$ and $\mathsf{P,Q\in Ast}$, then $\vdash_{\mathsf{SL+}}^\text{\circled{i}}\hl{\mathsf{P}}{\mathsf{r}}{\mathsf{Q}} \Rightarrow \llbracket \mathsf{r} \rrbracket^\mathsf{SL}_\text{\circled{i}} \llbrace \mathsf{P} \rrbrace \subseteq \llbrace \mathsf{Q} \rrbrace$. Moreover, $\llbracket \mathsf{r}\rrbracket^\mathsf{SL}_\text{\circled{i}} \llbrace \mathsf{P}\rrbrace \subseteq \llbrace \mathsf{Q}\rrbrace \Rightarrow \vdash_{\mathsf{SL+}}^\text{\circled{i}}\hl{\mathsf{P}}{\mathsf{r}}{\mathsf{Q}}$ when $\mathsf{r}$ is loop-free.
\end{theorem}
\begin{proof}
    Soundness can be proved by induction on the derivation tree.
    Relative completeness for loop-free programs will be proved for the memory model $\circled{1}$ since the proof for $\circled{2}$ is analogous, but it uses $\mathsf{\heapcore^{\circled{2}}}$.
    We suppose that $\semantics{r}{SL}{\circled{1}} \llbrace \mathsf{P}\rrbrace \subseteq \llbrace \mathsf{Q}\rrbrace$ and we prove that the triple $\vdash_{\mathsf{SL+}}^\text{\circled{1}}\hl{\mathsf{P}}{\mathsf{r}}{\mathsf{Q}}$ is derivable in SL+. 
    In particular, we show that we can derive the exact postcondition, then by applying the consequence rule, we can derive any (expressible) over-approximation. 
    The proof proceeds by induction on the structure of $\mathsf{r}$.\\
    \textbf{Case $\mathsf{r\triangleq c}$:}\\
    By Lemma \ref{lm:rewrite} we have $\mathsf{P \equiv \exists a'.\heapcore^{\circled{1}}\ast P'}$ and by Proposition \ref{lm:rewrite2} we also have $\mathsf{P \equiv \exists a'.\emp \ast \heapcore^{\circled{1}}\ast P'}$ which are respectively equivalent to:
    \begin{enumerate}
        \item $\mathsf{P \equiv \textcolor{gray}{\exists a',X'.}\emp_{\mathbb{X}_\mathsf{p}}\ast\emp \textcolor{gray}{\ast \heapcore^{\circled{1}}\ast P'[X'/\mathbb{X}_\mathsf{p}]}}$
        \item $\mathsf{P \equiv \textcolor{gray}{\exists a',X'.}\emp_{\mathbb{X}_\mathsf{p}}\ast\heapcore^{\circled{1}}\textcolor{gray}{\ast P'[X'/\mathbb{X}_\mathsf{p}]}}$
    \end{enumerate}
    In gray, we outline the idle part of the memories that can be added at a later time, while in black, we outline the minimal memory needed for the execution of atomic commands. 
    In particular, the black part of (1) is the minimal memory needed for the execution of $\mathsf{b?}$, $\mathsf{x:=e}$ and $\mathsf{x:=alloc()}$ and the black part of (2) for $\mathsf{free(x)}$, $\mathsf{x:=[y]}$ and $\mathsf{[x]:=y}$. 
    We proceed by case for (1) and (2):
    \begin{enumerate}
        \item We show a derivation for the triple $\mathsf{\hl{P}{c}{\semantics{c}{SL}{\circled{1}} P}}$ by starting with the black part of (1):
        \[
        \resizebox{.99\linewidth}{!}{
        \infer[\hllabel{Frame}{\circled{i}}]
            {\infer[\hllabel{Exists}{\circled{i}}]
                {\infer[\hllabel{Cons}{\circled{i}}]
                    {\mathsf{\vdash_{\mathsf{SL+}}^{\circled{1}}\hl{P}{c}{\llbracket c\rrbracket^\mathsf{SL}_\text{\circled{1}} P}}}
                    {\mathsf{\vdash_{\mathsf{SL+}}^{\circled{1}}\hl{\exists a',X'.\emp_{\mathbb{X}_\mathsf{p}}\ast \heapcore^{\circled{1}} \ast P'[X'/\mathbb{X}_\mathsf{p}]}{c}{\exists a',X'.\llbracket c\rrbracket^\mathsf{SL}_\text{\circled{1}} (\emp_{\mathbb{X}_\mathsf{p}}\ast\heapcore^{\circled{1}}\ast P'[X'/\mathbb{X}_\mathsf{p}])}}}}
                {\mathsf{\vdash_{\mathsf{SL+}}^{\circled{1}}\hl{\emp_{\mathbb{X}_\mathsf{p}} \ast \heapcore^{\circled{1}} \ast P'[X'/\mathbb{X}_\mathsf{p}]}{c}{\llbracket c\rrbracket^\mathsf{SL}_\text{\circled{1}} (\emp_{\mathbb{X}_\mathsf{p}}\ast \heapcore^{\circled{1}}\ast P'[X'/\mathbb{X}_\mathsf{p}])}}}}
            {\mathsf{\heapcore^{\circled{1}} \ast P'[X'/\mathbb{X}_\mathsf{p}]\in \mathbb{F}} & \mathsf{\vdash_{\mathsf{SL+}}^{\circled{1}}\hl{\emp_{\mathbb{X}_\mathsf{p}}}{c}{\llbracket c\rrbracket^\mathsf{SL}_\text{\circled{1}} (\emp_{\mathbb{X}_\mathsf{p}})}} & \mathsf{\heapcore^{\circled{1}} \ast P'[X'/\mathbb{X}_\mathsf{p}]\logicheapcomp (\emp_{\mathbb{X}_\mathsf{p}})}}
        }\]    
        To conclude the proof for this case, we compute the exact semantics for each atomic command on the minimal memory, showing that it is expressible in the assertion language:
        \begin{itemize}
            \item $\llbracket \mathsf{b?}\rrbracket^\mathsf{SL}_\text{\circled{1}} \llbrace \emp_{\mathbb{X}_\mathsf{p}}\rrbrace = \{(s,h)\;|\;(s,h)\in \llbrace \emp_{\mathbb{X}_\mathsf{p}} \rrbrace, \llparenthesis \mathsf{b}\rrparenthesis s = \true\} = \llbrace \mathsf{\emp_{\mathbb{X}_\mathsf{p}} \wedge b}\rrbrace$
            \item 
                $\begin{aligned}[t]
                    &\llbracket \mathsf{x:=e}\rrbracket^\mathsf{SL}_\text{\circled{1}} \llbrace \emp_{\mathbb{X}_\mathsf{p}}\rrbrace &\\
                    &= \{(s[\mathsf{x}\mapsto \llparenthesis \mathsf{e}\rrparenthesis s], h) \;|\; \forall \mathsf{y}\in \mathbb{X}_\mathsf{p}. \exists \mathsf{y'}\in \mathbb{X}_\mathsf{l}. s(\mathsf{y})=s(\mathsf{y'})\} & \\
                    &= \{(s[\mathsf{x}\mapsto \llparenthesis \mathsf{e}[\mathsf{x'}/\mathsf{x}]\rrparenthesis s], h) \;|\; \forall \mathsf{y}\in \mathbb{X}_\mathsf{p}. \exists \mathsf{y'}\in \mathbb{X}_\mathsf{l}. s(\mathsf{y})=s(\mathsf{y'}) \} \\
                    &= \llbrace \mathsf{\emp_{\progvar\setminus \{x\}}\wedge x=e[x'/x]}\rrbrace
                \end{aligned}$
            \item 
                $\begin{aligned}[t]
                    &\llbracket \mathsf{x:=alloc()}\rrbracket^\mathsf{SL}_\text{\circled{1}} \llbrace \emp_{\mathbb{X}_\mathsf{p}}\rrbrace & \\
                    &=\{(s[\mathsf{x}\mapsto l],h[l\mapsto v])\;|\; (s,h)\in \llbrace \emp_{\mathbb{X}_\mathsf{p}}\rrbrace, v\in \mathbb{V}\} &\\
                    &=\{(s[\mathsf{x}\mapsto l],[l\mapsto v])\;|\; \forall \mathsf{y}\in \mathbb{X}_\mathsf{p}. \exists \mathsf{y'}\in \mathbb{X}_\mathsf{l}. s(\mathsf{y})=s(\mathsf{y'}), v \in \mathbb{V}\} &\\
                    &=\llbrace \emp_{\mathbb{X}_\mathsf{p}\setminus \{\mathsf{x}\}} \ast \mathsf{x}\mapsto \_\rrbrace
                \end{aligned}$\\
                This exact postcondition can be obtained from axiom $\hllabel{Alloc}{\circled{i}}$ by using rule $\hllabel{Exists}{\circled{i}}$.
        \end{itemize}
        \item Although $\mathsf{\emp_{\mathbb{X}_\mathsf{p}}\ast\heapcore^{\circled{1}}}$ being the minimal memory for the execution of $\mathsf{free(x)}$, $\mathsf{x:=[y]}$ and $\mathsf{[x]:=y}$, their executions on the empty heap yelds the abort state $\abort$, which is not expressible in our assertion language. 
        Then, the minimal part of the memory needed for the execution of commands $\mathsf{free(x)}$, $\mathsf{x:=[y]}$ and $\mathsf{[x]:=y}$ yielding an expressible postcondition is described by the assertion $\mathsf{M_1 \triangleq \emp_{\mathbb{X}_\mathsf{p}}\ast (a'\doteq 2 \ast k_1'\mapsto k_2')}$. 
        We show a derivation for the triple $\mathsf{\hl{\exists a',X'. M_1\ast P[X'/\mathbb{X}_\mathsf{p}]}{c}{\semantics{c}{SL}{\circled{1}} (\exists a',X'. M_1\ast P[X'/\mathbb{X}_\mathsf{p}])}}$ by starting from $\mathsf{M_1}$:
        \[
        \resizebox{.99\linewidth}{!}{
        \infer[\hllabel{Frame}{\circled{i}}]
            {\infer[\hllabel{Exists}{\circled{i}}]
                {\mathsf{\vdash_{\mathsf{SL+}}^{\circled{1}}\hl{\exists a,X'. M_1\ast P'[X'/\mathbb{X}_\mathsf{p}]}{c}{\exists a,X'.\llbracket c\rrbracket^\mathsf{SL}_\text{\circled{1}} (M_1\ast P'[X'/\mathbb{X}_\mathsf{p}])}}}
                {\mathsf{\vdash_{\mathsf{SL+}}^{\circled{1}}\hl{M_1\ast P'[X'/\mathbb{X}_\mathsf{p}]}{c}{\llbracket c\rrbracket^\mathsf{SL}_\text{\circled{1}} (M_1\ast P'[X'/\mathbb{X}_\mathsf{p}])}}}}
            {\mathsf{P'[X'/\mathbb{X}_\mathsf{p}]\in \mathbb{F}} & \mathsf{\vdash_{\mathsf{SL+}}^{\circled{1}}\hl{M_1}{c}{\llbracket c\rrbracket^\mathsf{SL}_\text{\circled{1}} M_1}} & \mathsf{P'[X'/\mathbb{X}_\mathsf{p}]\logicheapcomp M_1}}
        }\]
        To conclude the proof for this case, we compute the exact semantics for each atomic command on the minimal memory, showing that it is expressible in the assertion language:
        \begin{itemize}
            \item 
                $\begin{aligned}[t]
                    &\llbracket \mathsf{free(x)}\rrbracket^\mathsf{SL}_\text{\circled{1}} \llbrace \mathsf{\emp_{\mathbb{X}_\mathsf{p}}\ast a'\doteq 2 \ast x'\mapsto k_2'}\rrbrace & \\
                    =& \{(s,[])\;|\; \forall \mathsf{y}\in \mathbb{X}_\mathsf{p}. \exists \mathsf{y'}\in \mathbb{X}_\mathsf{l}. s(\mathsf{y})=s(\mathsf{y'}), s(\mathsf{a'})=2\} &\\
                    =& \llbrace \mathsf{\emp_{\mathbb{X}_\mathsf{p}}\ast a'\doteq 2} \rrbrace 
                \end{aligned}$
            \item 
                $\begin{aligned}[t]
                    &\llbracket \mathsf{x:=[y]}\rrbracket^\mathsf{SL}_\text{\circled{1}} \llbrace \mathsf{\emp_{\mathbb{X}_\mathsf{p}}\ast a'\doteq 2 \ast y'\mapsto k_2'}\rrbrace \\
                    =& \{(s[\mathsf{x}\mapsto \mathsf{k_2'}],[\mathsf{y'}\mapsto \mathsf{k_2'}]) \;|\; \forall \mathsf{z}\in \mathbb{X}_\mathsf{p}. \exists \mathsf{z'}\in \mathbb{X}_\mathsf{l}. s(\mathsf{z})=s(\mathsf{z'}), s(\mathsf{a'})=2\} \\
                    =& \llbrace \mathsf{\emp_{\mathbb{X}_\mathsf{p}\setminus\{\mathsf{x}\}}\ast a'\doteq 2 \ast y'\mapsto k_2' \ast x\doteq k_2'}\rrbrace
                \end{aligned}$
            \item 
                $\begin{aligned}[t]
                    &\llbracket \mathsf{[x]:=y}\rrbracket^\mathsf{SL}_\text{\circled{1}} \llbrace \mathsf{\emp_{\mathbb{X}_\mathsf{p}}\ast a'\doteq 2 \ast x'\mapsto k_2'}\rrbrace \\
                    =& \{(s,[\mathsf{x'}\mapsto \mathsf{y}]) \;|\; \forall \mathsf{y}\in \mathbb{X}_\mathsf{p}. \exists \mathsf{y'}\in \mathbb{X}_\mathsf{l}. s(\mathsf{y})=s(\mathsf{y'}), s(\mathsf{a'})=2\} \\
                    =& \llbrace \mathsf{\emp_{\mathbb{X}_\mathsf{p}}\ast a'\doteq 2 \ast x'\mapsto y}\rrbrace
                \end{aligned}$
        \end{itemize}
    \end{enumerate}
    \textbf{Case $\mathsf{r\triangleq r_1;r_2}$:}\\
    Let us assume that $\semantics{r_1;r_2}{SL}{\circled{1}} \llbrace \mathsf{P}\rrbrace=\llbrace \mathsf{Q}\rrbrace$. 
    By definition of sequential composition $\semantics{r_1;r_2}{SL}{\circled{1}} = \semantics{r_2}{SL}{\circled{1}}(\semantics{r_1}{SL}{\circled{1}} \llbrace \mathsf{P}\rrbrace)$. 
    By inductive hypothesis on $\mathsf{r_1}$ there exists an assertion $\mathsf{S}$ s.t. $\semantics{r_1}{SL}{\circled{1}} \llbrace \mathsf{P}\rrbrace = \llbrace \mathsf{S}\rrbrace$ and $\vdash_{\mathsf{SL+}}^{\circled{1}}\hl{\mathsf{P}}{\mathsf{r_1}}{\mathsf{S}}$ is provable. 
    Furthermore, by inductive hypothesis on $\mathsf{r_2}$ there exists an assertion $\mathsf{Q}$ s.t. $\semantics{r_2}{SL}{\circled{1}} \llbrace \mathsf{S}\rrbrace = \llbrace \mathsf{Q}\rrbrace$ and $\vdash_{\mathsf{SL+}}^{\circled{1}}\hl{\mathsf{S}}{\mathsf{r_2}}{\mathsf{Q}}$ is provable. Then, by applying rule of Seq, we can derive $\vdash_{SL+}\hl{\mathsf{P}}{\mathsf{r_1;r_2}}{\mathsf{Q}}$:
    \[
    \infer[\hllabel{Seq}{\circled{i}}]
        {\vdash_{\mathsf{SL+}}^{\circled{1}}\hl{\mathsf{P}}{\mathsf{r_1;r_2}}{\mathsf{Q}}}
        {\vdash_{\mathsf{SL+}}^{\circled{1}}\hl{\mathsf{P}}{\mathsf{r_1}}{\mathsf{S}} & \vdash_{\mathsf{SL+}}^{\circled{1}}\hl{\mathsf{S}}{\mathsf{r_2}}{\mathsf{Q}}}
    \]
    \textbf{Case $\mathsf{r\triangleq r_1+r_2}$:}\\
    Let us assume that $\semantics{r_1+r_2}{SL}{\circled{1}} \llbrace \mathsf{P}\rrbrace=\llbrace \mathsf{Q}\rrbrace$. 
    By inductive hypothesis on $\mathsf{r_1}$ there exists an assertion $\mathsf{Q_1}$ s.t. $\semantics{r_1}{SL}{\circled{1}} \llbrace \mathsf{P}\rrbrace = \llbrace \mathsf{Q_1}\rrbrace$ and $\vdash_{\mathsf{SL+}}^{\circled{1}}\hl{\mathsf{P}}{\mathsf{r_1}}{\mathsf{Q_1}}$ is provable. 
    By inductive hypothesis on $\mathsf{r_2}$ there exists an assertion $\mathsf{Q_2}$ s.t. $\semantics{r_2}{SL}{\circled{1}} \llbrace \mathsf{P}\rrbrace = \llbrace \mathsf{Q_2}\rrbrace$ and $\vdash_{\mathsf{SL+}}^{\circled{1}}\hl{\mathsf{P}}{\mathsf{r_2}}{\mathsf{Q_2}}$. 
    Furthermore, by definition of nondeterministic choice $\semantics{r_1+r_2}{SL}{\circled{1}} \llbrace \mathsf{P}\rrbrace = \semantics{r_1}{SL}{\circled{1}} \llbrace \mathsf{P}\rrbrace \cup \semantics{r_2}{SL}{\circled{1}} \llbrace \mathsf{P}\rrbrace = \llbrace \mathsf{Q_1\vee Q_2}\rrbrace$. 
    Then by applying the rule of choice and consequence, we can derive $\vdash_{\mathsf{SL+}}^{\circled{1}}\hl{\mathsf{P}}{\mathsf{r_1+r_2}}{\mathsf{\mathsf{Q_1}\vee \mathsf{Q_2}}}$:
    \[
    \infer[\hllabel{Choice}{\circled{i}}]
        {\vdash_{\mathsf{SL+}}^{\circled{1}}\hl{\mathsf{P}}{\mathsf{r_1;r_2}}{\mathsf{Q_1\vee Q_2}}}
        {\infer[\hllabel{Cons}{\circled{i}}]{\vdash_{\mathsf{SL+}}^{\circled{1}}\hl{\mathsf{P}}{\mathsf{r_1}}{\mathsf{Q_1\vee Q_2}}}{\vdash_{\mathsf{SL+}}^{\circled{1}}\hl{\mathsf{P}}{\mathsf{r_1}}{\mathsf{Q_1}}} & 
         \infer[\hllabel{Cons}{\circled{i}}]{\vdash_{\mathsf{SL+}}^{\circled{1}}\hl{\mathsf{P}}{\mathsf{r_2}}{\mathsf{Q_1\vee Q_2}}}{\vdash_{\mathsf{SL+}}^{\circled{1}}\hl{\mathsf{P}}{\mathsf{r_2}}{\mathsf{Q_2}}}}
    \]
\end{proof}

\setcounter{theorem}{0} 
\renewcommand{\thetheorem}{\ref{th:relative_completeness_loopfree_isl}}

\begin{theorem}[ISL+ Sound and Relative Complete for Loop-Free Programs]
    Let $\mathsf{r\in Reg}$ and $\mathsf{P,Q\in Ast}$, then $\vdash_{\mathsf{ISL+}}^{\text{\circled{i}}}\il{\mathsf{P}}{\mathsf{r}}{\mathsf{Q}}\Rightarrow \llbrace \mathsf{Q} \rrbrace \subseteq \llbracket \mathsf{r} \rrbracket^\mathsf{ISL}_{\text{\circled{i}}} \llbrace \mathsf{P} \rrbrace$. Moreover, $\llbrace \mathsf{Q}\rrbrace\subseteq\llbracket \mathsf{r}\rrbracket^\mathsf{ISL}_{\text{\circled{i}}}\llbrace \mathsf{P}\rrbrace \Rightarrow \vdash_{\mathsf{ISL+}}^{\text{\circled{i}}}\il{\mathsf{P}}{\mathsf{r}}{\mathsf{Q}}$ when $\mathsf{r}$ is loop-free.
\end{theorem}
\begin{proof}
    The proof proceeds like the proof of Theorem \ref{th:sl_completeness} and by using Lemma \ref{lm:rewrite2}.
\end{proof}

\setcounter{theorem}{0} 
\renewcommand{\thetheorem}{\ref{th:statewise_completeness_isl}}

\begin{theorem}[State-wise completeness] 
    Given a program $\mathsf{r}$ and two states $(\oker: (s,h)),(\oker:(s',h'))$ such that, $(\oker:(s',h'))\in \semantics{r}{ISL}{\circled{i}} (\oker:(s,h))$, for every assertion $\mathsf{P}$ such that $(s,h) \in \llbrace \mathsf{P}\rrbrace$ there exists an assertion $\mathsf{Q}$ such that $(s',h') \in \llbrace \mathsf{Q}\rrbrace$ and $\vdash_{\mathsf{ISL+}}^{\circled{i}}\il{\mathsf{P}}{\mathsf{r}}{\mathsf{Q}}$ is provable. 
\end{theorem}
\begin{proof}
    The proof proceeds by induction on the structure of $\mathsf{r}$:
    \\\textbf{Case $\mathsf{r\triangleq c}$:}\\
    Atomic commands are loop-free programs. Completeness is established in Theorem \ref{th:relative_completeness_loopfree_isl}.
    \\\textbf{Case $\mathsf{r\triangleq r_1;r_2}$:}\\
    We suppose $(\oker:(s',h')) \in \semantics{r_1;r_2}{ISL}{\circled{i}} (\oker:(s,h))$ with $(s,h) \in \llbrace \mathsf{P}\rrbrace$. 
    By definition, we distinguish the two cases:
    \begin{itemize}
        \item $(\er:(s',h'))\in \semantics{r_1}{ISL}{\circled{i}}(\oker:(s,h))$: 
        By inductive hypothesis there exists an assertion $\mathsf{Q}$ such that $\il[blue][red]{\mathsf{P}}{\mathsf{r_1}}{\mathsf{Q}}$ and $(s',h')\in \llbrace\mathsf{Q}\rrbrace$. 
        By using rule $\illabel{SeqEr}{\circled{i}}$ we conclude $\vdash_{\mathsf{ISL+}}^{\circled{i}}\il[blue][red]{\mathsf{P}}{\mathsf{r_1;r_2}}{\mathsf{Q}}$.
        \item $(\ok:(s'',h''))\in \semantics{r_1}{ISL}{\circled{i}}(\oker:(s,h))$ and $(\oker:(s',h'))\in \semantics{r_2}{ISL}{\circled{i}} (\ok:(s'',h''))$. 
        By inductive hypothesis there exists an assertion $\mathsf{S}$ such that $\vdash_{\mathsf{ISL+}}^{\circled{i}}\il[blue][mygreen]{\mathsf{P}}{\mathsf{r_1}}{\mathsf{S}}$ and $(s'',h'') \in \llbrace \mathsf{S}\rrbrace$. 
        Again by inductive hypothesis, there exists an assertion $\mathsf{Q}$ such that  $\vdash_{\mathsf{ISL+}}^{\circled{i}}\il[blue][blue]{\mathsf{S}}{\mathsf{r_2}}{\mathsf{Q}}$ and $(s',h')\in \llbrace \mathsf{Q}\rrbrace$. 
        By using rule $\illabel{Seq}{\circled{i}}$ we conclude $\vdash_{\mathsf{ISL+}}^{\circled{i}}\il[blue][blue]{\mathsf{P}}{\mathsf{r_1;r_2}}{\mathsf{Q}}$.
    \end{itemize}
    \textbf{Case $\mathsf{r\triangleq r_1+r_2}$:}\\
    We suppose $(\oker:(s',h')) \in \semantics{r_1+r_2}{ISL}{\circled{i}} (\oker:(s,h)) = \semantics{r_1}{ISL}{\circled{i}} (\oker:(s,h)) \cup \semantics{r_2}{ISL}{\circled{i}} (\oker:(s,h))$. 
    Then there must exists $\mathsf{i\in \{1,2\}}$ such that $(\oker:(s',h')) \in \semantics{r_i}{ISL}{\circled{i}} (\oker:(s,h))$. 
    We suppose $\mathsf{i=1}$ since the other case is analogous. 
    By inductive hypothesis there exists an assertion $\mathsf{Q}$ such that $\vdash_{\mathsf{ISL+}}^{\circled{i}}\il[blue][blue]{\mathsf{P}}{\mathsf{r_1}}{\mathsf{Q}}$ and $(s',h') \in \llbrace \mathsf{Q}\rrbrace$. 
    To conclude the proof, we show a derivation for $\vdash_{\mathsf{ISL+}}^{\circled{i}}\il[blue][blue]{\mathsf{P}}{\mathsf{r_1+r_2}}{\mathsf{Q}}$:
    \[
        \infer[\illabel{Choice}{\circled{i}}]
            {\vdash_{\mathsf{ISL+}}^{\circled{i}}\il[blue][blue]{\mathsf{P}}{\mathsf{r_1+r_2}}{\mathsf{Q}}}
            {
                \infer{\vdash_{\mathsf{ISL+}}^{\circled{i}}\il[blue][blue]{\mathsf{P}}{\mathsf{r_1}}{\mathsf{Q}}}{\text{by induction}} 
                & 
                \infer[\illabel{Empty}{\circled{i}}]{\vdash_{\mathsf{ISL+}}^{\circled{i}}\il[blue][blue]{\mathsf{P}}{\mathsf{r_2}}{\mathsf{\false}}}{} 
            }
    \]
    \textbf{Case $\mathsf{r\triangleq r_1^\ast}$:}\\
    We suppose $(\oker:(s',h')) \in \semantics{r_1^\ast}{ISL}{\circled{i}} (\oker:(s,h)) = \bigcup_{\mathsf{i}\in \mathbb{N}} \semantics{r_1^i}{ISL}{\circled{i}} (\oker:(s,h))$. 
    Then there must exists an $\mathsf{m}\geq 0$ such that $\semantics{r_1^m}{ISL}{\circled{i}} (\oker:(s,h))$, i.e. there exists a sequence of states $\{\sigma_\mathsf{i}\}_{0 \leq \mathsf{i} \leq \mathsf{m}}$ where $\sigma_0=\oker:(s,h)$, $\sigma_m=\oker:(s',h')$ and $\sigma_{\mathsf{i+1}}\in \semantics{r_1}{ISL}{\circled{i}} \sigma_\mathsf{i}$. 
    By inductive hypothesis, fixed $\mathsf{P_0=P}$, there exists a sequence of assertions $\{\mathsf{P_i}\}_{0 \leq \mathsf{i} \leq \mathsf{m}}$ such that $\sigma_{\mathsf{i}}\in \llbrace\mathsf{P_{i}}\rrbrace$ and $\vdash_{\mathsf{ISL+}}^{\circled{i}}\il[blue][blue]{\mathsf{P_i}}{\mathsf{r_1}}{\mathsf{P_{i+1}}}$ is provable for all $0\leq \mathsf{i} \leq \mathsf{m}$. 
    We take $\mathsf{Q=P_m}$ and we show a derivation for $\vdash_{\mathsf{ISL+}}^{\circled{i}}\il[blue][blue]{\mathsf{P}}{\mathsf{r_1^\ast}}{\mathsf{Q}}$:
    \[\resizebox{.99\linewidth}{!}{
        \infer[\illabel{Iterate}{\circled{i}}]
            {\vdash_{\mathsf{ISL+}}^{\circled{i}}\il[blue][blue]{\mathsf{P_0}}{\mathsf{r_1^\ast}}{\epsilon:\mathsf{P_m}}}
            {\infer[\illabel{Seq}{\circled{i}}]
                {\vdash_{\mathsf{ISL+}}^{\circled{i}}\il[blue][blue]{\mathsf{P_0}}{\mathsf{r_1^\ast; r_1}}{\epsilon:\mathsf{P_m}}}
                {\infer[\illabel{Iterate}{\circled{i}}]
                    {\vdash_{\mathsf{ISL+}}^{\circled{i}}\il[blue][blue]{\mathsf{P_0}}{\mathsf{r_1^\ast}}{\epsilon:\mathsf{P_{m-1}}}}
                    {\infer[\illabel{Seq}{\circled{i}}]
                        {\vdash_{\mathsf{ISL+}}^{\circled{i}}\il[blue][blue]{\mathsf{P_0}}{\mathsf{r_1^\ast; r_1}}{\epsilon:\mathsf{P_{m-1}}}}
                        {
                            \infer[\illabel{Iterate}{\circled{i}}]
                                {\vdash_{\mathsf{ISL+}}^{\circled{i}}\il[blue][blue]{\mathsf{P_0}}{\mathsf{r_1^\ast}}{\epsilon:\mathsf{P_{m-2}}}}
                                {\infer
                                    {\vdots}
                                    {
                                \infer[\illabel{Iterate\text{-}zero}{\circled{i}}]
                                    {\vdash_{\mathsf{ISL+}}^{\circled{i}}\il[blue][blue]{\mathsf{P_0}}{\mathsf{r_1^\ast}}{\epsilon:\mathsf{P_0}}}{}}
                                }
                            &
                            \infer{\vdash_{\mathsf{ISL+}}^{\circled{i}}\il[blue][blue]{\mathsf{P_{m-2}}}{\mathsf{r_1}}{\epsilon:\mathsf{P_{m-1}}}}{\text{by induction}}
                        }
                    }
                &
                \infer{\vdash_{\mathsf{ISL+}}^{\circled{i}}\il[blue][blue]{\mathsf{P_{m-1}}}{\mathsf{r_1}}{\epsilon:\mathsf{P_m}}}{\text{by induction}}}
            }
    }\]
\end{proof}

\clearpage
\section{Detailed Examples}\label{sec:additional_material}
This appendix contains the full derivations of sample triples discussed in the paper.
They witness the increased expressiveness of SL+ and ISL+ w.r.t. SL and ISL, respectively.
For the sake of readability and to avoid including uninformative intermediate steps, we rely on two facts in the derivations.

First, according to the new frame rule, we only add memories that constrain logical variables. 
Nonetheless, if, for example, we need to use a frame $\mathsf{x'\mapsto \_}$ on a memory where $\mathsf{x\doteq x'}$, then we directly use the assertion $\mathsf{x\mapsto\_}$ on the program variable.
This prevents the need to apply the frame rule followed by the consequence rule to substitute an equivalent assertion: $\mathsf{x\doteq x' \ast x'\mapsto\_ \equiv x\doteq x' \ast x\mapsto\_}$.

Second, when we apply the exists rule, we rely on the fact that $\mathsf{\exists x'. x=x’ \ast x\mapsto \_}$ is equivalent to $\mathsf{x\mapsto\_}$, when $\mathsf{x'}$ is a logical variable. Again, this prevents the need to apply the exists rule followed by the consequence rule to substitute an equivalent assertion.

\begin{figure*}[t]
    \centering
    \[\resizebox{.99\linewidth}{!}{
        \infer[\hllabel{Choice}{\circled{1}}]
            {\dagger}
            {\infer[\hllabel{Cons}{\circled{1}}]
                {\vdash_\mathsf{SL+}^{\circled{1}}\hl{\mathsf{\emp_{\progvar}}}{\mathsf{(\false?; x:=[y])}}{\mathsf{\emp_{\progvar}}}}
                {
                \infer[\hllabel{Seq}{\circled{1}}]
                    {\vdash_\mathsf{SL+}^{\circled{1}}\hl{\mathsf{\emp_{\progvar}}}{\mathsf{(\false?; x:=[y])}}{\mathsf{\false}}}
                    {\infer[\hllabel{Assume}{\circled{1}}]
                        {\vdash_\mathsf{SL+}^{\circled{1}}\hl{\mathsf{\emp_{\progvar}}}{\mathsf{\false?}}{\mathsf{\false}}}
                        {}
                         & 
                     \infer[\hllabel{Frame}{\circled{1}}]
                        {\vdash_\mathsf{SL+}^{\circled{1}}\hl{\mathsf{\false}}{\mathsf{x:=[y]}}{\mathsf{\false}}}
                        {\infer[\hllabel{Load}{\circled{1}}]
                            {\vdash_\mathsf{SL+}^{\circled{1}}\hl{\mathsf{\emp_{\progvar}\ast y\mapsto v}}{\mathsf{x:=[y]}}{\mathsf{\emp_{\progvar\setminus \{x\}}\ast y\mapsto v \wedge x=v}}}
                            {}
                        }
                    }
                }
            }
    }\]
    \[
        \infer[\mathsf{\hllabel{Choice}{\circled{1}}}]
            {\ddagger}
            {\infer[\hllabel{Seq}{\circled{1}}]
                    {\vdash_\mathsf{SL+}^{\circled{1}}\hl{\mathsf{\emp_{\progvar}}}{\mathsf{(\true?;\true?)}}{\mathsf{\emp_{\progvar}}}}
                    {\infer[\hllabel{Assume}{\circled{1}}]
                        {\vdash_\mathsf{SL+}^{\circled{1}}\hl{\mathsf{\emp_{\progvar}}}{\mathsf{\true?}}{\mathsf{\emp_{\progvar}}} }
                        {}
                    & \infer[\hllabel{Assume}{\circled{1}}]
                        {\vdash_\mathsf{SL+}^{\circled{1}} \hl{\mathsf{\emp_{\progvar}}}{\mathsf{\true?}}{\mathsf{\emp_{\progvar}}}}
                        {}
                    }
            }
    \] 
    \[\resizebox{.99\linewidth}{!}{
        \infer[\hllabel{Choice}{\circled{1}}]
            {\infer[\hllabel{Exists}{\circled{1}}]
                {\vdash_\mathsf{SL+}^{\circled{1}}\hl{\mathsf{x\mapsto\_}}{\mathsf{(\false?; x:=[y])+(\true?;\true?)}}{\mathsf{x\mapsto\_}}}
                {
                    \infer[\hllabel{Frame}{\circled{1}}]
                        {\vdash_\mathsf{SL+}^{\circled{1}}\hl{\mathsf{\emp_{\progvar}\ast x\mapsto\_}}{\mathsf{(\false?; x:=[y])+(\true?;\true?)}}{\mathsf{\emp_{\progvar}\ast x\mapsto\_}}}
                        {\vdash_\mathsf{SL+}^{\circled{1}}\hl{\mathsf{\emp_{\progvar}}}{\mathsf{(\false?; x:=[y])+(\true?;\true?)}}{\mathsf{\emp_{\progvar}}}}
                }
            }
            {\dagger & \ddagger}
    }\]
     \caption{A sample derivation in SL+ showing that it is more expressive than SL: the triple $\hl{\mathsf{x\mapsto \_}}{\mathsf{r}}{\mathsf{x\mapsto \_}}$ is not derivable in SL.}
    \Description{A sample derivation in SL+ showing that it is more expressive than SL}
    \label{fig:moreslplus}
\end{figure*}

\begin{example}[SL+ more expressive than SL] \label{eg:sl_derivation}
    Let us take $\mathsf{r\triangleq (\false?; x:=[y]) + (\true?;\skipexp)}$. As shown in Section \ref{sec:introduction} in SL we can derive the triple $\hl{\mathsf{\emp}}{\mathsf{r}}{\mathsf{\emp}}$ but we cannot extend such triple with a frame $\mathsf{x\mapsto \_}$ because $\mathsf{x\in mod(r)}$ and we cannot apply the frame rule, despite the assignment $\mathsf{x:=[y]}$ is dead code. In Figure~\ref{fig:moreslplus}, we show a derivation in SL+ for the triple $\hl{\mathsf{x\mapsto \_}}{\mathsf{r}}{\mathsf{x\mapsto \_}}$. Note that, according to Section \ref{sec:background} we have that $\mathsf{\skipexp \triangleq true?}$.
\end{example}

\begin{figure*}[t]
    \centering
\[
    \infer[\illabel{Seq}{\circled{2}}]
        {(\dagger)}
        {
            \infer[\illabel{Cons}{\circled{2}}]
                {\vdash_\mathsf{ISL+}^{\circled{2}}\il[blue][mygreen]{\mathsf{\emp_{\mathbb{X}_\mathsf{p}}\ast y\mapsto v}}{\mathsf{x:=alloc()}}{\mathsf{\emp_{\mathbb{X}_\mathsf{p}}\ast x\mapsto \_ \ast y\mapsto v}}}
                {\infer[\illabel{Exists}{\circled{2}}]
                    {\vdash_\mathsf{ISL+}^{\circled{2}}\il[blue][mygreen]{\mathsf{\emp_{\mathbb{X}_\mathsf{p}}\ast y\mapsto v}}{\mathsf{x:=alloc()}}{\mathsf{\emp_{\mathbb{X}_\mathsf{p}\setminus \{x\}} \ast x\mapsto \_ \ast y\mapsto v}}}
                    {\infer[\illabel{Frame}{\circled{2}}]
                        {\vdash_\mathsf{ISL+}^{\circled{2}}\il[blue][mygreen]{\mathsf{\emp_{\mathbb{X}_\mathsf{p}}\ast y\mapsto v}}{\mathsf{x:=alloc()}}{\mathsf{\emp_{\mathbb{X}_\mathsf{p}\setminus \{x\}} \ast x\mapsto z'\ast y\mapsto v}}}
                        {\infer[\illabel{Alloc}{\circled{2}}]
                            {\vdash_\mathsf{ISL+}^{\circled{2}}\il[blue][mygreen]{\mathsf{\emp_{\mathbb{X}_\mathsf{p}}}}{\mathsf{x:=alloc()}}{\mathsf{\emp_{\mathbb{X}_\mathsf{p}\setminus \{x\}} \ast x\mapsto z'}}}
                            {}
                        }
                    }
                }
        }
\]

\[
    \infer[\illabel{Seq}{\circled{2}}]
        {(\ddagger)}
        {
            \infer[\illabel{Exists}{\circled{2}}]
                {\vdash_\mathsf{ISL+}^{\circled{2}}\il[blue][mygreen]{\mathsf{\emp_{\mathbb{X}_\mathsf{p}} \ast y\mapsto v \ast x\mapsto \_}}{\mathsf{free(y)}}{\mathsf{\emp_{\mathbb{X}_\mathsf{p}}}\ast {\mathsf{y\not\mapsto}} \ast \mathsf{ x\mapsto \_}}}
                {\infer[\illabel{Frame}{\circled{2}}]
                    {\vdash_\mathsf{ISL+}^{\circled{2}}\il[blue][mygreen]{\mathsf{\emp_{\mathbb{X}_\mathsf{p}} \ast y\mapsto z' \ast z'\doteq v \ast x\mapsto\_}}{\mathsf{free(y)}}{\mathsf{\emp_{\mathbb{X}_\mathsf{p}}}\ast {\mathsf{y\not\mapsto}} \ast \mathsf{z'\doteq v \ast x\mapsto\_}}}
                    {\infer[\illabel{Free2}{\circled{2}}]
                        {\vdash_\mathsf{ISL+}^{\circled{2}}\il[blue][mygreen]{\mathsf{\emp_{\mathbb{X}_\mathsf{p}} \ast y\mapsto z'}}{\mathsf{free(y)}}{\mathsf{\emp_{\mathbb{X}_\mathsf{p}}\ast y\not\mapsto}}}
                        {}
                    }
                }
        }
\]

\[
    \infer[\illabel{Seq}{\circled{2}}]
        {(\star)}
        {
            \infer[\illabel{Frame}{\circled{2}}]
                {\vdash_\mathsf{ISL+}^{\circled{2}}\il[blue][red]{\mathsf{\emp_{\mathbb{X}_\mathsf{p}}}\ast {\mathsf{y\not\mapsto}} \ast \mathsf{x\mapsto \_}}{\mathsf{x:=[y]}}{\mathsf{\emp_{\mathbb{X}_\mathsf{p}}}\ast {\mathsf{y\not\mapsto}} \ast \mathsf{x\mapsto \_}}}
                {\infer[\illabel{LoadEr3}{\circled{2}}]
                    {\vdash_\mathsf{ISL+}^{\circled{2}}\il[blue][red]{\mathsf{\emp_{\mathbb{X}_\mathsf{p}}\ast y\not\mapsto}}{\mathsf{x:=[y]}}{\mathsf{\emp_{\mathbb{X}_\mathsf{p}}\ast y\not\mapsto}}}
                    {}
                }
        }
\]

\[
    \infer[\illabel{Exists}{\circled{2}}]
        {\vdash_\mathsf{ISL+}^{\circled{2}}\il[blue][red]{\mathsf{y\mapsto v}}{\mathsf{r}}{{\mathsf{y\not\mapsto}}\ast \mathsf{ x\mapsto\_}}}
        {\infer[\illabel{Seq}{\circled{2}}]
            {\vdash_\mathsf{ISL+}^{\circled{2}} \il[blue][red]{\mathsf{\emp_{\mathbb{X}_\mathsf{p}}\ast y\mapsto v}}{\mathsf{r}}{\mathsf{\emp_{\mathbb{X}_\mathsf{p}}}\ast {\mathsf{y\not\mapsto}} \ast \mathsf{x\mapsto\_}}}
            {(\dagger)&(\ddagger)&(\star)}
        }
\]
     \caption{A sample derivation in ISL+ showing that it is more expressive than ISL: the triple $\il[blue][red]{\mathsf{y\mapsto v}}{\mathsf{r}}{{\mathsf{y\not\mapsto}}\ast \mathsf{x\mapsto \_}}$ is not derivable in ISL.}
    \Description{A sample derivation in ISL+ showing that it is more expressive than ISL}
    \label{fig:moreislplus}
\end{figure*}

\begin{example}[ISL+ more expressive than ISL] \label{eg:isl_derivation}
    Let us take $\mathsf{r \triangleq x:=alloc();free(y);x:=[y]}$. 
    As shown in Section \ref{sec:introduction} in ISL the triple $\il[blue][red]{\mathsf{y\mapsto v}}{\mathsf{r}}{{\mathsf{y\not\mapsto}}\ast \mathsf{x\mapsto \_}}$ cannot be derived. This is because when handling the command $\mathsf{x:=[y]}$ after $\mathsf{y}$ has been deallocated, we end up in an error, without any actual modification to the variable $\mathsf{x}$. Despite that, we cannot use the frame rule to frame out the assertion on $\mathsf{x\mapsto\_}$ due to the side condition $\mathsf{mod(r)\cap fv(R)= \emptyset}$, which prevents the frame from mentioning the variable $\mathsf{x}$, leading to the impossibility of apply the local axiom for $\mathsf{x:=[y]}$. A possible derivation in ISL+ for such a triple is shown in Figure~\ref{fig:moreislplus}.
\end{example}

\end{document}